\documentclass[oneside]{book}

\usepackage{amssymb,amsmath,amsfonts,dsfont}
\usepackage{graphicx}
\usepackage[figuresright]{rotating}
\usepackage{xifthen,amscd,amsthm}

\usepackage[utf8]{inputenc}

\usepackage{hyperref}
\hypersetup{
    bookmarks=true,         % show bookmarks bar?
    unicode=false,          % non-Latin characters in Acrobat?Äôs bookmarks
    pdftoolbar=true,        % show Acrobat?Äôs toolbar?
    pdfmenubar=true,        % show Acrobat?Äôs menu?
    pdffitwindow=false,     % window fit to page when opened
    pdfstartview={FitH},    % fits the width of the page to the window
    pdfauthor={Reinhard Heckel},     % author
    pdfsubject={Subject},   % subject of the document
    pdfcreator={Reinhard Heckel},   % creator of the document
    pdfproducer={Producer}, %!TEX encoding = UTF-8 Unicode: producer of the document
    pdfnewwindow=true,      % links in new window
    colorlinks=true,       % false: boxed links; true: colored links
    linkcolor=red,          % color of internal links
    citecolor=green,        % color of links to bibliography
    filecolor=magenta,      % color of file links
    urlcolor=cyan           % color of external links
}

\usepackage[
backend=biber,
url=false,isbn=false,doi=false,
hyperref=auto,
style=alphabetic,
natbib=true,
sorting=nty,
firstinits=true,
maxnames=2, maxbibnames=10,
]{biblatex}

\bibliography{./refs.bib}

\newtheorem{lemma}{Lemma}

\newtheorem{theorem}{Theorem}

\usepackage{tikz}
\usepackage{pgfplots}
\usepgfplotslibrary{groupplots} 
\usetikzlibrary{pgfplots.groupplots}

\usepgfplotslibrary{fillbetween}

\usepackage{graphicx}% Include figure files
\usepackage{dcolumn}% Align table columns on decimal point
\usepackage{bm}% bold math
\usepackage{amscd,amsthm}
\usepackage{ifthen}
\usepackage{booktabs}

\newtheorem{conjecture}{Conjecture}

\usetikzlibrary{matrix,arrows,decorations.pathmorphing}

\usepackage{tikz,pgf,pgfplots}
\usetikzlibrary{calc,shadows}
\usetikzlibrary{pgfplots.groupplots}

\usepackage{xr}

\usepackage{amsthm}
\usepackage{amsmath}    % Define \boldsymbol (in amsbsy too) and align
\usepackage{amssymb}    % Define \mathbb (in amsfonts too)
\usepackage{mathrsfs}   % Define \mathscr, a script font
\usepackage{epsfig}     % Define \epsfig
\usepackage{graphicx}
\usepackage{url}

\usepackage{url}
\usepackage{color}
\usepackage{booktabs}

\usepackage{subfigure}

\usepackage{enumitem}

\usepackage{pgfplotstable}
\usepackage{tikz,pgf,pgfplots}
\usetikzlibrary{calc,shadows}
\usetikzlibrary{pgfplots.groupplots}

\usepackage{ifthen}
\usepackage{mathtools}

\renewcommand{\caption}[1]{\singlespacing\hangcaption{#1}\normalspacing}

\usepackage{amsmath}
\usepackage{amsfonts}
\usepackage{dsfont}
\usepackage{mathrsfs}

\usepackage{algorithmic}
\usepackage{algorithm}
\usepackage{multirow}

\newcommand{\x}{x^L}
\newcommand{\y}{y^L}

\newcommand{\matG}{{\bf G}}

\newcommand{\Z}{{\mathds Z}}

\newcommand{\ep}{\epsilon}

\newcommand{\bA}{\mathsf{A}}
\newcommand{\bC}{\mathsf{C}}
\newcommand{\bG}{\mathsf{G}}
\newcommand{\bT}{\mathsf{T}}

\newcommand{\C}{{\mathcal C}}

\newcommand{\E}{{\mathcal E}}

\newcommand{\M}{{\mathcal M}}
\newcommand{\N}{{\mathcal N}}

\newcommand{\T}{{\mathcal T}}

\newcommand{\X}{{\mathcal X}}
\newcommand{\Y}{{\mathcal Y}}
\newcommand{\F}{{\mathbb F}}

\newcommand{\q}[2]{Q_{s_{#1},d_{#2}}}

\newcommand{\eras}{{\varepsilon}}

\newcommand{\peff}{p_{\rm eff}}

\newcommand{\Bernoulli}{{\rm Bernoulli}}
\newcommand{\Geometric}{{\rm Geometric}}

\newcommand{\bc}{{\bf c}}

\newcommand{\bt}{{\bf t}}
\newcommand{\by}{{\bf y}}
\newcommand{\bx}{{\bf x}}

\newcommand{\tpcx}{x^n}

\newtheorem{cor}{Corollary} %[chapter]

\newcommand{\one}{\mathds 1}

\newcounter{constcount}
\newcommand{\eref}[1]{(\ref{#1})}

\newcounter{numcount}
\newcommand{\eqnum}{\stackrel{(\roman{numcount})}{=}\stepcounter{numcount}}

\newcounter{thmcnt}
  \let\Oldsection\section
\renewcommand{\section}{\stepcounter{thmcnt}\Oldsection}

\newenvironment{lemmarep}[1]{\noindent {\bf Lemma #1.}\begin{it}}{\end{it}}

\allowdisplaybreaks

\newcommand{\Pois}{{\rm Poisson}}
\newcommand{\BSC}{{\rm BSC}}
\newcommand{\BEC}{{\rm BEC}}

\newcommand{\aln}[1]{\begin{align*}#1\end{align*}}

\newcommand{\al}[1]{\begin{align}#1\end{align}}

\def\Item$#1${\item $\displaystyle#1$
   \hfill\refstepcounter{equation}(\theequation)}

\newcommand{\bea}{\begin{eqnarray}}
\newcommand{\eea}{\end{eqnarray}}
\newcommand{\beas}{\begin{eqnarray*}}
\newcommand{\eeas}{\end{eqnarray*}}

\newcommand\ML{ML} % length of the sequences

\newcommand\NL{NL}

\newcommand\Tex{}
\newcommand\PR[2][\Tex]{
\ifthenelse{\equal{#1}{}}{{\mathrm{Pr}}\left(#2\right)}{\ensuremath{{\mathrm{Pr}}_{#1}\left[ #2\right]}}}

\newcommand\Ex{\mathbb E} 

\newcommand\EX[2][\Tex]{
\ifthenelse{\equal{#1}{}}{{\mathbb E}\left[#2\right]}{\ensuremath{{\mathbb E}_{#1}\left[ #2\right]}}}

\newcommand\Var[2][\Tex]{
\ifthenelse{\equal{#1}{}}{{\mathrm{Var}}\left[#2\right]}{\ensuremath{{\mathrm{Var}}_{#1}\left[ #2\right]}}}
\newcommand\defeq{\coloneqq}

\renewcommand\M{M} % number of molecules
\renewcommand\N{N} % number of draws
\newcommand\len{L} % length of the sequences
\renewcommand\q{q} % length of the sequences

\newcommand\vf{\mathbf{f}} % for the frequency vector; bf since vectors are boldface throughout, and F is often used as a distribution

\newcommand\cNM{\lambda} % N = cM

\newcommand{\type}{\bt}

\newcommand\ind[1]{\mathds{1}\left\{#1\right\}}

\newcommand\norm[2][\Tnorm]{\ensuremath{{\left\|#2\right\|}_{#1}}}

\newcommand\Emolseen{1-q_0}

\newcommand\Cbsc{C_{\text{BSC}}} % capacity of BSC channel

\newcommand\Rbsc{R_{\text{BSC}}} 
\newcommand\Rind{R_{\text{index}}}

\include{channel_fig}

\title{Information-Theoretic Foundations of DNA Data Storage}

\begin{document}

\begin{center}

{
\Huge Information-Theoretic \\[0.3em]
 Foundations of DNA Data  
 Storage
}

\vspace*{.2in}

{\large{
\begin{tabular}{cccc}
Ilan Shomorony$^\star$ and Reinhard Heckel$^{\ast}$
\end{tabular}
}}

\vspace*{.05in}

\begin{tabular}{c}
$^\star$University of Illinois at Urbana-Champaign, \href{mailto:ilans@illinois.edu}{ilans@illinois.edu} \\
$^\ast$ Technical University of Munich and Rice University, 
\href{mailto:reinhard.heckel@tum.de}{reinhard.heckel@tum.de}
\end{tabular}

\vspace*{.1in}

\textbf{Abstract}
\end{center}

Due to its longevity and enormous information density, DNA is an attractive medium for archival data storage. 
Natural DNA more than 700.000 years old has been recovered, and about 5 grams of DNA can in principle hold a Zetabyte of digital information, orders of magnitude more than what is achieved on conventional storage media. 
Thanks to rapid technological advances, 
% in technologies for writing (synthesis) and reading (sequencing), 
DNA storage is becoming practically feasible, as demonstrated by a number of experimental storage systems, making it a promising solution for our society's increasing need of data storage.

While in living things, DNA molecules can consist of millions of nucleotides, due to technological constraints, in practice, data is stored on many short DNA molecules, which are preserved in a DNA pool and cannot be spatially ordered. Moreover, imperfections in sequencing, synthesis, and handling, as well as DNA decay during storage, introduce random noise into the system, making the task of reliably storing and retrieving information in DNA challenging. 

This unique setup raises a natural information-theoretic question: how much information can be reliably stored on and reconstructed from millions of short noisy sequences? The goal of this monograph is to address this question by discussing the fundamental limits of storing information on DNA.
% short, noisy and shuffled sequences.
Motivated by current technological constraints on DNA synthesis and sequencing, we propose a  probabilistic channel model that captures  three key distinctive aspects of the DNA storage systems: (1) the data is written onto many short DNA molecules that are stored in an unordered fashion; (2) the molecules are corrupted by noise and (3) the data is read by randomly sampling from the DNA pool. 
Our goal is to investigate the impact of each of these key aspects on the capacity of the DNA storage system.
Rather than focusing on coding-theoretic considerations and computationally efficient encoding and decoding, we aim to build an information-theoretic foundation for the analysis of these channels, developing tools for achievability and converse arguments.

\vspace{0.5cm}
%{\bf Citation:}
This is a preprint of the following publication: 
Ilan Shomorony and Reinhard Heckel (2022),``Information-Theoretic Foundations of DNA Data Storage'', Foundations and Trends in Communications and Information Theory: Vol. 19, No. 1, pp 1-106. DOI: 10.1561/0100000117.

\newpage

%\makeabstracttitle

\tableofcontents

\chapter{Introduction}
\label{ch:intro}

% \begin{itemize}
%     \item Motivation for storing information in DNA
%     \item History of DNA storage, etc
% \end{itemize}

%The rate at which conventional storage media such as hard drives, magnetic tapes, and memory chips improve has 

%Recent years have seen an explosion in the amount of data that needs to stored, posing a significant challenge to data centers 

%Conventional storage technologies such as hard drives, magnetic tapes, and memory chips are under pressure to keep up with the ever-increasing amount of data. 
%DNA---a molecule that carries the genetic instructions of living organisms---might be the answer to the challenge or archival data storage, due to its enormous information density and its longevity. 
%V2:
%As we move rapidly into the era of big data, 
%With the ever increasing amount of data in the world,

In recent years, the number of applications that require or are enabled by digital data storage and the amount of data generated by a variety of technologies have increased dramatically. 
%The resulting demand for storage poses significant challenges to current data centers and storage technologies.
% The amount of data that is generated and needs to be stored, as well as the amount of applications which require data-storage solutions is exploding. Both the amount and particular, applications specific storage demands pose significant challenges due to limitations of current storage technologies.  
%Recent years have seen an explosion in the demand for data storage, posing significant challenges to current data centers and storage techniques.
This has spurred significant interest in new storage technologies beyond hard drives, magnetic tapes, and memory chips. 
In this context, DNA---the molecule that carries the genetic instructions of all living organisms---emerged as a promising storage medium. 
DNA has two key advantages over conventional digital storage technologies: extreme longevity and information density.
This
% , illustrated Figure~\ref{fig:longevity_density}, 
makes DNA an interesting storage medium, particularly for archival storage.

Data on DNA can last very long, if stored appropriately, as nature itself proves. As demonstrated by recently sequenced DNA extracted from a mammoth tooth found in the Siberian permafrost~\citep{Valk_2021}, the information in a DNA molecule can be preserved for more than a million years. In contrast, information on memory chips lasts no more than a few years, and data on hard drives and magnetic tapes lasts no more than a few decades. 
While conventional storage media could be redesigned to preserve data longer, the longevity of DNA is currently unmatched, as illustrated in Figure~\ref{fig:longevity_density}(a).
% DNA has a record of being able to preserve data for exceptionally long periods of time.

\begin{figure} 
     \centering
        %  \hspace{-2mm}
\subfigure[\label{fig:longevity}]{
      \includegraphics[width=0.45\linewidth]{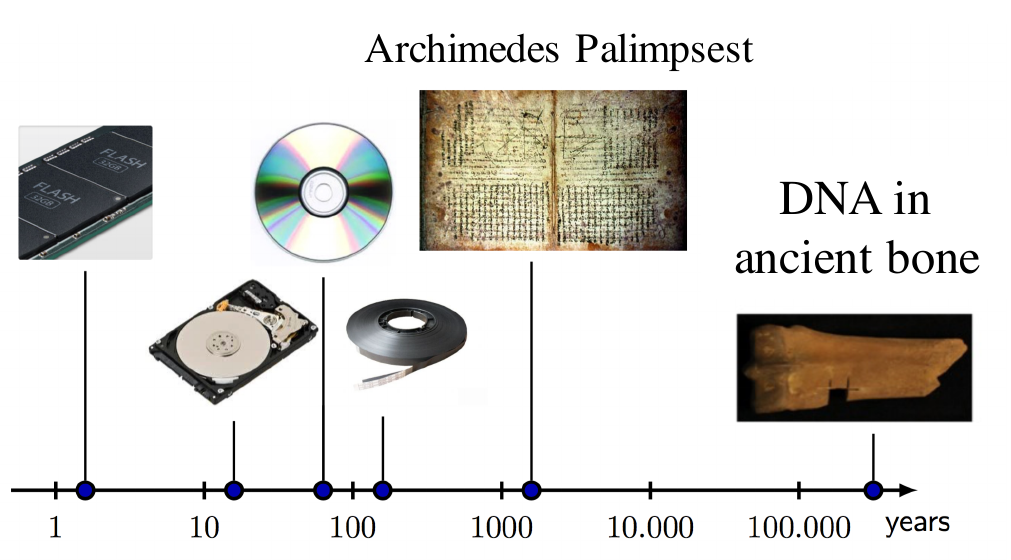}}
    % \\  
    %\hspace{0mm}
\subfigure[\label{fig:density}]{
      \includegraphics[width=0.45\linewidth]{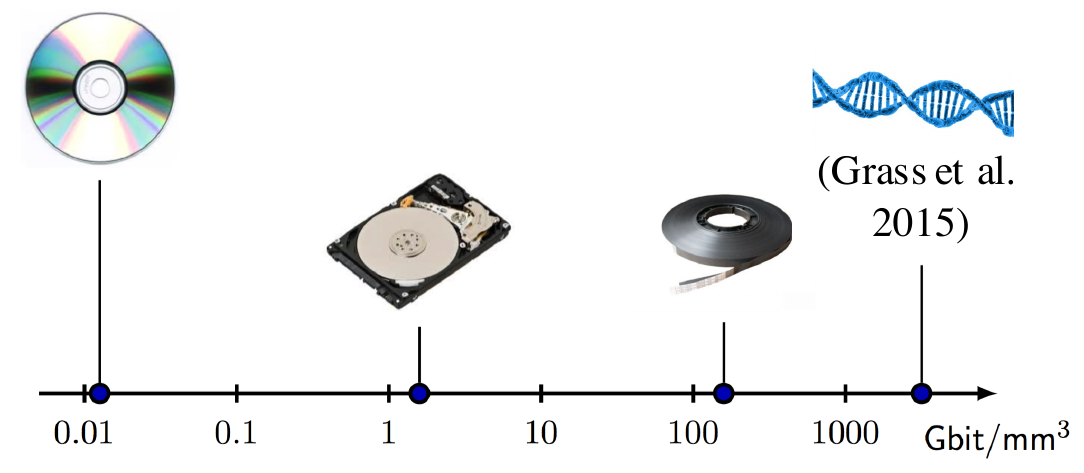}}
     \caption{(a) Longevity of different data storage media. Archimedes Palimpsest containing ``The Methods of Mechanical Theorems'' survived more than 1000 years. Recently, 700,000 year old DNA from ancient horse bones has been successfully sequenced \citep{horse}. (b) Information density of different storage media. The proof-of-concept DNA-based storage system of \citep{grass_robust_2015} achieved an information density that is over one order of magnitude higher that of magnetic tape.
     \label{fig:longevity_density}}
\end{figure}

The information density of DNA is also extremely large. 
Just 5 grams of DNA contain about $4 \cdot 10^{21}$ nucleotides, which in principle could hold $8\cdot 10^{21}$ bits, or one zettabyte. 
In a practical system, the redundancy redundancy required for error-correction coding required to build a reliable system reduces these numbers, 
% As we will see in this monograph, the capacity of a DNA storage channel is a bit larger than those 2 bits/nucleotide
% \iscomment{Why larger? Should it be smaller?}, thus we would need a few extra grams of DNA to reliably store a Zettabyte, 
% but these figures illustrate the extraordinary information density that are possible in DNA. 
but we can achieve information densities orders of magnitude larger than the highest information densities achieved on hard drives and tapes, as shown in Figure~\ref{fig:longevity_density}(b).

%\section{Towards real DNA storage systems}
\section{A brief history of DNA data storage}

Computer scientists and engineers have dreamed of harnessing DNA's storage capabilities already in the 60s~\citep{neiman_fundamental_1964,baum_building_1995}, and in recent years DNA data storage, or more broadly, molecular information storage, developed into an active field of research.
In 2012 and 2013 groups lead by Church~\citep{church_next-generation_2012} and Goldman~\citep{goldman_towards_2013} independently stored about a megabyte of data in DNA. 
Later, \citet{grass_robust_2015} demonstrated that millennia-long storage times are possible by protecting the data both physically and information-theoretically, and designed a robust DNA data storage scheme using error-correcting codes. 
\citet{yazdi_rewritable_2015} showed how to selectively access files, and \citet{erlich_dna_2016} demonstrated that a DNA storage system can achieve very high information densities, close to the absolute maximum of two bits per nucleotide. 
In 2018,  \citet{organick_scaling_2017} scaled up these techniques and stored about 200 megabytes of data.
% The idea to store information on DNA is relatively new~\citep{baum_building_1995}, and has recently gained significant attention~\citep{church_next-generation_2012, goldman_towards_2013,grass_robust_2015,bornholt_dna-based_2016,erlich_dna_2016,yazdi_rewritable_2015}, due to advances in synthesis and sequencing. 
% Several works \citep{church_next-generation_2012, goldman_towards_2013,grass_robust_2015,bornholt_dna-based_2016,erlich_dna_2016,yazdi_rewritable_2015} have 
Together, these and other works 
demonstrated that writing, storing, and retrieving data using DNA as a medium
%\footnote{However little mega bytes (MBs) may seem today, the world's first commercial hard drive, introduced by IBM in 1956, was only able to store 5MB on fifty 21 inch disks.} 
 is possible with today's technology, 
and achieves information densities and information lifetimes that are far beyond what state-of-the-art tapes and discs achieve. 
% We refer to \citep{yazdi_survey_2015,erlich_dna_2016} for a survey and an overview of the recent advances in the area.

% \begin{figure}[h]
% \center
% \includegraphics[width=0.9\linewidth]{Figures/temp_table.png}
% \caption{
% \label{fig:table}
% Comparison of DNA-based storage prototypes.
% }
% \end{figure}

DNA is a long molecule made up of four nucleotides (Adenine, Cytosine, Guanine, and Thymine) and, for storage purposes, can be viewed as a string over a four-letter alphabet.
However, 
there are hard technological constraints for writing on DNA and for reading DNA, which need to be considered in the design of a practical DNA-based storage system. % storing information on DNA.  
%consisting of repeating units called nucleotides, each composed of a nucleobase (Adenine (A), Cytosine (C), Guanine (G), Thymine (T)) and a backbone made of sugars.
While in a living cell a DNA molecule may consist of millions of bases 
(the human chromosome 1, for example, is 250 million bases long),
% (that of a human, for example is $3.2 \cdot 10^9$ bases long), 
due to practical technological constraints, it is difficult and inefficient to synthesize long strands of DNA. 
% Thus, data is stored on short DNA molecules which are preserved in a solution.
%and cannot be spatially ordered.
For that reason, 
 all recent works that have demonstrated working DNA storage systems stored information on molecules of no longer than $100$-$200$ nucleotides~\citep{church_next-generation_2012,goldman_towards_2013,grass_robust_2015,yazdi_rewritable_2015,erlich_dna_2016,organick_scaling_2017}.

The process of determining the order of nucleotides in a DNA molecule, or DNA sequencing, suffers from similar length constraints.
State-of-the-art sequencing platforms such as Illumina cannot sequence DNA segments longer than a few hundred nucleotides.
While recently developed, so-called third-generation technologies such as
% While in a living cell DNA can consist of millions of bases, 
Pacific Biosciences and Oxford Nanopore can provide reads that are several thousand bases long, their error rates and reading costs are significantly higher \citep{van2018third,pacbio_errorRate}.
% Moreover, long reads are in general an overkill for storage systems in view of the current synthesizable lengths.
%and more importantly are the DNA molecules rather short due to the constraints on synthesizing DNA in any case. 

Due to those constraints in the length of the DNA molecules that can be synthesized, stored and sequenced, a practical DNA-based storage system consists of many \emph{short} DNA molecules stored in an \emph{unordered fashion} in a solution.

Another technological limitation to DNA-based storage comes from the fact that state-of-the-art sequencing technologies rely on the \emph{shotgun} sequencing paradigm.
This corresponds to (randomly) sampling and reading sequences from the DNA pool. 
Furthermore, sequencing is usually preceded by several cycles of Polymerase Chain Reaction (PCR) amplification. 
In each cycle, the amount of each DNA molecule is amplified by a factor between 1.6 and 1.8, and this factor can be sequence-dependent, thus leading to very different concentration of distinct DNA molecules.
Last but not least,
% Thus, the proportions of the sequences just before sequencing and therefore the probability that a given sequence is read depends on the synthesis method, the PCR steps, and the decay of DNA during storage. Finally, 
the DNA molecules in a DNA-based storage system are subject to errors such as insertions, deletions, and substitutions of nucleotides at the time of synthesis, during the storage period, and during sequencing.
%A discussion of the error sources and probabilities is provided in the next chapter. %~\citet{heckel_channel_2019} and in the next chapter. 

%%%

\section{Overview of existing DNA storage systems}

In the last decade, several research groups have shown that using today's technologies, it is possible to store on the order of megabytes of data reliably.
All systems that stored megabytes of data and demonstrated correct recovery relied on array-based synthesis, where data is stored on a set of many short sequences. 
In Table~\ref{tab:errorprobs}, we list several of these implementations with some of their key parameters in chronological order. 
%Moreover,
% most of these systems relied on error-correcting codes to guarantee that every bit of information could be retrieved correctly.
All systems listed in the table stored a unique index on each sequence to deal with the shuffling character of the channel at decoding, and starting from~\citet{grass_robust_2015}, all systems used outer error-correcting codes to deal with the loss of individual sequences. 
% all systems that successfully retrieved every single bit of information, 
% which is critical for storage applications, also relied on error-correcting codes.
% In order to give an idea about the parameters of existing systems, we list in Table~\ref{tab:errorprobs} a few DNA storage systems or experiments chronologically with some of their key parameters.
\begin{table}
\small
\begin{center}
\begin{tabular}{l*{6}{c}r}
       & length & number & data stored & error  \\
       & of seqs. & of seqs & (in MB) & correction \\
       \hline
\cite{church_next-generation_2012} & 115 & 54,898 & 0.65 & None \\
\cite{goldman_towards_2013} & 117 & 153,335 & 0.75 & Repetition \\
\cite{grass_robust_2015} & 117 & 4,991 & 0.08 & RS \\
%\cite{yazdi_rewritable_2015} & 1000 & 32 & 0.017 \\
\cite{blawat_forward_2016} & 190 & 900,000 & 22 & RS \\
\cite{bornholt_dna-based_2016} & 120 & 45,652 & 0.15 & RS  \\
\cite{erlich_dna_2016}   & 152 & 72,000 &  2.14 & Fountain\\
\cite{organick_scaling_2017} & 150 & $13.4\cdot 10^9$ & 200.2 & RS \\
\cite{chandak2019improved} & 150 & 13,716 & 0.192 & LDPC \\
Netflix Biohackers, 2020 & 105 &  $3.88\cdot 10^9$ & 63.1 & RS \\
\cite{antkowiak_low_2020}  & 60 & 16,383 & 0.1 & RS\\
\end{tabular}
\end{center}
\caption{
\label{tab:errorprobs}
Parameters of a few DNA storage systems using array-based synthesis, listed chronologically.}
\end{table}

There are alternatives to DNA data storage based on array-based synthesis.
% are currently not scalable.
For example, \citet{yazdi_rewritable_2015} used column-based synthesis, where individual sequences are generated one-by-one.
With this approach, the authors successfully stored 0.017 MB on 32 sequences of length 1000 each. 
Another example is the implementation by \citet{Lee_Kalhor_Goela_Bolot_Church_2019}, which stored 18 bytes using enzymatic synthesis. 
\citet{tabatabaei_DNA-punch-cars_2020} avoided synthesis altogether, and stored data in form of nicks at certain positions on the backbone of existing DNA. 
Yet another example is the work by~\citet{yim_robust_direct_2021}, which stored data (about 72 bits) in living cells via CRISPR arrays. 
All these approaches are conceptually interesting alternatives to array-based synthesis, but scaling them to store more than a few bytes of data currently looks challenging.
% but are currently not scalable.

Array-based synthesis generates many copies of each sequence, by building up each sequence on a different spot of the array.  It is also possible to modify array-based synthesis to  grow sequences at one spot of the array so that a predefined fraction of the sequences contain, say, nucleotide A but others contain, say,  nucleotide G at a given position. \citet{anavy_composite_2019} explored this idea and proposed the notion of  \emph{composite DNA letters} to reduce the number of synthesis cycles. 
However, the use of composite letters comes at a higher sequencing cost to guarantee that  composite letters can be identified from the sequenced DNA molecules, and also comes at a higher decoding complexity.

While the proof-of-concept implementations listed in Table~\ref{tab:errorprobs}   demonstrate that DNA storage can be practical, their overall cost is still an obstacle for them to become practically viable.
For example, using the architecture proposed in \citet{erlich_dna_2016},
% Erlich and Zielinski (2017), 
the estimated cost of synthesizing 1GB of data was \$3.27 million \citep[Supplementary Material]{erlich_dna_2016}, and the current cost of storing a Megabyte of DNA is around \$500~\citep{antkowiak_low_2020}. 
However, until DNA storage becomes viable for commercial archival storage applications, there are already applications of DNA storage that are not possible with other storage media. For example, \citet{Koch_Gantenbein_Masania_Stark_Erlich_Grass_2020} demonstrated that data stored in DNA can be embedded into any product made of plastic, providing information about the product inside the product.

% IS: moved some of this to the paragraph before the table. I don't think the table illustrates the role of information theory, since these are all standard coding techniques designed for generic applications
% The table also illustrates the role of information theory and error correcting codes: All systems listed in the table indexed the sequences to deal with the shuffling character of the channel, and starting from row~\cite{grass_robust_2015}, all systems used outer error correcting codes to deal with the loss of individual sequences. 

%Cite survey \citep{yazdi_survey_2015}
%\begin{itemize}
%    \item synthesis, storage, and sequencing technologies
%    \item Overview of proof-of-concept systems 
%\end{itemize}

\section{Information Theory and DNA storage}

DNA-based storage is a fundamentally new way of storing data, due to the way DNA is written, stored, and read. 
The technology is still under development, and details such as error profiles, the exact length of synthesized DNA molecules, and the sequencing throughput are likely to change. 
% and can be 
% with unique noise characteristics, and the involved technologies are still under development.
%An information-theoretic perspective can highlight the more basic ways in which this new paradigm is different from existing data storage systems.
An information-theoretic perspective and an understanding of the fundamental limits of DNA storage will enable a system design based on 
% lead to the development of practical codes based on 
key conceptual insights and tradeoffs.
%An information-theoretic perspective will enable system designs that 

% Furthermore, it will
%   tradeoffs such as using cheaper synthesis at the cost of higher error rates. 

\paragraph{Joint design of physical system and coding schemes: } Conceptual advances in terms of how to optimally code for this new storage paradigm can inform biochemists in their development of new synthesis and sequencing technologies for building DNA storage systems.
For example, is it worth developing very expensive technologies to allow long DNA molecules to be synthesized?
Or is it possible to store data at high densities with very short molecules?
An information-theoretic perspective may provide the foundations to answering these questions, enabling the design of an efficient ``physical layer'' for DNA storage systems.

\paragraph{Emerging technologies beyond DNA: } There are many other interesting media for future storage. %DNA is not the only candidate to be the storage medium of the future.
For example, synthetic polymers are also a potential substrate for data storage \citep{pattabiraman2019reconstruction}, in which case
tandem mass spectrometry could be used instead of sequencing for data retrieval \citep{laure2016coding}.
Futhermore, the idea of storing data in quartz glass \citep{anderson2018glass} also
promises to achieve incredibly high information densities.
% These new technologies should be compared at a fundamental level

As new technologies are proposed and compared, it is important to obtain a basic understanding of their capabilities.
In this context, an information-theoretic perspective may allow these emerging approaches to be compared at a more fundamental level, rather than based on specific prototypes.
Moreover, design principles may be transferable between them, and different technologies may be more suitable to specific applications depending on basic tradeoffs between cost, computational complexity, and reading and writing speeds.
In Section~\ref{sec:newtech} we briefly discuss additional storage capacity problems that are motivated by recent technological advances.

\paragraph{Connections with classical information theory problems: }
The growing interest on DNA data storage has sparked renewed interest in classical problems in information theory.
As several DNA sequencing technologies suffer from insertion and deletion errors, new attention has been given to the capacity of insertion/deletion channels and ``sticky'' channels \citep{cheraghchi2019sharp}.
Moreover, many of the topics explored in this monograph will be connected with permutation channels
\citep{ahlswede1987optimal,benjamin1975coding}
and, in Section~\ref{sec:short}, we discuss a connection between DNA storage in the short-molecule regime and discrete-time Poisson channels \citep{lapidoth2008capacity}.

% \begin{itemize}
    % \item joint design of physical system and the coding schemes
    % \item fundamental tradeoffs
    % \item connections with other channels
    % \item macromolecule storage beyond DNA
    % \item connections with DNA-based computation, etc
% \end{itemize}

\section{Organization of this monograph}

In Chapter~\ref{ch:model}, we formalize and discuss a general class of channels to model DNA storage systems. 
This general model, called the noisy shuffling-sampling channel, will be the main object of our information-theoretic analysis of DNA-based data storage.
% be explored throughout the entire monograph. 
In Chapter~\ref{ch:shuffled}, we start the exploration of the capacity of DNA storage systems by studying a simple noiseless shuffling-sampling channel, where the input sequences are shuffled and sampled before being observed at the channel output.
The capacity expression provides a precise understanding of the storage rate costs that the shuffling and sampling operations incur, and provide intuition on how to develop optimal codes for such channels. 
In the same chapter, we also consider channels that break the sequences at random points and shuffle the resulting pieces.

In Chapter~\ref{ch:noisy}, we study the impact of adding noise to the shuffling-sampling channels. 
Specifically we will study the capacity of channels where the sequences are not only shuffled and randomly sampled, but also contain errors, and we will discuss the impact of different noisy channel models on the overall capacity.

In Chapter~\ref{ch:multi}, we study the multi-draw nature of DNA storage channels. Since multiple copies of each sequence are typically stored in DNA storage systems, they can be sequenced with potentially different noise patterns, which can help in error correction. 
% Particularly in the regime where the noise level in the sequences is high, reliable storage of information may only be possible by combining multiple noisy draws arising from the same input sequences. 
% This setup is challenging even for the case of a single input molecule, which is known as the trace reconstruction problem. 
This is captured in our information-theoretic framework by considering multi-draw channels, where each sequence in the DNA library may yield multiple copies with independent noise patterns at the output.
In this setting, a natural approach is to first cluster the sequences and then solve several \emph{trace reconstruction} problems. 
We will study whether such clustering-based schemes are optimal, and we will provide results on the fundamental limits of DNA storage channels with multi-draws.

DNA storage is a relatively new field and many important questions remain unanswered. 
In Chapter~\ref{ch:discussion}, we end the monograph with a collection of important open problems and a discussion of connections 
% , as well as pointing out connections 
to existing channel models and classical results in information theory.

\chapter{Channel Model}
\label{ch:model}

As a consequence of the unique way in which data is stored in DNA and retrieved from DNA, studying the fundamental limits of DNA-based data storage requires the development of new channel models.

Figure~\ref{fig:fig1} illustrates the key features of a DNA-based storage system.
The data to be stored, described by a sequence of bits, is first encoded into a set of short strings over the alphabet $\{\bA,\bC,\bG,\bT\}$. The strings have length of typically less than 200 nucleotides, because existing sequencing technologies can only synthesize short sequences. 
These strings are synthesized as actual DNA molecules, via a noisy synthesis process. 
The synthesis process typically generates multiple noisy copies of each DNA molecule, as it is currently not possible to synthesize individual molecules efficiently. 
The resulting pool of DNA molecules is either stored directly or first amplified with Polymerase Chain Reaction (PCR) and then stored. 
Prior to being stored, the DNA molecules are often protected in some way, for example they can be encapsulated in silica~\citep{grass_robust_2015} (see~\citet{Organick_Nguyen_McAmis_Chen_Kohll_Ang_Grass_Ceze_Strauss_2021} for an empirical comparison of a few preservation options).  
In this form, the DNA molecules suffer relatively small amounts of deterioration over long periods of time under appropriate environmental conditions. 

At the time of reading, the DNA is usually first amplified with PCR, effectively replicating each of the molecules a large number of times, and then sequenced.
% IS: I don't think people think of nanopore sequencing as not high-throughput sequencing
% Then, either high-throughput shotgun sequencing or nanopore sequencing is applied to the DNA library.
Sequencing can be thought of as randomly sampling from the mixture. Since multiple copies of each molecule are present due to the amplification step (possibly corrupted by different noise patterns), the same molecule can be observed multiple times at the output (or not at all).
% can never be observed at the output. 
The decoder operates on the set of sequenced strings with the goal of recovering the original bit string that encodes the data.

\begin{figure}[t]
\center
\includegraphics[width=0.97\linewidth]{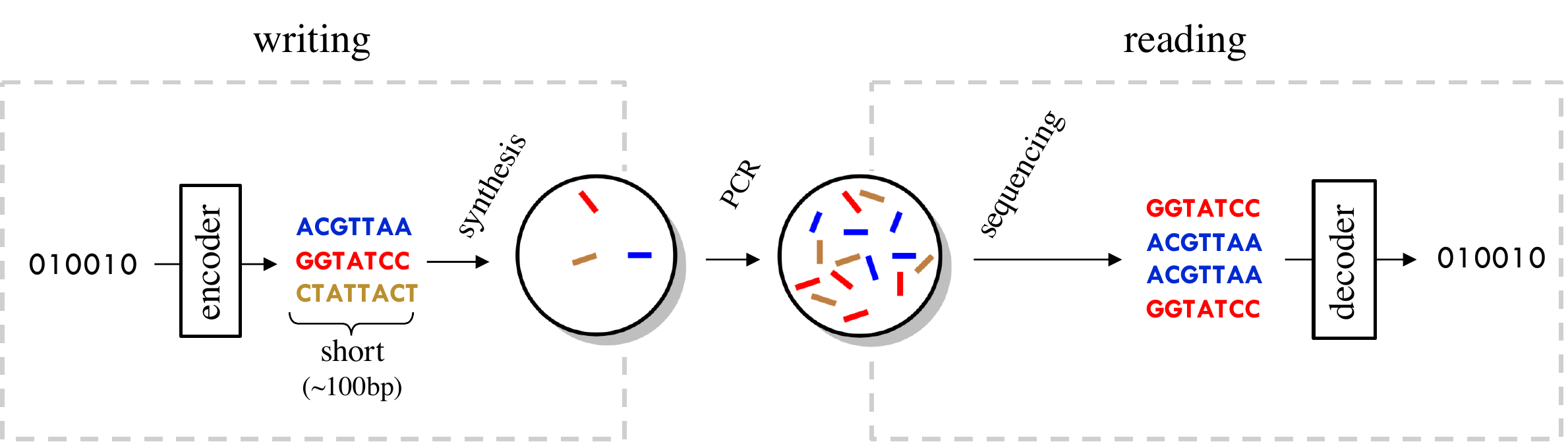}
\caption{
\label{fig:fig1}
Illustration of a DNA-based storage system. 
% \iscomment{Rectangles for writing and reading}
% Channel model for DNA storage systems. 
In the writing phase, a sequence of bits is encoded into a multi-set of short sequences over the alphabet $\{\bA,\bC,\bG,\bT\}$, which are then synthesized into actual DNA molecules and stored.
In the reading phase, PCR is often used to replicate the DNA molecules, followed by sequencing, which randomly samples molecules from the DNA library and reads them.
The decoder is applied to the resulting multi-set of sequences in order to recover the original sequence of bits.
% \rh{PCR in the writing part might be a bit missleading, because it's often applied only after storage.}
% The input to the channel is a multi-set of $\M$ length-$\len$ DNA molecules and the output is a multi-set of $\N$ draws from the pool of DNA molecules that are perturbed by insertions, substitutions, and deletions (marked as lowercase and boldface letters).
%The sampling distribution as well as the process inducing errors in individual molecules account for errors in synthesis, storage, and sequencing.
}
\end{figure}

% \section{A model for DNA-based storage: the noisy shuffling-sampling channel}
\section{The noisy shuffling-sampling channel model}
\label{sec:motivation_shuffling}
% \subsection{Towards DNA storage capacity: The Shuffling Channel}

%While several works have focused on the design of codes tailored to the specific error profiles of DNA synthesis (see Section~\ref{sec:noisyadditional}), storage, and sequencing, attempts at characterizing the information-theoretic limits of the end-to-end system have been much fewer in number.
% Consequently, our understanding of the fundamental capabilities of storage capacity in DNA-based systems is still in its infancy, and the design of optimal coding schemes (from the storage rate or the error probability perspectives) remains an open problem.
%\rh{The part below would fit better into chapter 3, ``coding for shuffling channels''. Then chapter 3 could contain sections ``shuffling channel'', and ``multi shuffling channel''}

In order to study the fundamental limits of DNA-based data storage, we propose a general channel model, which we call the noisy shuffling-sampling channel.
% We propose a general channel model for DNA sthat capture the key aspects of DNA-based storage systems.
% The input to the channel are $M$ strings of length $L$ and 
This channel model captures three key elements of DNA storage:
\begin{itemize}
    \item {\bf Sampling}: Since the reading of a DNA library is done via  sequencing,
    % or nanopore sequencing, % nanopore is also "shotgun"
    strings are randomly sampled from the pool. Moreover, since synthesis technologies generate multiple copies of each sequence and PCR replicates the molecules in the pool, 
    % , and also due to the PCR amplification step, which replicates the molecules in the pool, 
    the same input sequence can be observed multiple times at the output.
    %\rh{Throughout, could shotgun sequencing be confusing? I think it's used typically if a genome is split randomly and then reasembled, could people thing that that also happens in DNA storage? Perhaps we just say sequencing?}
    
    \item {\bf Shuffling}: The ordering of the DNA molecules is lost during storage, effectively causing them to be shuffled.
    % The channel output is independent of the ordering of the $M$ input strings; i.e., the channel acts on the multi-set of $M$ strings. Note that there can be identical strings in the input.
    
    \item {\bf Noise}: DNA molecules are corrupted by noise at the synthesis, storage and sequencing steps.
    The type of error (substitutions, insertions, deletions) varies according to the technologies used.
\end{itemize}
% The channel resulting from the combination of these three elements is illustrated in Figure~\ref{fig:generalchannel}.
Our general channel model for DNA storage is obtained by concatenating a sampling channel (which creates a random number of copies of each string), a shuffling channel (which shuffles a set of sequences uniformly) and a noisy channel (which corrupts each of the strings).
% The result is what we call the \emph{noisy shuffling-sampling channel}.

%Most of the ideas discussed herein can be extended to this more general case.
%\rh{is this model really so different? Can't this just be modeled by the channel at the end?}

% , one before shuffling, one between shuffling and sampling, and one after sampling. 
% It is not difficult to see that the first two such noisy channels can be merged into one.

% We notice that, in a regime where $M$ and $L$ are moderately large, the capacity can be thought of as $2-\frac{\log M}{L}$.
% This expression suggests that 

%%%

\begin{figure}[t]
\center
\includegraphics[width=0.93\linewidth]{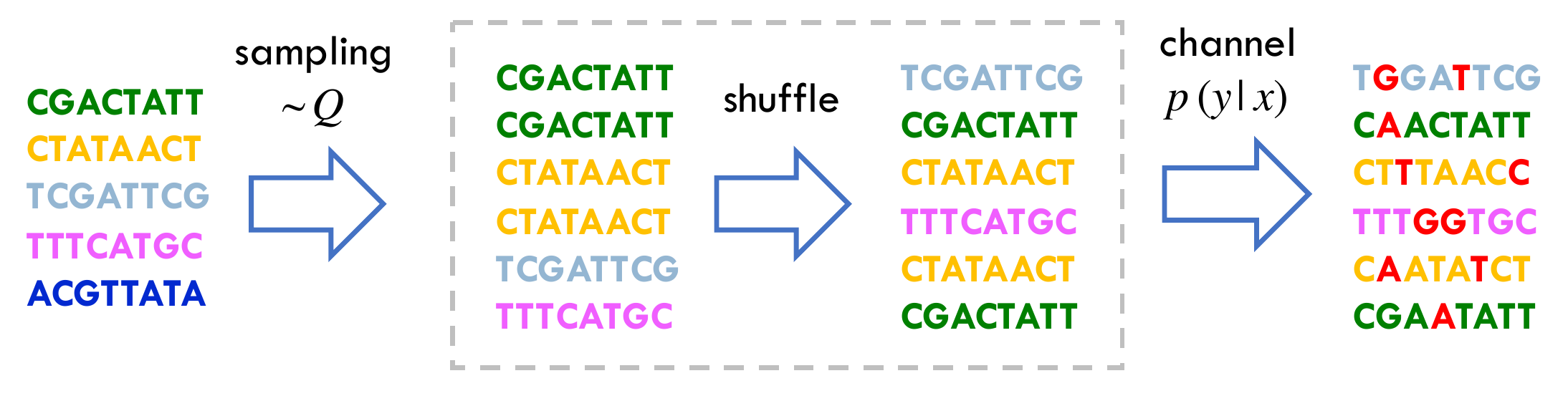}
\caption{
\label{fig:generalchannel}
The general noisy sampling-shuffling channel.
The input are $M$ sequences of length $L$.
The $i$th input sequence is sampled a number $N_i$ of times, where each $N_i$ is distributed according to a distribution $Q$.
The resulting $N$ sequences are shuffled out of order and each one passed through a noisy channel $p(y|x)$. 
}
\end{figure}

% IS: 07-13: Added the formalization of the channel here

% The general noisy shuffling-sampling channel is described as follows.

We formalize the noisy shuffling-sampling channel as follows.
The input to the channel is a list $[\x_1,\ldots,\x_M]$ of $M$ sequences of length $L$ over an alphabet $\Sigma$.
The channel performs three operations:
\begin{itemize}
\item {\bf Sampling:}  
% For a given codeword $[\x_1,\ldots,\x_M] \in \C$, e
Each sequence $\x_i$ is sampled a number $\N_i \sim Q$ of times, for some distribution $Q = (q_0,q_1,\ldots)$, 
where $\q_n = \Pr(N_i = n)$ is the probability that $n$ copies of the sequence $\x_i$ are drawn.
The sampling of different strings is assumed to be independent.
We let $N = \sum_{i=1}^\M N_i$ be the total number of resulting strings, 
and we define $\cNM \defeq \Ex[\N] /\M = \Ex[\N_i]$ to be the sequencing \emph{coverage depth} \citep{LanderWaterman,motahari2013information}.
% The distribution $Q$ models 
% the loss of whole sequences at synthesis and during storage, the multiple replicates of each molecule due to synthesis and PCR amplification, and the nature of sequencing, which chooses samples randomly from a pool of DNA.

% imperfections in synthesis, sequencing, and a loss of whole sequences during storage (see \citep{heckel_channel_2019} for a detailed discussion and how this distribution looks like for specific choices of sequencing and synthesis technologies). 
%

\item {\bf Shuffling:} The $\N$ strings are shuffled uniformly at random. 

% We consider two cases
% i) where the channel is error free, and 
% ii) where the memoryless channel is a binary symmetric channel with error probability $p$ (BSC($p$)).
% \iscomment{I think we don't want to specify this here} 
% Depending on the synthesis and sequencing technologies this may or may not be a good model; for the majority of past DNA storage experiments, substitution errors dominate, but sometimes deletions and insertion errors dominate. We focus on substitution errors in this work
%

% \iscomment{Need to update this part to include the general description of the shuffling vector $S^N$}

\item {\bf Noise:} Each of the $\N$ strings is independently passed through a discrete memoryless channel $p(y|x)$, producing a list of $N$ output sequences 
 $[\y_1,\ldots,\y_N]$. 
Equivalently, the output of the channel can be taken to be the (unordered) multi-set  $\{\y_1,\ldots,\y_N\}$. 
\end{itemize}

Different specifications of the sampling distribution $Q$ and the noisy channel $p(y|x)$ lead to a rich class of information-theoretic channels whose capacity is in general unknown and nontrivial.
As described throughout this monograph, several important questions related to the joint encoding of information across multiple strings arise in the context of this general channel model.

We point out that, since synthesis, storage and sequencing are all error-prone procedures, an even more general channel model would have another noisy channel prior to the sampling step.
For simplicity, we will focus on the case with a single, final, noisy channel as this is already a good modeling assumption.

%%%
\section{Motivation for the channel model and error sources}
\label{sec:error_source}

% \iscomment{maybe this section should be in chapter 1? It's more of a description of the technologies in current implementations}

% \rh{That's an option; another option is to call this section `motivation for the channel model' or merge it with the intro section of this chapter 2.}

The channel model in Section~\ref{sec:motivation_shuffling} is motivated by the technologies used to synthesize, process, store, handle, and sequence DNA, and errors arise in all those steps. In this section, we discuss the technologies and sources of errors at these stages.
% and the resulting channel statistics. 
We refer to~\citet{heckel_channel_2019} and references therein for a more detailed descriptions of the error mechanisms.

\paragraph{Synthesis.} 
Three different array-based synthesis methods are currently used for DNA data storage. Array-based synthesis methods grow several strands simultaneously on an array. One end of the DNA molecule is attached to the array, and the nucleotides are added one by one through the following mechanisms.
The first method is a material deposition or printing-based technology commercialized by the companies Agilent and Twist Biosciences. This method was used for DNA storage by~\citet{organick_scaling_2017,erlich_dna_2016}. The second is electro-chemical, commercialized by CustomArray, and was used for DNA storage by~\citet{grass_robust_2015}. The third is based on light-directed placement of nucleotides~\citep{Singh-Gasson_Green_Yue_Nelson_Blattner_Sussman_Cerrina_1999} and was used for DNA storage by~\citet{antkowiak_low_2020}.

Each of these technologies produces thousands to millions of copies of each sequence. The printing based and electrochemical synthesis technologies produce relatively few errors within the sequences, but the light-directed technology produces many insertion, deletion, and substitution errors within the sequences.
% ; more on this later.
In summary, DNA synthesis produces a pool of DNA sequences, which contain many noisy copies of each sequence.

\paragraph{Processing and storage.}
The DNA pool is either stored as it is, or we may only select a small part of the synthesized DNA pool for storage and dilute it so that the physical redundancy
% ; i.e., the number of copies of the same sequence 
is relatively small. 
As mentioned previously, the DNA might be protected through encapsulation or by other means. 
Either way, during the processing and storage of the DNA, the DNA sequences are prone to chemical decay, which causes substitution errors and 
occasionally results in the full loss of sequences.
% , as well as to substitution errors.

Finally, PCR cycles are often used in a DNA storage system to amplify the DNA. In each PCR cycle, a given sequence is amplified on average by a factor slightly less than two, and through multiple cycles this can lead to a significant imbalance in the sequences in the pool.

\paragraph{Sequencing.} 
The currently most utilized sequencing technology is  commercialized by Illumina. This technology induces relatively few substitution errors within the sequences and even fewer insertion and deletion errors. 
\citet{heckel_channel_2019} estimated the reading error probabilities based on the two-sided reads to be between 0.15\% and 0.4\% and those errors are almost exclusively substitution errors. \citet{Schirmer_DAmore_Ijaz_Hall_Quince_2016} estimated a per-base error probability of 0.2\% for Illumina reads, a very similar number. 

Another popular sequencing technology is nanopore sequencing. Nanopore sequencing has significantly higher error probabilities, and introduces substitution, deletion, and insertion errors. The per-pase error of the Oxford Nanopore Technologies (ONT) MinION sequencer is about 10\%, albeit the exact number depends on the operating mode. The relative amount of deletions, substitutions, and insertions is similar for this technology.

% In either case, the sequencing process generates a set of reads.

%%%%%%%%%%%%%

\section{Channel statistics}
\label{sec:distributions}

As described in Section~\ref{sec:motivation_shuffling}, the channel parameters that need to be specified are the sampling distribution $Q$  and the noisy channel $p(y|x)$.
The sampling and noise distributions 
% The statistics of the channel 
depend on the error sources discussed in Section~\ref{sec:error_source} and are determined by the technologies that are used. 
Here, we discuss the sampling and noise statistics for a few real-world DNA storage systems.

\paragraph{Sampling distribution.} 
% IS: We've said this many times
% An important characteristic of the DNA storage channel is the sampling-shuffling character. 
% Since the data is stored on a set of sequences, the number and distribution of those sequences determine the sampling distribution. % channel statistics.
The sampling distribution is determined by several different factors.
The first one is the number of physical copies of each sequence initially stored in the DNA pool, known as the physical redundancy (PR).
For example, if the physical redundancy is 100, then the DNA pool consists of $100\M$ DNA molecules, where $M$ is the original number of synthesized sequences.
The second is the number of PCR cycles (PCRC) used.
Each PCR cycle on average multiplies each sequence by a factor slightly less than two.
Furthermore, the sequencing read coverage (i.e., the number of sequencing reads divided by $M$) can be adjusted.

Figure~\ref{fig:dist_molecules} depicts the empirical sampling distribution from \citet{heckel_channel_2019} for four different DNA storage experiments from the literature. 
The read coverage of all four experiments is relatively similar (about 300-500). 
However, they have vastly different physical redundancies (PR), i.e., number of physical sequences in the stored pool per original sequence. 
DNA libraries with smaller physical redundancy typically require more PCR cycles prior to sequencing.
% To read the DNA, a dataset with smaller physical redundancy requires more PCR cycles (PCRC). 
However, each PCR cycle increases the imbalance of the sequences in the dataset, because they may lead to some sequences being amplified more than others.
This leads to a loss of sequences and increases the tail of the distribution. Figure~\ref{fig:dist_molecules} illustrates that the lower the physical redundancy, the more long-tailed is the distribution, and more of the sequences are never seen at the output.
More details on the quantification of the sampling distribution can be found in~\citet{heckel_channel_2019,chen_quantifying_2020}.

\definecolor{DarkBlue}{rgb}{0,0,0.7} 
\definecolor{BrickRed}{RGB}{203,65,84}

\begin{figure}[t]
\begin{center}
\begin{tikzpicture}

\draw[draw,fill=DarkBlue!50,DarkBlue!50] (0,2.5) -- (10.7,2.5) -- (0,3.3);
\draw[draw,fill=BrickRed!50,BrickRed!50] (0,3.5) -- (10.7,3.5) -- (10.7,2.7);
\node[white] at (2,2.8) {Physical Redundancy};
\node[white] at (9.45,3.1) {PCR cycles};

\pgfplotsset{every tick label/.append style={font=\scriptsize}}
\begin{groupplot}[
title style={at={(0.5,-0.5)},anchor=north}, 
         group style={group size=4 by 1, %xlabels at=edge bottom, 
         ylabels at=edge left, yticklabels at=edge left,    
         horizontal sep=0.1cm,vertical sep=0.3cm}, xlabel={\small \# reads/seq.}, ylabel={\small frequency},
         width=0.355\textwidth,height=0.34\textwidth,
         y label style={at={(0.3,0.5)}},
]

	\nextgroupplot[title = {\parbox{3.2cm}{\centering (a) \\ \small PR $1.28e7$ \\ PCRC 10 }},yticklabels={,,}] 
	\addplot +[ycomb,DarkBlue!65,mark=none] table[x index=0,y index=1]{./data/erlich_hist.dat};
	\draw (0,0) circle(1.5pt) -- (1.5cm,1.5cm) node[above right]{\scriptsize $0\%$};

	\nextgroupplot[title = {\parbox{3cm}{\centering (b) \\ \small PR 22172 \\  PCRC 22 }},] 
	\addplot +[ycomb,DarkBlue!65,mark=none] table[x index=0,y index=1]{./data/goldmanhist.dat};
	\draw (0,0) circle(1.5pt) -- (1.1cm,1.1cm) node[above right]{\scriptsize $0.0098\%$};
        %\addlegendentry{Empirical};

	\nextgroupplot[title = {\parbox{3cm}{\centering (c) \\ \small PR 3898 \\  PCRC 65  }}, xmax=1500,yticklabels={,,}] 
	\addplot +[ycomb,DarkBlue!65,mark=none] table[x index=0,y index=1]{./data/S1hist.dat}; 
     	\draw (0,0) circle(1.5pt) -- (1cm,1cm) node[above right]{\scriptsize $1.6\%$};

	\nextgroupplot[title = {\parbox{3cm}{\centering (d) \\ \small PR 1.2 \\  PCRC 68 }},
	,xmax=1500]
	%\addplot +[ycomb,brown,mark=none] table[x index=0,y index=1]{./data/GM1hist.dat}; 
	%\draw (0,0) circle(1.5pt) -- (1cm,1cm) node[above right]{$20\%$};
            
	\addplot +[ycomb,DarkBlue!65,mark=none] table[x index=0,y index=1]{./data/GM6hist.dat};
	\draw (0,0) circle(1.5pt) -- (1cm,1cm) node[above right]{\scriptsize $34\%$};

\end{groupplot}          
\end{tikzpicture}
\end{center}
%%%
\caption{
\label{fig:dist_molecules}
The empirical distributions (Figure~5 in~\citet{heckel_channel_2019}) of the number of reads of four different datasets from~\citet{goldman_towards_2013,erlich_dna_2016,grass_robust_2015} per each given sequence that has been synthesized, along with the physical redundancy (PR) and PCR cycles (PCRC). 
The data in (a) and (b) is approximately negative binomial distributed, whereas the data in (c) and (d) have a long tail and a peak at zero. 
The percentages in the figure are the fraction of sequences that are never seen at the output. 
The read coverage (number of reads per original sequence) of all datasets is comparable. 
Likely, the difference in distribution of (a) and (b) to (c) and (d) is due to the significantly more cycles of PCR in (a) and (b), while the difference of (a) and (b) is due to  low physical redundancy and dilution. 
}
\end{figure}

%%%
\paragraph{Errors within sequences.}
Errors within sequences are due to synthesis, sequencing, and decay of the DNA. 
Table~\ref{tab:errorprobs} contains the estimated error probabilities from four DNA storage experiments in the literature. 
The probabilities indicate the error probability per nucleotide, e.g., a 0.56\% substitution error probability means that roughly one out of 200 nucleotides contains a substitution error. 
In all experiments, the data was stored and read without any artificial aging \citep{grass_robust_2015}. 
Also, in all experiments low-error Illumina sequencing technology was used. Therefore, the errors within sequences in those experiments are to a large extent due to the errors in synthesis. It can be seen that for the established printing-based and electrochemical synthesis technologies, the errors within sequences are relatively small, while the errors for light-directed synthesis are very large. 

Note that in the first three datasets, the majority of sequences that contain insertion and deletion errors can be identified easily by their length being longer or shorter  than the original sequence, and have been removed in the respective experiments. This results in a set of DNA sequences which are almost free of deletion and insertion errors. For the dataset in~\citet{antkowiak_low_2020} this wouldn't make sense, as most sequences contain insertion or deletion errors, and too many sequences would be removed.
\begin{table}[]
\small
\begin{center}
\begin{tabular}{l*{6}{c}r}
Experiment       & Subst. & Del. & Insert. & Length & Synthesis \\
\hline
\cite{grass_robust_2015} & 0.56\% & 0.99\% & 0.09\% &  117 & CustomArr.  \\
\cite{erlich_dna_2016}            & 0.36\% & 0.16\% & 0.18\% & 152 & Twist \\
\cite{goldman_towards_2013}           & 0.08\% & 0.54\% & 0.02\% & 117 & Agilent  \\
\cite{antkowiak_low_2020}     & 2.6\% & 6.2\% & 5.7\% &  60 & light-dir.   \\
\end{tabular}
\end{center}
\caption{
\label{tab:errorprobs}
Error probabilities within sequences in four data storage experiments.
}
\end{table}

For more error statistics, see the papers by~\citet{heckel_channel_2019,erlich_dna_2016,sabary_solqc_2021}.

%%%%%%%%%%%%

\section{Definitions and notation}

% Motivated by the discussion in Section~\ref{sec:motivation_shuffling}, we now provide notation and definitions that are used throughout the monograph.

In Section~\ref{sec:motivation_shuffling}, we formally defined the noisy shuffling-channel, which will be the main object of study of this monograph.
In this section we provide additional definitions and notation required to formally analyze the capacity of this channel.

% An $(\M,\len)$ DNA storage code $\C$ 
A code $\C$ for a noisy shuffling-sampling channel
is a set of codewords,  %$\C \subset 2^{\Sigma^{\len}}$, 
each of which is a list $[\x_1,\ldots,\x_M]$ of $M$ strings of length $\len$,
% $\Sigma = \{\bA,\bC,\bG,\bT\}$, 
 together with a decoding procedure. 
The alphabet $\Sigma$ of practical interest is typically $\{\bA,\bC,\bG,\bT\}$, corresponding to the four nucleotides that compose DNA. 
However, to simplify the exposition we will often focus on the binary case $\Sigma = \{0,1\}$. 
Most of the results can be extended to a quaternary alphabet in a straightforward way, as we briefly discuss in Section~\ref{sec:general}.

In practice, a codeword $[\x_1,\ldots,\x_M] \in \C$, which is a set of $\M$ sequences each of length $\len$, is stored as a physical mixture of $\M$ synthesized DNA molecules. 
Hence, we use the words molecule and sequence interchangeably. 
Our main parameter of interest is the storage rate
%A DNA storage code has two main parameters of interest.
%The first one is the storage (or coding) rate, 
\al{
R \defeq \frac{\log|\C|}{\M\len}, \label{eq:Rs}
}
i.e., the number of bits stored per DNA base synthesized.
% Due to the nature of the reading process, performed via sequencing, another parameter of interest is the \emph{reading rate}, defined as the number of bits recovered per DNA base sequenced:,
% \al{
% R_r \defeq \frac{\log|\C|}{\N\len}. \label{eq:Rr}
% }
To formally define the storage capacity, we consider an asymptotic regime where $M \to \infty$.
%Motivated by \eref{eq:noisefree} and by the practical constraint that the molecules should be short, 
We let the sequences have length 
\al{
\len \defeq \beta \log M
}
for a fixed parameter $\beta > 0$. 
This is motivated by the practical constraint that the synthesized molecules should be short.
As we will see in the next chapter, if $L$ grows slower than $\log M$,
% and by the fact that the capacity expression derived for shuffling channels in the next section, revealing that for $\beta \leq 0$ 
no positive rate is achievable. 
%\rh{to do: above needs to be changed, should we move the shuffling channel discussion.}
Moreover, as our results show, $L = \Theta(\log M)$ is the asymptotic regime of interest for this problem.
%In Section \ref{sec:discussion}, we discuss the case where $\len$ does not scale with $\M$.

We say that storage rate $R$ is achievable if there exists a sequence of DNA storage codes $\{\C_\M\}$, each with rate $R$, such that the decoding error probability tends to $0$ as $\M \to \infty$. 
The storage capacity is the supremum over all achievable storage rates.

\paragraph{Additional notation.} 
Throughout this monograph, $\log(\cdot)$ is the logarithm base $2$ and $\ln(\cdot)$ is the natural logarithm.
For a real number $x$, we let $(x)^+ = \max(x,0)$.
We let $\Z_+$ be the set of non-negative integers, and for a positive integer $N$, we let $[1:N] = \{1,\dots,N\}$.
For a vector $\bx = (x_1,\dots,x_d)$, 
we let $\norm[2]{\bx} = \sqrt{\sum_{i=1}^d x_i^2}$ be its $\ell_2$ norm, $\norm[1]{\bx} = \sum_{i=1}^d |x_i|$ its $\ell_1$ norm, and $\norm[\infty]{\bx} = \max_i |x_i|$ its $\ell_{\infty}$ norm.

We say that a random variable $X$ has the $\Bernoulli(p)$ distribution if $\Pr(X=1) = p$ and $\Pr(X=0) = 1-p$.
We say that a  non-negative integer-valued random variable $X$ has the $\Pois(\lambda)$ distribution if $\Pr(X=k) = \frac{e^{-\lambda}\lambda^k}{k!}$, for $k \geq 0$.

\chapter{Shuffling Channels}
\label{ch:shuffled}

As discussed in Chapter~\ref{ch:model}, a distinguishing aspect of DNA-based storage is the fact that the data stored in the DNA library is read out of order, in a shuffled manner.
In this chapter, we focus on the coding and capacity problems that arise in settings where ordering information is lost.
% An important property of DNA storage channels is the fact that the order or the molecules are lost. 

To build intuition, we first briefly discuss a simple shuffling channel which just shuffles the set of input sequences.
We then consider the (noise-free) shuffling-sampling channel, where all reads are noise-free; i.e., the discrete memoryless channel $p(y|x)$ in Figure~\ref{fig:generalchannel} is just the identity channel.
This will enable us to identify the costs and challenges that arise from the unordered nature of DNA-based storage.
Finally, we consider a generalization of the shuffling channel that allows variable-length pieces.
This generalization is motivated by the fact that the stored DNA molecules may also be subject to random breaks.
% To capture this phenomenon, we propose the ``torn-paper channel'', which assumes that the input to the DNA storage channel is torn up to pieces of different sizes.

%%%

\section{The shuffling channel}
\label{sec:simpleshuffling}

% A key feature of DNA storage is that data is stored on many short unordered sequences.
% In order to characterize the information-theoretic limits of DNA storage, we therefore start with an elementary channel model which models the aspect of sequences being unordered. 
% To that end, 
To begin our information-theoretic analysis of DNA-based data storage, 
we first put aside the noise that corrupts each of the individual DNA molecules and assume that the writing and reading processes are perfect; i.e., every single molecule is synthesized and sequenced correctly.
With these assumptions, the DNA storage channel is described by a shuffling channel, which corresponds to the dashed rectangle in  Figure~\ref{fig:generalchannel}.
The input to a shuffling channel is a list of  $M$ strings of length $\len$ and the output is an unordered, or shuffled, version of these strings.
% Studying the shuffling channel captures the loss of ordering information incurred in a DNA-based storage system.

The capacity of this noise-free shuffling channel
% that simply shuffles the input sequences uniformly, displayed in the dashed box in Figure~\ref{fig:generalchannel}, can be computed in bits per nucleotide as 
can be shown to be 
\al{
\lim_{M\to \infty} \left( \log |\Sigma| -  \frac{\log M}{L}\right)^+ = \left( \log |\Sigma| - \frac{1}{\beta}\right)^+,
 \label{eq:noisefree}
}
as we argue below.
This expression 
% in equation~\eref{eq:noisefree} 
supports the intuition that, when the length $L$ of the DNA molecules is large, the impact of shuffling is small, and we achieve close to the 2 bits/nucleotide that can be achieved for $|\Sigma| = 4$ if there is no shuffling and no noise.
% we can design codes by focusing on each individual molecule going through the channel.
Furthermore, 
% However, synthesizing long DNA molecules is challenging and costly, and practical DNA-based storage systems must store information on short DNA molecules, as discussed in the previous chapter.  %\citep{church_next-generation_2012,goldman_towards_2013,grass_robust_2015,yazdi_rewritable_2015,erlich_dna_2016,organick_scaling_2017}.
% Hence, understanding the storage capacity in the short-molecule regime is crucial.
% In particular, as suggested by 
equation~\eref{eq:noisefree} shows that 
when $L$ scales logarithmically in $M$, the capacity of the shuffling channel is nontrivial.

The capacity expression in equation~\eref{eq:noisefree} can be proved with a simple counting argument. 
% Since the order of the $M$ output sequences of the shuffling channel contains no information about the input, we may as well treat the output as an (unordered) multi-set.
Notice that if we view the input and output of the channel as  multi-sets of $M$ strings of length $\len$, then the channel does not affect the input at all.
Hence, the capacity is simply the logarithm of the number of distinct multi-sets of $M$ strings of length $\len$, divided by the number of nucleotides stored $M\len$.

The operation of counting the number of distinct multi-sets of a given size will be used often throughout this monograph. 
Notice that multi-sets can be equivalently represented by a histogram, which records how many times each element, out of a list of possible elements, appears.
If the number of possible distinct elements is $a$, a multi-set with $b$ elements can be represented by a vector $\type \in \Z_+^a$  with $\|\type \|_1 = b$.

% \rh{Would it make sense to have $t$ boldface? Later we use $\vf$ as a frequency vector, in a similar context, so it would probably make sense..} % Yes, I changed that.

\begin{lemma} \label{lem:histograms}
The number of distinct vectors $\type \in \Z_+^{a}$ with $\|\type \|_1 = b$ is
\aln{
\T[a,b] \defeq {a+b-1 \choose b} < \left( \frac{e(a+b-1)}{b} \right)^b.
}
\end{lemma}

\begin{proof}
Notice that vectors  $\type \in \Z_+^{a}$ with $\|\type \|_1 = b$ are in one-to-one correspondence with 
binary strings containing $(a-1)$ zeros and $b$ ones.
For $\type = (t_1,\ldots,t_a)$,
the corresponding string is 
\al{
\underbrace{1 \, \ldots \, 1}_{t_1} \, 0 \, \underbrace{1 \, \ldots \, 1}_{t_2} \, 0 \, \ldots \, 0 
\underbrace{1 \, \ldots \, 1}_{t_a}.
}
Such a string has $(a-1)$ zeros and $b$ ones,
and distinct strings with $(a-1)$ zeros and $b$ ones correspond to distinct vectors $\type$. 
The number of distinct strings of this form is
%permutations of (\ref{eq:string}) is 
\aln{
\frac{(a-1+b)!}{(a-1)!\, b!} = {a+b-1 \choose b}.
}
Finally, the upper bound in the statement of the lemma is a standard bound for binomial coefficients.
\end{proof}

Since the input to the shuffling channel is a multi-set of $M$ strings of length $L = \beta \log M$, the capacity of the shuffling channel is given by
\aln{
\lim_{M\to \infty} \frac{1}{ML} \log \T[|\Sigma|^L,M].
}
Using Lemma~\ref{lem:histograms} and the bounds $k \log (n/k) \leq \log {n\choose k} < k \log (en/k)$, it follows that the capacity is given by
\al{
\lim_{M\to \infty} \frac{1}{ML} \log {{|\Sigma|^{L}+M-1}\choose{M}} 
& = \lim_{M \to \infty} \frac{1}{L} \log \left(\frac{|\Sigma|^L+M-1}{M}\right) \nonumber \\
& = \lim_{M \to \infty} \frac{1}{L} \log \left(\frac{|\Sigma|^L}{M} + 1\right) \nonumber \\
& = \lim_{M \to \infty} \left( \log |\Sigma| - \frac{\log M}{L} \right)^+ \nonumber \\
& = (\log |\Sigma| - 1/\beta)^+, \label{eq:countingconverse}
}
as we wanted to show.
In particular, if $\beta \leq \frac{1}{\log|\Sigma|}$, the capacity is zero.
Notice that this was not obvious a priori since, even for $\beta \leq \frac{1}{\log|\Sigma|}$, one still has a large number $M$ of length-$L$ strings with $L \to \infty$ as $M \to \infty$, and it might have been possible to encode data on them at a positive rate.

% If we are operating over a quaternary alphabet $\Sigma = \{\bA,\bC,\bG,\bT\}$, 
% then the 

% and assuming that $4^L \gg M$, the storage capacity of the noise-free shuffling channel discussed above, in terms of \emph{bits per nucleotide}, is 
% \al{
% \frac{1}{ML} \log {{4^{L}+M-1}\choose{M}} \approx \frac{1}{L} \log \left(\frac{4^{L}+M-1}{M}\right) \approx 2 - \frac{\log M}{L},
% % +\frac{\log e}{L} + O\left(\frac{\log M}{ML}\right),
% %\label{eq:noisefree}
% }
% as claimed. 
% Here, we used Lemma~\ref{lem:histograms} and the approximation $\log {n\choose k} \approx k \log (n/k)$, which holds when $n \gg k$.
% where we used Stirling's approximation $\log n! \approx n \log n - n \log e + O(\log n)$.

% However, in order to make the shuffling channel a more realistic model of a real DNA-based 

\section{Capacity of the shuffling-sampling channel}
\label{sec:noisefree}

The simple analysis in Section~\ref{sec:simpleshuffling} shows that the shuffling aspect of DNA storage  effectively  causes the capacity to be reduced by $1/\beta$. 
% and, in particular, if $\beta \leq 1/\log|\Sigma|$, the capacity is zero.
In this section we analyze the impact of also bringing sampling into the picture.
The main result of this section generalizes the simple capacity expression in (\ref{eq:noisefree}) and settles the capacity of the noise-free shuffling-sampling channel, for a general sampling distribution $Q$.
% As it turns out, the capacity expression only depends on $Q$ through $q_0 = \Pr(N_i = 0)$.
% , given by the following theorem.
% It depends on the parameter $\Emolseen$, which is defined as the expected fraction of stored molecules we see at the output: 
% \[
% \Emolseen 
% %\lim_{\M\to \infty}
% %\EX{\frac{1}{\M} \sum_{i=1}^\M \ind{\N_i > 0}}
% = 1 - \q_0.
% \]
% Recall that the channel depends on the distribution $Q$ according to which the $N_i$'s are drawn. 
% For example, $Q$ can be a a Poisson distribution with mean $\cNM$. 
% %parameters $\{\p_i\}_{i=1}^\M$. We view the $\p_i$ as a function of $\M$; e.g., drawing uniformly at random from the pool of DNA fragments is modeled by $\p_i = 1/\M$. 
% For this particular case, we obtain $\Emolseen = 1 - e^{-\cNM}$. 
% We discuss other distributions later. 
As described in Chapter~\ref{ch:model}, the input to the channel are $M$ sequences of length $L = \beta \log M$.
The channel samples the $i$th sequence $N_i \sim Q$ times, and shuffles all sampled sequences.
Notice that $q_0$ is the probability that zero copies of $\x_i$ are drawn; i.e., $q_0$ is the expected fraction of input sequences never seen at the output. 
We focus on the case $\Sigma = \{0,1\}$ for simplicity of exposition, but the extension to general $\Sigma$ is straightforward, as discussed in Section~\ref{sec:general}.

\begin{theorem} \label{thm:noisefree}
% The storage rate $R$ is achievable if and only if 
The capacity of the shuffling-sampling channel is
\al{
C  & = (1-q_0) \left(1 - 1/\beta \right). \label{eq:storagecap}
%R_r & \leq \frac1c(1-e^{-\cNM})\left(1 - \frac1\beta \right)  \label{eq:rrbound}
}
In particular, if $\beta \leq 1$, no positive rate is achievable.
\end{theorem}

%This will allow us to establish the tradeoff between storage and recovery rates in Section \ref{sec:tradeoff}.
%from the relationship between $R_s$ and $R_r$ in (\ref{eq:Rr}).
%
%The following theorem characterizes the storage-sequencing tradeoff.

%\bfcomment{should we first give the two easier bounds as motivation? R: that is a good idea, lets do that. Perhaps we can also add intuition why no positive rate is achievable when $\beta \leq 1$, by elaborating on the case $\len=1$. }

%Given that $R_r = R_s/c$, we only need to show that $R_s$ is achievable if and only if (\ref{eq:rsbound}) holds.

%In the remainder of this section we provide the achievability and the converse for this result.

The capacity expression in equation~\eqref{eq:storagecap} can be intuitively understood through the achievability argument.
A storage rate of $R = (1-q_0) \left(1 - 1/\beta \right)$ can be easily achieved by prefixing all the molecules with a distinct tag or index, which effectively converts the channel to a block-erasure channel, as illustrated in Figure~\ref{fig:index}.
\begin{figure}[t]
\center
\includegraphics[width=0.92\linewidth]{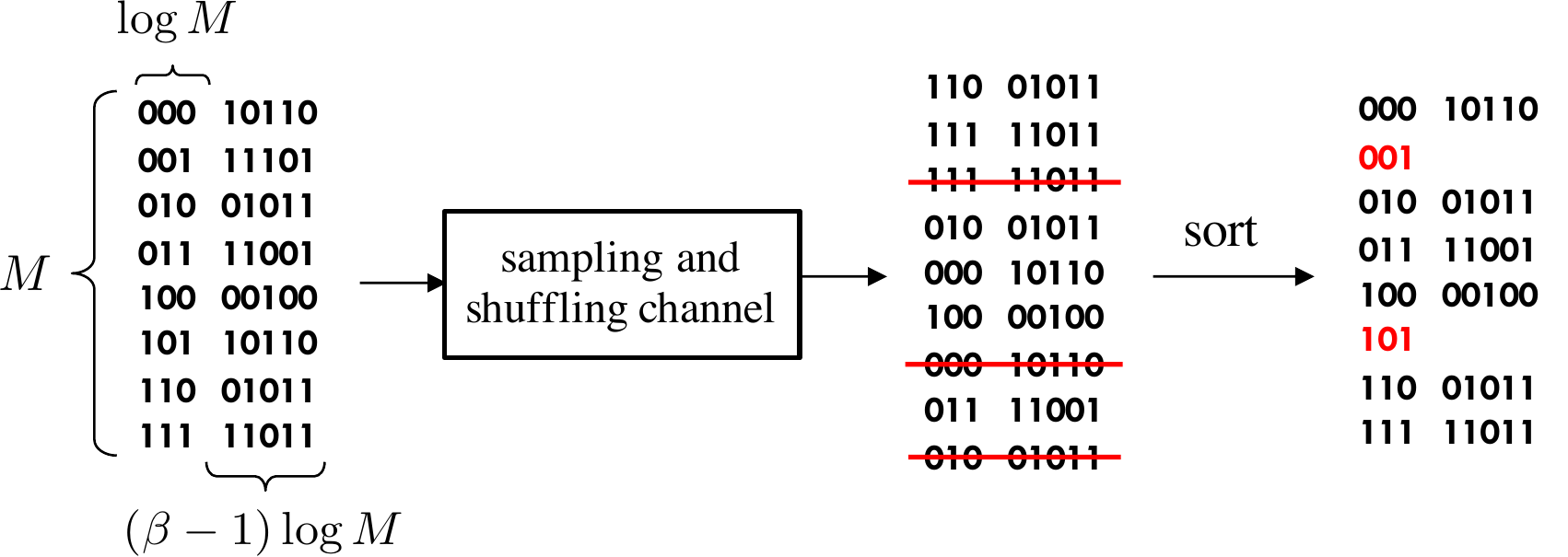}
\caption{
\label{fig:index}
Index-based scheme for the shuffling-sampling channel. 
All $M$ input sequences are prefixed with a unique index of length $\log M$.
At the output of the channel, the decoder uses the indices to remove duplicates and sort the molecules.
Notice that the missing indices (001 and 101) can be thought of as block erasures.
Hence, this scheme effectively creates a block-erasure channel, where blocks of size $L$ are erased with probability $q_0$.
}
\end{figure}

More precisely, we use the first $\log \M$ bits of each molecule to encode a distinct index. % for each molecule.
Then we have $\len - \log \M = (\beta - 1) \log M$ symbols left per molecule to encode data. 
The decoder can use the indices to remove duplicates and sort the molecules that are sampled. 
This effectively creates an erasure channel, where molecule $i$ is erased if it is not drawn (i.e., $\N_i=0$) which occurs with probability $q_0$. 
 %$e^{-\p_i \N}$.
%   where each symbol (i.e., each molecule) is erased with probability $(1 - 1/\M)^{\M \cNM}$. 
% Since $\Emolseen 
% = 
% \lim_{\M\to \infty}
% \EX{\frac{1}{\M} \sum_{i=1}^\M \ind{\N_i > 0}} = 1 - q_0 = \Emolseen$, 
% the storage rate $R$ achieved is $(1-q_0)(1-1/\beta)$.
Since the expected number of erasures is 
\aln{
\EX{\frac{1}{\M} \sum_{i=1}^\M \ind{\N_i = 0}} = q_0,
} 
we achieve storage rate 
\al{ \label{eq:indexingachievable}
\frac{(1-q_0)M(L-\log M)}{ML}=(1-q_0)(1-1/\beta).}
The surprising aspect of Theorem~\ref{thm:noisefree} is that this simple index-based scheme is optimal.
It is also worth noting that the capacity expression only depends on the sampling distribution $Q$ through the parameter $q_0$, i.e., the fraction of sequences that is not seen at the output.

%\subsection{Capacity expression for specific sampling distributions}

% \iscomment{May be good to write out the capacity expression for specific sampling distributions}
% \mycomment{Moved this here, because there is not too much to write, and added the paragraph below, ok? The paragraph is similar to one in the introduction.}

In order to gain intuition on a practical implication of this theorem, suppose that each sequence is drawn according to a Poisson distribution with mean $\cNM$, so that in expectation $E[N] = \cNM M$ sequences in total are drawn and $\cNM$ can be thought of as the sequencing coverage depth. 
The probability that a sequence is never drawn is $e^{- \cNM}$ and the capacity is
\al{
C = (1-e^{-\lambda})(1-1/\beta).
}
This suggests that practical systems should not operate at a very high coverage depth $\lambda$, as high coverage depth significantly increases the time and cost of reading, but only provides little storage gains, according to the capacity expression. 
Notice that, in order to guarantee that all $M$ sequences are observed at least once, we need $\N = \Omega(\M \log \M)$ \citep{LanderWaterman,motahari2013information}.
When $\M$ is large, it is wasteful to operate in this regime, 
% \iscomment{Should this say just $\Omega(\M)$?}
% regime where each sequence is read at least once, 
as this only gives a marginally larger storage capacity, but the sequencing costs can be exorbitant. 

The result in Theorem~\ref{thm:noisefree} is flexible to allow different sampling models.
In particular, one can consider separating the PCR amplification performed on each synthesized molecule from the sequencing step.
Since one cannot control the PCR amplification factor precisely, it is reasonable to assume that a molecule $\x$ is first randomly amplified and a total of $A \geq 0$ copies are stored.
If we consider a Poisson sampling model for the sequencing step, the effective coverage depth is $\lambda/\EX{A}$ (since we are actually sampling from $M \EX{A}$ molecules).
In this case, the probability that none of the copies of $\x$ is sampled at the output is $\EX{(e^{-\lambda / \EX{A}})^A} = \EX{(e^{(-\lambda / \EX{A})A}}$.
This can be recognized as the moment-generating function of $A$ evaluated at $-\lambda/ \EX{A}$.
In particular, when PCR is also modeled as a Poisson random variable with mean $\EX{A}=\alpha$, 
$\EX{e^{\theta A}}= e^{\alpha (e^\theta -1)}$, and
the capacity of the resulting noise-free shuffling-sampling channel is
\al{
C = \left(1-e^{-\alpha(1-e^{-\lambda/\alpha})}\right)(1-1/\beta).
}
The capacity can be similarly computed based on other sampling distributions (such as the Negative Binomial distribution observed in Section~\ref{sec:distributions}).

%\mycomment{
%Suppose the elements are drawn from a discrete probability distribution with the probability that element $i$ is drawn is given as $p_i$, $\sum_i p_i = 1$. Suppose we take $\N$ draws from this distribution. 
%The probability that segment $i$ is not drawn is $(1-p_i)^\N$.
%Thus the expected fraction of segments not drawn is
%\[
%\frac{1}{\M} \sum_{i=1}^\M (1 - p_i)^\N.
%\]
%}
%
%\mycomment{
%Suppose now that the number of times a given fragment is drawn is Poission distributed with mean $\cNM$. Then the expected number of fragments is $\cNM \N$.
%The probability a given segment $i$ is not drawn is $e^{-\cNM}$. 
%Thus the expected fraction of segments not drawn is $(1 - e^{-\cNM})$, 
%that is the same as in our model. 
%}

%\subsection{Achievability: Index-Based Coding}
%\input{./files/achievability}

\subsection{
\label{sec:motivationconverse}
Motivation for converse}

% \iscomment{this part needs to be rewritten to flow with section 3.1}

One can attempt to prove a converse by following the approach in Section~\ref{sec:simpleshuffling}.
Notice that one can view the noise-free shuffling-sampling channel as a channel where
the encoder chooses a histogram over the $2^L$ sequences of length $L$
% distribution (or a type) over the alphabet $\Sigma^\len$
and the decoder observes a noisy version of this histogram where the frequencies are perturbed according to $Q$.
% takes $N$ samples from this distribution.
From this angle, the question becomes ``how many histograms $\type \in \Z_+^{2^L}$ with $\|\type \|_1 = M$ can be reliably decoded?''
% , and restricting ourselves to index-based schemes restricts the set of types to those with $\|\type \|_\infty = 1$; i.e., no duplicate molecules are stored.
%, which may seem suboptimal.
% R: I think its clear from here that it is potentially suboptimal, as one restricts the space of codewords..
% For instance, including a codeword with $\x_1 = \x_2 = \ldots = \x_M = 00\ldots0$ might be reasonable as its output is predictable ($\y_1 = \ldots = \y_N = 00\ldots0$), and easy to decode.

Following the counting argument in Section~\ref{sec:simpleshuffling}, the total number of distinct histograms $\type \in \Z_+^{2^L}$ with $\|\type \|_1 = M$ is $\T[2^L,M]$.
Since this is the total number of histograms that can be used as an input, \eref{eq:countingconverse} implies that a bound to the shuffling-sampling channel capacity is
\al{
C \leq 1- 1/\beta.  \label{eq:simplebound2}
}
Therefore, if we had a shuffling-sampling channel  where the decoder observed exactly the $M$ stored molecules (i.e., $q_0 = 0$ and $q_1 = 1$), the index-based approach would be optimal from a rate standpoint.
%However, it is not clear that in the stochastic channel described in Section~\ref{sec:problem}, 
%restricting ourselves to schemes that do not encode duplicates among the $M$ molecules is optimal.

%The surprising aspect about the result in Theorem~\ref{thm:storagecap} is that this simple index-based coding scheme is optimal.
A simple outer bound that involves $q_0$ can be obtained by considering a genie that provides the decoder with the ``true'' index of 
each sampled molecule.
% In other words, $[\x_1,\ldots,\x_\M]$ are the stored molecules, and the decoder observes $[(\y_1,i_1),(\y_2,i_2),\ldots,(\y_\N,i_\N)]$ where $i_j$ is such that $\y_j = \x_{i_j}$.
In other words, $[\x_1,\ldots,\x_\M]$ are the stored molecules, and the decoder observes $[\y_1,\ldots,\y_N]$ and the mapping $\sigma \colon \{1,\ldots, N\} \to \{1,\ldots, M \}$ so that
$\y_j = \x_{\sigma(j)}$.
This converts the channel into an erasure channel with block-erasure probability $\q_0$, which yields
\al{
C \leq 1 - \q_0. \label{eq:simplebound1}
}
It is intuitive that the bound~\eqref{eq:simplebound1} should not be achievable, as the decoder in general cannot sort the molecules and create an effective erasure channel.
However, it is not clear a priori  whether prefixing every molecule with an index is 
optimal.
%an optimal strategy.
%optimal from the rate standpoint.

Combining (\ref{eq:simplebound2}) and (\ref{eq:simplebound1}) implies that $C \leq \min(1-q_0,1-1/\beta)$, but there is still a gap to the achievable rate in (\ref{eq:indexingachievable}).
The converse presented in the next section utilizes a more careful genie that does not give us the permutation $\sigma$ (but something weaker) to show
% that the potential gain of considering codewords of this form is vanishingly small.
that $C \leq (1-q_0)(1-1/\beta)$,
% the bounds in  (\ref{eq:simplebound1}) and (\ref{eq:simplebound2}) can in fact be ``combined,'' 
implying the optimality of index-based coding approaches.
\subsection{Converse}
\label{sec:converseCounting}

\newcommand\set{\mathrm{set}}

\newcommand\Nsampled{{\tilde N}}

%Let $\setx$ be the multi-set of $M$ molecules $x_{(i)}^K$, $i=1,\ldots,M$, each of which is a 
%length-$\len$ sequences, and $\sety$ be the multi-set of molecules $y_{(i)}^K$, $i=1,\ldots,N$ observed by the decoder.
Let $[\x_1,\ldots,\x_M]$ be the $\M$ length-$\len$ molecules written into the channel 
and $[\y_1,\ldots,\y_N]$ be the length-$\len$ molecules observed by the decoder. 
Notice that, whenever the channel output is such that $\y_i = \y_j$ for $i \ne j$, the decoder cannot determine whether both $\y_i$ and $\y_j$ were sampled from the same molecule $\x_\ell$ or from two different molecules that obey $\x_\ell = \x_k, \ell \neq k$. 
In order to derive the converse, we consider a genie-aided channel that removes this ambiguity. 
As illustrated in Figure~\ref{fig:geniechannel}, before sampling the $\N$ molecules, the genie-aided channel appends a unique index of length $\log \M$ to each molecule $\x_i$, which results in the set of tagged molecules $\{(\x_i,z_i)\}_{i=1}^\M$. 
We emphasize that the indices $z_i$ are all unique, and are chosen randomly and independently of the input sequences $\{\x_i\}_{i=1}^\M$. 
Notice that, in contrast to the naive genie discussed in Section~\ref{sec:motivationconverse}, this genie does \emph{not} reveal 
%which molecule gets which index. 
the index $i$ of the molecule $\x_i$ from which $\y_\ell$ was sampled.
Therefore, the channel is not reduced to an erasure channel, 
and intuitively the indices are only useful for the decoder to determine whether two equal samples $\y_\ell = \y_k$ came from the same molecule or from distinct molecules. 
%Intuitively, once this information is provided to the decoder, it is optimal to discard $\y_\ell$ if it is known to be a sample from the same molecule as $\y_k$.
%Hence, the decoder now has roughly $(1-e^{-\cNM})M$ useful samples 

\begin{figure}
\vspace{0mm}
	\center
       \includegraphics[width=0.84\linewidth]{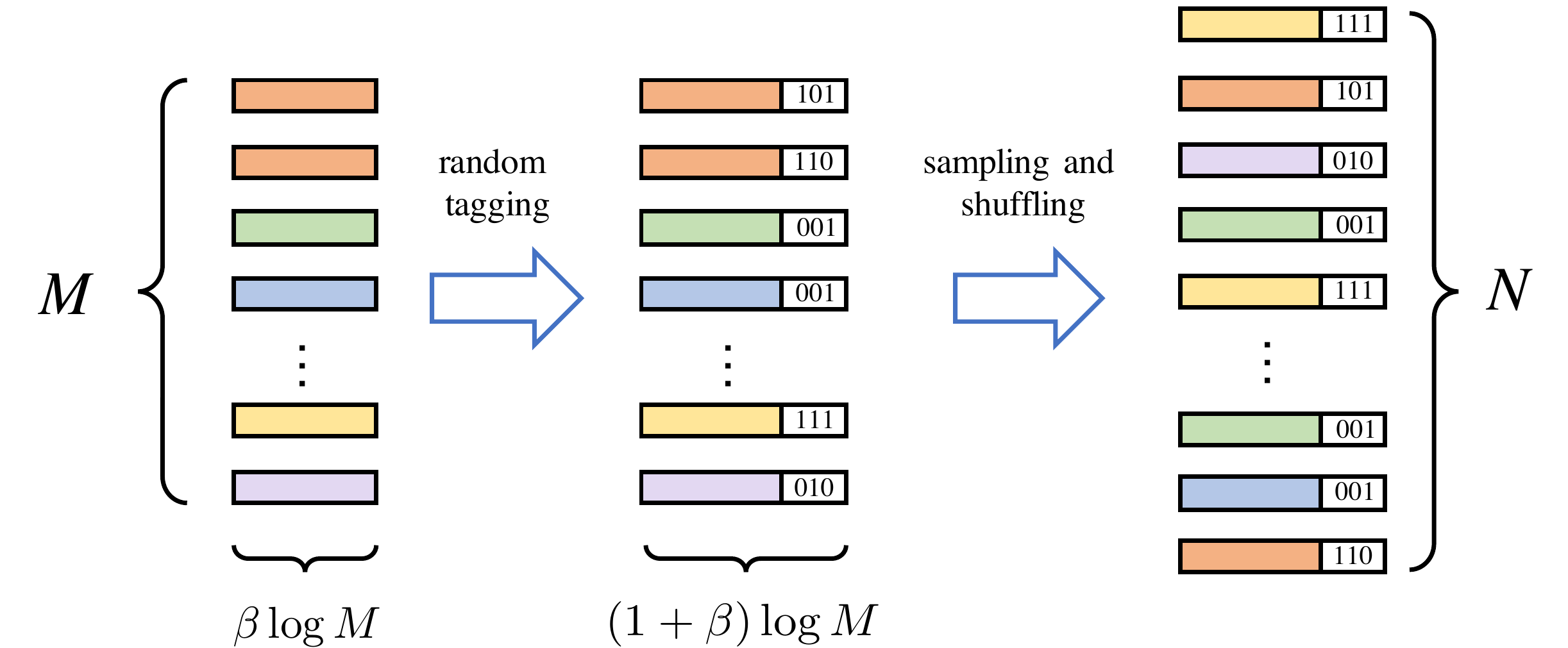}
       \caption{Genie-aided channel for converse for the shuffling-sampling channel. 
       The genie appends a random (but unique) index to each of the input sequences $\x_1,\dots,\x_M$.
       This allows the decoder to identify which outputs sequences originated from the same input sequences, and remove them.
       %\mycomment{change $\N$ to $\Nsampled$}
       \label{fig:geniechannel}}
\end{figure}

The output of the genie-aided channel, denoted by $\{(\y_i,z_{\sigma(i)})\}_{i = 1}^\N$, is then obtained 
by sampling from the set of tagged molecules $\{(\x_i,z_i)\}_{i=1}^\M$, in the same way as the original channel samples the original molecules.
The mapping $\sigma \colon [1:N] \to [1:M]$ is such that $\y_i$ was sampled from $\x_{\sigma(i)}$.
Notice that the actual mapping $\sigma$ is not revealed to the decoder.

It is clear that any storage rate  achievable in the original channel can be achieved on the genie-aided channel, as the decoder can simply discard the indices, or stated differently, the output of the original channel can be obtained from the output of the genie-aided channel. 

Notice that $\{(\y_i,z_{\sigma(i)})\}_{i = 1}^\N$ is in general a multi-set.
We will let $\set( \{(\y_i,z_{\sigma(i)})\}_{i = 1}^\N )$ be the set obtained from $\{(\y_i,z_{\sigma(i)})\}_{i = 1}^\N$ by removing any duplicates. 
Then $\set( \{(y_i,z_{\sigma(i)})\}_{i = 1}^\N )$ is a sufficient statistic for $\{\x_i\}_{i=1}^\M$ since all tagged molecules are distinct objects, and sampling the same tagged molecule $(\x_i,z_i)$ does not yield additional information on $\{\x_i\}_{i=1}^\M$. 
More formally, conditioned on $\set( \{(\y_i,z_{\sigma(i)})\}_{i = 1}^\N )$, $\{\x_i\}_{i=1}^\M$ is independent of the genie's channel output $\{(\y_i,z_{\sigma(i)})\}_{i = 1}^\N$. 

Next, we define the frequency vector $\vf \in \Z_+^{\M^{\beta}}$ 
(note that $|\Sigma^{\len}| = 2^{\beta \log \M} = \M^\beta$) that is 
obtained from $\set( \{(\y_i,z_{\sigma(i)})\}_{i = 1}^\N )$ as follows.
% in the following way. 
The entry of $\vf$ corresponding to the molecule $\y \in \Sigma^{\len}$ is given by 
\aln{
\vf[\y] 
\defeq \left| \left\{ (\y_j,z_{\sigma(j)}) \in \set( \{(\y_i,z_{\sigma(i)})\}_{i = 1}^\N )
\colon  \y_j = \y \right\} \right|. 
}
The frequency vector $\vf$ is essentially a histogram that counts the number of occurrences of $\y$ in the set of tagged molecules $\{(\y_i,z_{\sigma(i)})\}_{i = 1}^\N$.
Notice that the entries of $\vf$ can take values greater than one, because at the input we can choose to use the same molecule for multiple $\x_i$.

Since $\set( \{(\y_i,z_{\sigma(i)})\}_{i = 1}^\N )$ is a sufficient statistic for $\{\x_i\}_{i=1}^\M$
and the tags added by the genie were chosen at random and independently of $\{\x_i\}_{i=1}^\M$, it follows that $\vf$ is also a sufficient statistic for $\{\x_i\}_{i=1}^\M$. 
Hence, we can view the (random) frequency vector $\vf$ as the output of the channel without any loss. 
Notice that $|\set( \{(\y_i,z_{\sigma(i)})\}_{i = 1}^\N )| =  \| \vf \|_1$, 
% \bluechange{
and we have $\| \vf \|_1 \leq M$ and 
$\EX{ \| \vf \|_1 / \M } = 1 - \q_0$. 
% }
The following lemma asserts that $\norm[1]{\vf}$ does not exceed its expectation by much.

\begin{lemma} \label{lem:conc}
For any $\delta > 0$, the frequency vector $\vf$ at the output of the genie-aided channel satisfies
\aln{
\PR{
 \frac{\| \vf \|_1}{\M} >
1 - \q_0
% \Emolseen
 + \delta 
 } 
 \to 0, \text{ as } \M \to \infty. % \Emolseen \defeq \EX{\frac{1}{\M} \sum_{i=1}^\M \ind{\N_i > 0}}.
}
\end{lemma}
\begin{proof}
%We next provide a proof for the case where $\N_i$ is drawn from a Poisson distribution with mean $\p_i \N$. 
Note that the number of distinct fragments drawn is 
\[
\frac{\norm[1]{\vf}}{\M} = \frac{1}{\M} \sum_{i=1}^\M \ind{ \N_i > 0 }.
\]
Since $\ind{ \N_i > 0 }$ are independent random variables with expectation $1-\q_0$, %$1- e^{-\p_i \N}$, 
Hoeffding's inequality yields
\[
\PR{ \frac{\norm[1]{\vf}}{\M} \geq (1-\q_0)%\frac{1}{\M} \sum_{i=1}^\M (1- e^{-\p_i \N}) 
+ \delta} 
\leq e^{-2\M \delta^2},
\]
which concludes the proof. \end{proof}

We remark that an analogue of Lemma~\ref{lem:conc} can be proved for the sampling-with-replacement model, as described in \citet{DNAStorageIT}.

%R: used f' instead of \hat f, since hat f is typically used as the estimate of f

We now append the coordinate
$
f_0 = ( 1-q_0 + \delta) \M - \| \vf \|_1
$
to the beginning of $\vf$ to construct $\vf' = (f_0,\vf)$.
Notice that when $\| \vf \|_1 \leq (1-q_0 + \delta)\M$ (which by Lemma~\ref{lem:conc} happens with high probability), we have $\| \vf' \|_1 = ( 1-q_0 + \delta)\M$.
% \bluechange{
This construction of $\vf'$ will allow us to utilize Lemma~\ref{lem:histograms} below.
% }
Fix $\delta > 0 $, and define the event 
\begin{align}
\label{eq:defeventE}
    \E =  \{ \| \vf \|_1 > (1-q_0 + \delta)\M \}
\end{align} 
with indicator function $\one_\E$.
By Lemma~\ref{lem:conc}, $\PR{\E} \to 0$ as $\M \to \infty$. 
Consider a sequence of codes $\{\C_\M\}$ with rate $R$ and vanishing error probability.
Let $W$ be the message to be encoded, chosen uniformly at random from $\{1,\ldots,2^{\M \len R}\}$.
From Fano's inequality we have
\al{
\M \len R_s  
&= 
H(W)
=
I(W; \vf') + H(W| \vf') \nonumber \\
&\leq H(\vf') + 1 + P_e \M \len R_s,
}
where $P_e$ is the probability of a decoding error, which by assumption goes to zero as $M \to \infty$. 
We can then
upper bound the achievable storage rate $R$ as
% \aln{
% \M & \len R (1-P_e) 
% \leq H (\vf' )  +1 \leq  H\left( \vf', \one_\E \right) +1.
% }
\al{
\M \len R (1-P_e)  
% \nonumber \\
% &
&\leq H (\vf' )  +1 \nonumber \\
&\leq  H\left( \vf', \one_\E \right) +1 \nonumber \\
&\leq \PR{ \E } H\left( \left. \vf'\, \right|  \E  \right) + \PR{ \bar \E } H \left( \left. \vf' \, \right|  \bar \E  \right) + H(\one_\E) + 1.
% \nonumber \\
% &\leq \PR{ \E } \log \T [\M^\beta+1,\M] \nonumber \\ 
% &\quad \quad + \log \T [\M^\beta+1, (\Emolseen + \delta)\M] + 2, 
\label{eq:rsbound}
}
% \bluechange{
Note that the vector $\vf'$ above has dimension 
$M^\beta + 1$
and, given the event $\bar \E$ occurs, $\|\vf'\|_1 = (\Emolseen + \delta)\M$, and we have $H(\vf'|\bar \E) 
% = H(\vf|\E) 
\leq \log \T[M^\beta + 1, (\Emolseen + \delta)\M]$,
where $\T[a,b]$ is the number of vectors $x \in \Z_+^a$ with $\| x \|_1 = b$.
%in the appendix, which allows us to obtain
%\aln{
%\log & \T[a,b] = \log {a+b-1 \choose b}\\
%&  < \log \left( \frac{e(a+b-1)}{b} \right)^b = b \log \left( e + \frac{e(a-1)}{b} \right).
%}
%\aln{
%\log & \T[a,b] < \log \left( \frac{e(a+b-1)}{b} \right)^b = b \log \left( e + \frac{e(a-1)}{b} \right),
%}
%and thus
From Lemma~\ref{lem:histograms}, 
\aln{
\log \T & [\M^\beta+1, (\Emolseen + \delta)\M] \\
%& < (1-e^{-\cNM} + \delta)M \log \left( \frac{e(a+b-1)}{b}  \right)  \\
& \leq (\Emolseen + \delta)\M \log \left( e + \frac{e M^{\beta-1}}{(\Emolseen + \delta)} \right) \\ 
& \leq (\Emolseen + \delta)\M \log \left( \alpha M^{\beta-1} \right) \\
& \leq (\Emolseen + \delta)\M  [ (\beta-1) \log \M + \log \alpha ],
}
where $\alpha$ is a positive constant. 
Moreover, we notice that $\vf'$ is a function of $\vf$, which is a vector in $\Z_+^{M^\beta}$ with $\|\vf\|_1 \leq M$. 
Next, we define $\vf'' = (f_0,\vf)$, 
where $f_0 = M - \|\vf\|_1$ so that $\|\vf''\|_1 = M$ and we can apply Lemma~\ref{lem:histograms}.
% , and 
We note that
\aln{
H(\vf'|\E) & = H(\vf''| \E) \leq 
\log \T [\M^\beta+1,\M]  \\ 
& \leq \M \log \left( \frac{e(M+M^\beta}{M}\right) \\
& \leq \M ( (\beta -1) \log \M + \log \alpha' ),
}
where $\alpha'$ is another positive constant.
% }
Dividing (\ref{eq:rsbound}) by $\M\len$ and applying the bounds above yields
\aln{
R (1 - P_e) &\leq \Pr( \E ) \frac{ \M [(\beta -1) \log \M + \log \alpha' ]}{\M \len} \\
& + \frac{( \Emolseen + \delta)\M  [ (\beta-1) \log \M + \log \alpha ]}{\M \len} + \frac{2}{\M\len} \\
& \leq \Pr( \E ) \left( \frac{\beta - 1}{\beta} +\frac{\log \alpha'}{\beta \log M}  \right) \\ 
&  + (\Emolseen + \delta) \left(1-\frac{1}{\beta}  + \frac{\log \alpha}{\beta \log \M}\right)  + \frac{2}{\M\len}.
}
Finally, letting $\M \to \infty$ yields
\aln{
R \leq ( \Emolseen + \delta ) \left(1- 1/\beta \right),
}
since $\Pr(\E) \to 0$ by Lemma~\ref{lem:conc}.
Since $\delta > 0$ can be chosen arbitrarily small, this concludes the converse proof.
% of Theorem~\ref{thm:noisefree}.

\section{Storage-sequencing tradeoff} 
\label{sec:tradeoff}

Most studies on DNA-based storage emphasize the storage rate (or storage density) as the main figure of merit, while sequencing  costs are often disregarded.
This is due to the fact that current costs of sequencing technologies are orders of magnitude lower than synthesis costs.
However, it is still important to understand, for a given storage rate, how much sequencing is required for reliable decoding, as this determines the time and cost required for retrieving the data.

From this perspective, it makes sense to analyze the tradeoff between storage rates and the amount of sequencing required for reliable recovery.
To do this, we can consider, in addition to the storage rate, the \emph{recovery rate}, defined as the number of bits recovered per DNA base sequenced,
\al{
R_r \defeq \frac{\log|\C|}{\N\len}. \label{eq:Rr}
}
In a practical setting, one can control the amount of sequencing performed, typically specified in terms of the coverage depth $N/M$.
If we consider the error-free shuffling-sampling channel from Section~\ref{sec:noisefree}, in the case
% In this part we briefly discuss the storage-recovery tradeoff for 
where $Q$ is a Poisson distribution with mean $\cNM$, then $\cNM = N/M$ is the coverage depth, and one would like to choose a value of $\lambda$ that achieves a good trade-off between storage rate and recovery rate.

If we let $R_s$ be the storage rate (previously just $R$, see \eqref{eq:Rs}), from Theorem~\ref{thm:noisefree} and 
%(\ref{eq:Rr}), the $(R_s,R_r)$ feasibility region follows directly:
the fact that $R_s = \cNM R_r$, the $(R_s,R_r)$ feasibility region can be fully characterized.

\begin{cor}
For the error-free shuffling-sampling channel with $Q \sim {\Pois}(\cNM)$, rates $(R_s,R_r)$ are achievable if and only if, for some $\lambda > 0$,
\al{
R_s  & \leq (1-e^{-\cNM})\left(1 - 1/\beta \right),\nonumber \\
R_r & \leq \frac{1-e^{-\cNM}}{\cNM}\left(1 - 1/\beta \right). 
\nonumber
}
\end{cor}

%A simple outer bound to the $(R_s,R_r)$ region can be obtained in the following way.
%The probability that a given block $\bc_i$ is not sampled is given by
%\aln{
%\left( 1 - 1/M \right)^N = \left( 1 - 1/M \right)^{Mc} \to e^{-\cNM}, \text{ as } M \to \infty. 
%}
%Therefore, in expectation, a $ (1 - e^{-\cNM})$ fraction of the length-$\len$ blocks is observed by the decoder, and a natural upper bound on the storage rate is
%\aln{
%R_s \leq (1 - e^{-\cNM}) \log |\Sigma|  = 1 - e^{-\cNM}.
%}
%Hence, an outer bound for the $(R_s,R_r)$ achievable region is the region enclosed by the curve $(R_s,R_r) = 2(1-e^{-\cNM},\frac{1-e^{-\cNM}}{c})$, for $0<c<\infty$, which is shown below:

This region is illustrated in Figure~\ref{fig:region}.
This tradeoff suggests that a good operating point is achieved by not trying to maximize the storage rate (which technically requires $\cNM \to \infty$).
Instead, by using some modest coverage depth $\cNM=1,2,3$, most of the storage rate ($63 \%,86\%,95\%$, respectively) can be achieved.
%In addition, the indexing-based decoding should be computationally efficient.
This is in contrast to what has been done in practical DNA storage systems that have been developed thus far, where the decoding phase utilizes very deep sequencing.

\begin{figure}
	\center
       \includegraphics[width=6cm]{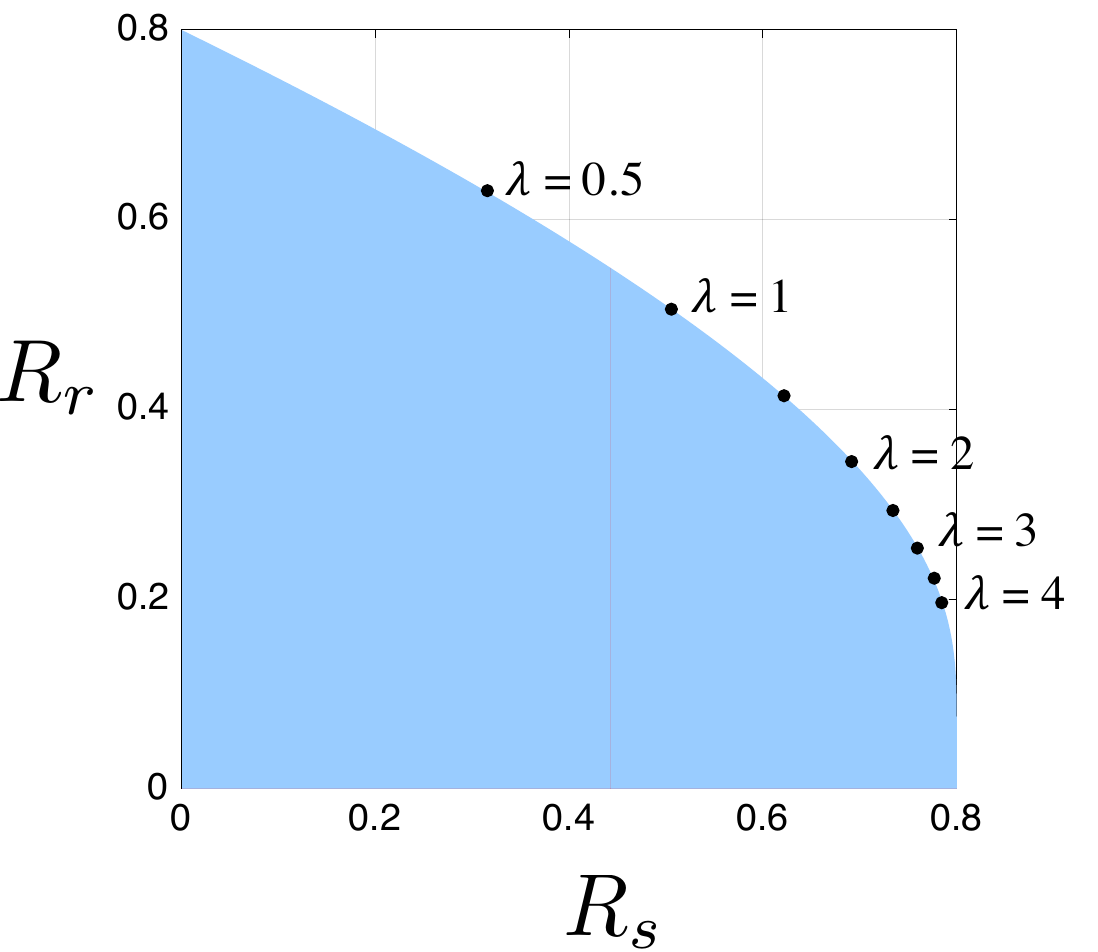} 
       \caption{$(R_s,R_r)$ feasibility region for $\beta = 5$. \label{fig:region}}
%       \vspace{-0.5cm}
\end{figure}

To be concrete, suppose we are interested in minimizing the overall cost of storing data on DNA. 
% In the Biohackers experiment \cite{biohackers}, the cost of synthesis was roughly $\$0.0001$ per base, while 
Synthesis costs
can be several orders of magnitude higher than sequencing costs.
% are currently larger than sequencing costs by about a factor $q = 10^4$-$10^5$ \iscomment{need citation, Erlich?}. %
% \rh{let's take the biohackers experiment: roughly $1.0684e-04$ USD per basepair for 4 Million, because it's a large amount. Then sequencing with Illumina, say with 10-15 Million reads with Miseq, about 1000 USD.
% So in this example, sequencing costs are only about a factor of 50 larger. Can we make the calculation above for say $q=100$, but then also say ``Suppose synthesis costs are by a factor 100 larger. That's a reasonable number for a 2020 experiments, but of course this varies significantly with the synthesis and sequencing technology used.
% }
Suppose that per-base synthesis costs are $\alpha$ times larger than per-base sequencing costs.
Then, if our goal is to minimize the cost for synthesizing and sequencing a given number of bits, the overall cost is proportional to 
\aln{
\frac{q}{R_s} + \frac{1}{R_r} = \frac{q+\cNM}{(1-e^{-\cNM})(1-1/\beta)}.
}
This quantity can be minimized over $\cNM$, yielding the optimal cost per bit.
For example, if $q=100$, $\cNM = 4.7$ and if  $q=10000$, $\cNM = 9.2$. 
This suggests that it is possible to achieve reliable DNA storage that minimizes overall costs using only
moderate coverage depths.
% than the one used in existing implementation 
We point out that, in practice, one may be interested in optimizing other quantities such as reading time or considering a scenario where the data is read more than once.

\section{DNA breaks and variable-length pieces}

% DNA molecules are also subject to breaks.
% What happens if we just write one long sequence, and have it break into pieces of different lengths?
% Torn-Paper channel.

During synthesis, storage, and sequencing, the DNA molecules in a DNA storage system are subject to random breaks.
While a careful handling of the DNA library can minimize the occurrence of such breaks, from a system design point of view, it is interesting to study the impact that such breaks may have on the overall system capacity.
Such an understanding can enable system designs that combat DNA degradation through coding, and that therefore do not require a tight control of the temperature of the DNA library. 
Furthermore, in current implementations of DNA storage systems, the data is read via high-throughput sequencing, which is typically preceded by physical fragmentation of the DNA with techniques like \emph{sonication},
which utilizes sound vibrations to fragment the DNA in random locations~\citep{pomraning2012library}.
% \rh{we should explain what sonication is, in a side-sentence, I haven't heard that expression before}
% In addition, the torn-paper channel is related to the DNA shotgun sequencing channel, studied in \citep{MotahariDNA,BBT,gabrys2018unique}, but in the context of variable-length reads, which are obtained in nanopore sequencing technologies \citep{laver2015assessing,mao2018models}.

A basic initial model to study breaks in DNA molecules is the \emph{torn-paper channel}, proposed by \citet{tpc-globecom}.
In this setting, the channel input is a length-$n$ binary sequence $\tpcx$ that is torn into pieces by the channel.
A tearing point between any two consecutive symbols in $\tpcx$ is assumed to occur with a fixed probability $p_n \in (0,1)$ and independently from all other possible tearing locations.
As a result, the lengths $N_1,N_2,\ldots$ of each of the fragments produced by the channel has a Geometric$(p_n)$ distribution.
The channel output is a shuffled list of these pieces or, equivalently, an unordered set containing these pieces.

\begin{figure}
% \vspace{-6mm}
\centering
\includegraphics[width=0.7\linewidth]{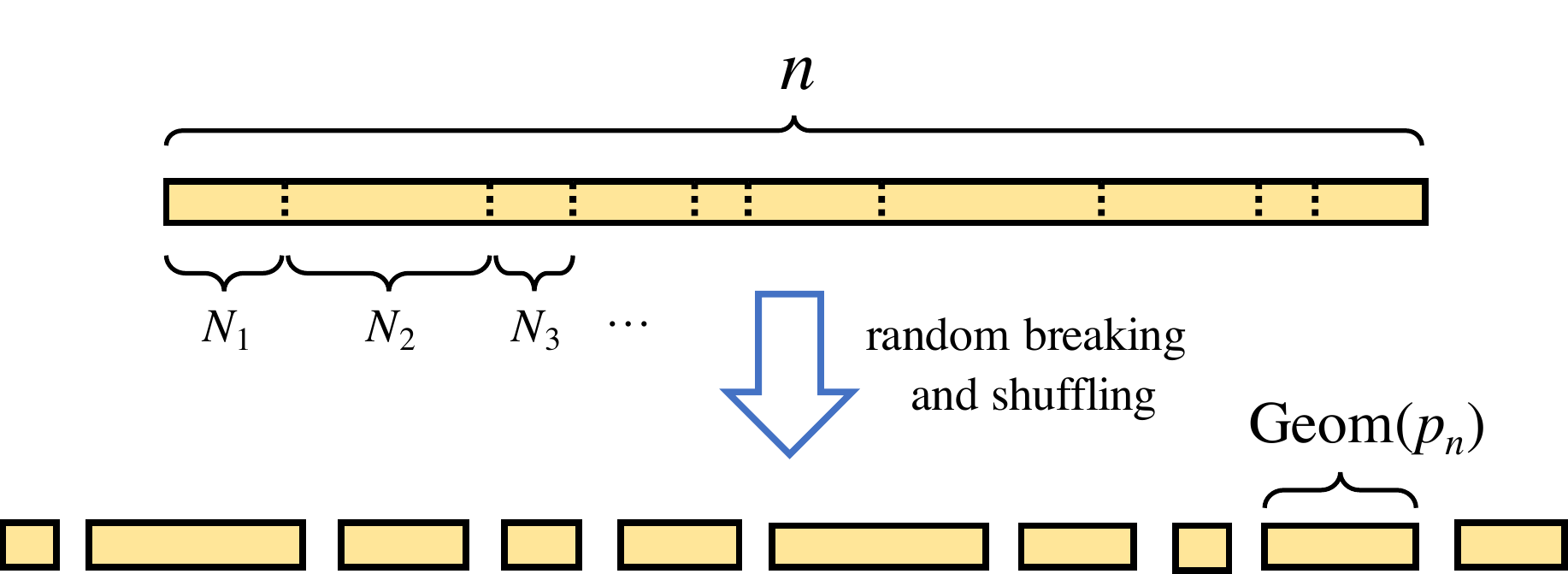}
\vspace{4mm}
\caption{In the torn-paper channel, a single (binary) sequence of length $n$ is sent through the channel.
The channel tears the input sequence into pieces of random sizes.
If tearing points occur according to an i.i.d.~Bernoulli process, as considered in \citet{tpc-globecom}, then the fragment sizes will follow a Geometric$(p_n)$ distribution.
\label{fig:variable}
}
\end{figure}

As it turns out, characterizing the capacity of this channel is non-trivial even in this noise-free setting.
To build some intuition, notice that the expected fragment length is $E[N_i]=1/p_n$.
Suppose that, instead of random-length fragments, we have fragments of a deterministic length $1/p_n$.
In this case, the channel breaks the length-$n$ sequence $\tpcx$ into $n p_n$ pieces of equal length.
Notice that since, in this case, the tearing locations are known a priori, the channel is identical to the error-free shuffling channel discussed in Section~\ref{sec:noisefree}, with every piece being sampled exactly once.
In the notation from Section~\ref{sec:noisefree}, the capacity of an error-free shuffling channel where each piece is drawn exactly once is $1- 1/\beta$, where $\beta$ is the ratio between the logarithm of the number of pieces and the piece length.
This implies that, for the torn-paper channel with pieces of a deterministic size $1/p_n$, the capacity is
\aln{
\lim_{n \to \infty} 1 - \frac{ \log (n p_n) }{1/p_n}
= \lim_{n \to \infty} 1 -  p_n \log n - p_n \log p_n.
}
We see that the regime of interest (where the capacity is non-trivial) is when $p_n$ scales as $1/\log n$,
% \rh{'$p_n$ decays as $1/\log n$' means '$p_n = c/\log n$'? } % IS: changed to scales as
in which case the capacity reduces to
% Coding for the case of deterministic fragments of length $N_i = 1/p_n$ is easy: since the tearing points are known, we can prefix each fragment with a unique identifier, which allows the decoder to correctly order the $n p_n$ fragments.
% From the results in \citep{noisyshuffling}, such an index-based coding scheme is capacity-optimal for the shuffling channel, and any achievable rate must satisfy, for large $n$,
\al{
C_{\text{deterministic length}} = (1 - \lim_{n\to\infty} p_n \log n)^+.
\label{eq:capshuffling}
}
Mimicking the notation from previous sections, it is reasonable to let 
\aln{
\beta \defeq \lim_{n \to \infty} \frac{1/p_n}{\log(np_n)}
= \lim_{n \to \infty} \frac{1}{p_n \log n},
}
where the last equality assumes the regime $p_n \sim 1/\log n$.
The quantity $\beta$ plays the same role of a ``normalized fragment length''  (normalized by the logarithm of the number of pieces) as it did in previous sections. 
From \eref{eq:capshuffling}, we now have that 
the capacity for the case of deterministic fragment lengths becomes 
% $(1-\alpha)^+$.
\al{
\left(1-\frac{1}{\beta}\right)^+,
}
analogous to the noise-free shuffling channel considered in Section~\ref{sec:motivation_shuffling}.

It is not clear a priori whether the capacity of the torn-paper channel with random fragment lengths should be higher or lower than $(1-1/\beta)^+$.
The fact that the tearing points are not known to the encoder makes it challenging to place a unique index in each fragment, suggesting that the torn-paper channel is ``harder'' and should have a lower capacity.
However, this intuition is incorrect and the capacity of the torn-paper channel with ${\rm Geometric}(p_n)$-length fragments is in fact higher than $(1-1/\beta)^+$.
More precisely, we 
have the following result.
\begin{theorem} \label{thm:captpc}
The capacity of the torn-paper channel is 
\al{\label{eq:captpc}
C_{\rm TPC} = e^{-1/\beta},
}
where $\beta = \lim_{n \to \infty} \frac{1}{p_n \log n}$.
\end{theorem}
% show that the capacity of the torn-paper channel is 
% $C = e^{-\alpha}$.
The comparison between $C_{\rm TPC}$ and (\ref{eq:capshuffling}) is shown in Figure~\ref{fig:curves}.
Intuitively, this boost in capacity comes from the tail of the geometric distribution, which guarantees that a fraction of the fragments will be significantly larger than the mean $E[N_i] = 1/p_n$.
This allows the capacity to be positive even for $\beta \leq 1$, in which case the capacity of the deterministic-tearing case in (\ref{eq:capshuffling}) becomes $0$.
% In the next subsection, we provide an intuition for the capacity expression in Theorem~\ref{thm:captpc}.
We refer to \citet{tpc-globecom} for the proof of Theorem~\ref{thm:captpc}.

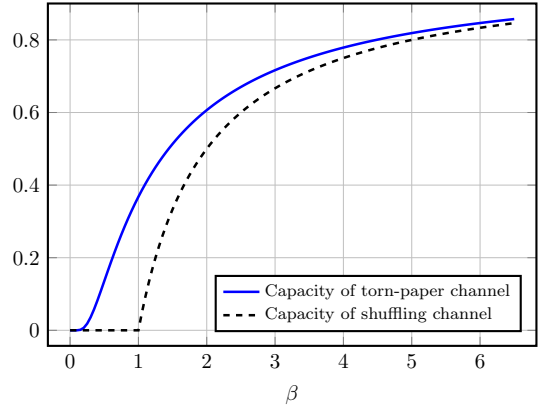
\begin{figure}
\centering 
% \vspace{1cm}
\begin{tikzpicture}
[scale=0.8]

\begin{axis}[enlargelimits=0.05,xlabel=$\beta$,line width=1pt,grid=major,
legend cell align={left}, legend style={font=\footnotesize,at={(0.97,0.04)},anchor=south east},
width=0.8\textwidth,
height=0.6\textwidth,
% width=0.54\textwidth,
% height=0.41\textwidth,
]
\addplot[very thick, name path=f,color=blue,mark=none,draw=blue!] table[x index=0,y index=2]{data/tpc.dat};   
\addlegendentry{Capacity of torn-paper channel}

\addplot[dashed, very thick, name path=g,color=black,mark=none,draw=black!] table[x index=0,y index=1]{data/tpc.dat};   
\addlegendentry{Capacity of shuffling channel}

% \addplot[very thick, name path=g,color=red,mark=none,draw=red!] table[x index=1,y index=0]{./achievable.dat};   
% \addlegendentry{Interleaved-pilot scheme}

\end{axis}
\end{tikzpicture}  
\vspace{0mm}
\caption{\label{fig:curves} 
Comparison between the capacity of the torn-paper channel $C = e^{-1/\beta}$ and the capacity of the shuffling channel with fragments of fixed length $1/p_n$.
% , and the rate achieved on the torn-paper channel by the explicit code construction based on the  interleaved-pilot scheme.
}
% \vspace{-2mm}
\end{figure}

\paragraph{General torn-paper channel capacity.}
The capacity expression of the torn-paper channel can be generalized to the case of arbitrary fragment length distribution.
Moreover, it can be generalized to allow some of the fragments to be lost and not be observed at the output.

Consider a torn-paper channel that breaks the length-$n$ sequence $\tpcx$ into pieces of lengths $N_1,N_2,\dots$, following a common distribution.
In addition, pieces of length $\ell$ are deleted according to a length-dependent deletion probability $d(\ell)$.
As shown by \citet{tpc-isit}, the capacity of this channel is described by the informal expression
\al{ \label{eq:covreor}
C_{\text{general TPC}} = \text{coverage} - \text{reordering-cost}.
}
Here, ``coverage'' refers to the fraction of $\tpcx$ that is covered by  pieces of length at least $\log n$ that are 
observed at the channel output, and ``reordering cost'' refers to the fraction of bits that would need to be dedicated to indices in order to ``unshuffle'' the pieces longer than $\log n$.

In the case of the error-free shuffling channel where each piece is observed exactly once, as long as the pieces have length at least $\log n$, the coverage is $1$, and the reordering cost is the fraction of bits needed for adding $\log(np_n)$ indexing bits to each of the $np_n$ fragments; i.e.,
\aln{
\lim_{n\to\infty} \frac{n p_n \log(n p_n)}{n} = \lim_{n\to\infty} p_n \log n = \frac{1}{\beta},
}
yielding the capacity formula $(1-1/\beta)^+$.
In the case of piece lengths distributed as $\Geometric(p_n)$, using the Exponential approximation to a Geometric random variable, the coverage by pieces of length at least $\log n$ can be shown to be $(1+1/\beta)e^{-1/\beta}$, and the reordering cost can be shown to be $\beta^{-1} e^{-1/\beta}$, yielding the capacity expression $e^{-1/\beta}$.
% \citep{tpc-IT}. \rh{ref missing}

The recipe in (\ref{eq:covreor}) can be used to derive other closed-form capacity expressions for different piece-length distributions and deletion probability functions $d(\ell)$.
As shown in \citet{tpc-isit}, the capacity depends on $d(\ell)$ through the asymptotic behavior of $d(\ell)$, which is captured by \aln{
\hat d(\xi) \defeq \lim_{n \to \infty} d( \xi \log n) \log n.
}
Table~\ref{tab:table1} shows several examples of capacity expressions for torn-paper channel with specific choices of fragment length distribution and deletion probability $\hat d(\xi)$.

\begin{table}[h]
% \vspace{-4mm}
  \begin{center}
    \caption{Capacity for different torn-paper channels}
    % \vspace{-1mm}
    \label{tab:table1}
    \begin{tabular}{c|c|c} 
      $N_i$ & $\hat{d}(\xi)$ & \text{Capacity}\\
      \hline
      \Geometric($p_n$) & $0$ &  $e^{-1/\beta}$\\
      \Geometric($p_n$) & $\epsilon$ & $(1- \epsilon)e^{-1/\beta}$\\
       \Geometric($p_n$) & $e^{-\gamma\xi}$ & $e^{-1/\beta}\left(1 - \frac{\beta^{-2}e^{-\gamma}}{(\beta^{-1} + \gamma)^2}\right)$\\
      $\text{Unif}[0:\gamma\log{n}]$, $\gamma \geq 1$ & $0$ &$\left((\gamma - 1)/\gamma\right)^2$ \\
       \text{fixed} $\ell_n$, $\ell_n \geq \log{n}$ & $0$ & $(1-1/\beta)^+$
      \vspace{3mm}
    \end{tabular}
  \end{center}
%   \vspace{-3mm}
\end{table}

\chapter{Noisy Shuffling Channels}
\label{ch:noisy}

In Chapter~\ref{ch:shuffled}, we studied the impact of two key aspects of DNA storage systems: the natural shuffling that occurs due to the fact that molecules are stored in an unordered fashion, and the output sampling that occurs due to the sequencing operation.
In this chapter, we study the additional effect of errors within the stored sequences.

In order to focus on the impact of noise, we consider a simpler sampling distribution $Q$ than the general one considered in Section~\ref{sec:noisefree}.
Specifically, we focus on a ``single-draw'' setting, where each sequence is drawn either once or not at all (but never multiple times).
The main result in this chapter characterizes the capacity of the noisy shuffling-sampling channel in Figure~\ref{fig:generalchannel} when $Q$ is a  $\Bernoulli(1-q)$ distribution and the noisy channel is a binary symmetric channel with crossover probability $p$.
We then discuss extensions to the binary erasure channel and to more general channels.
In all these cases, the capacity expression, at least in some parameter regime, is shown to be
\al{ \label{eq:generalcap}
(1-q)(C_{\text{noisy}} - 1/\beta)^+,
}
where $C_{\text{noisy}}$ is the capacity of the noisy channel $p(y|x)$ in Figure~\ref{fig:generalchannel}, leading to the conjecture that (\ref{eq:generalcap}) holds for general discrete memoryless channels (see Section~\ref{sec:general}).

\section{Noisy shuffling-sampling channel with single draws}
\label{sec:bsc_thm}

 In this section we study the capacity of the noisy shuffling-sampling channel where the sampling distribution is $\text{Bernoulli}(1-q)$ for a fixed parameter $q$.
 Hence 
 we have $\Pr(N_i= 0) = q$ and $\Pr(N_i=1)=1-q$, for $i = 1,\ldots,M$, and the number of output sequences $N$ satisfies $E[N] = (1-q)M$.
Moreover, we will assume that the molecules are each corrupted by a BSC with crossover probability $p$.
We refer to this channel as the BSC shuffling-sampling channel.

% We show that provided that the length of the sequences is sufficiently large relative to the number of sequences (i.e., $\beta$ is large enough), treating each length-$\len$ sequence as the input to a separate BSC and encoding a unique index into each sequence is capacity achieving.

As in the error-free shuffling-sampling channel considered in Section~\ref{sec:noisefree}, we consider an index-based coding scheme. 
As we will show, for a large set of parameters $p$ and $\beta$, this scheme is capacity-optimal.
More precisely, we describe a scheme based on an outer and an inner code and argue that it achieves a rate arbitrary close to
\al{
\Rind = (1-q)(C_{\rm BSC} - 1/\beta), \label{eq:lower}
}
where $C_{\rm BSC} = 1 - H(p)$ is the capacity of a BSC with crossover probability $p$.
This scheme is depicted in Figure~\ref{fig:innerouter}.
As the outer code, we take an erasure-correcting code with block length $\M$ and rate $(1-q)$, where each symbol is itself a binary string of length
% and with symbols in an alphabet of size $\Sigma^(L(1-H(p))$ of length
$\len(1-H(p) - 1/\beta-\epsilon) \approx L C_{\rm BSC} - \log M$, for some small $\epsilon>0$.
As inner code, we take a code designed for a BSC with codewords of length $\len$ and rate $R_{\text{BSC}} = 1-H(p) - \ep \approx C_{\rm BSC}$.
We first encode the data using the outer code, which yields $\M$ symbols given as binary strings of length 
\aln{
\len(1-H(p) - 1/\beta-\epsilon)
% = L(1-H(p)-\epsilon) - \log M 
= L R_{\text{BSC}} - \log M.
}
We take each symbol, add a unique binary index of length $\log \M$ and encode the resulting sequence using the BSC code, which yields $M$ length-$L$ sequences.
With this scheme, we encode a total of $(1-q)M(\len \Rbsc - \log M)$ data bits, with a data rate of 
%The rate achieved by this scheme is thus
\al{
\frac{(1-q)\M \left( \len \Rbsc - \log \M\right)}{\M \len} = (1-q)(\Rbsc - 1/\beta). \label{eq:schemerate}
}
Since $\ep > 0$ can be chosen arbitrarily small,
%the capacity of the BSC-shuffling channel is lower bounded by
this scheme achieves a rate arbitrarily close to the rate given in~\eqref{eq:lower}, as claimed.
% if $1 - H(p) - 1/\beta > 0$.
% \bluechange{
For simplicity, in this short argument we did not take into account that the inner codeword is decoded in error with a vanishing probability; we refer to \citet{DNAStorageIT} for a more formal achievability argument taking this into account.

\begin{figure}[t]
	\center
       \includegraphics[width=0.75\linewidth]{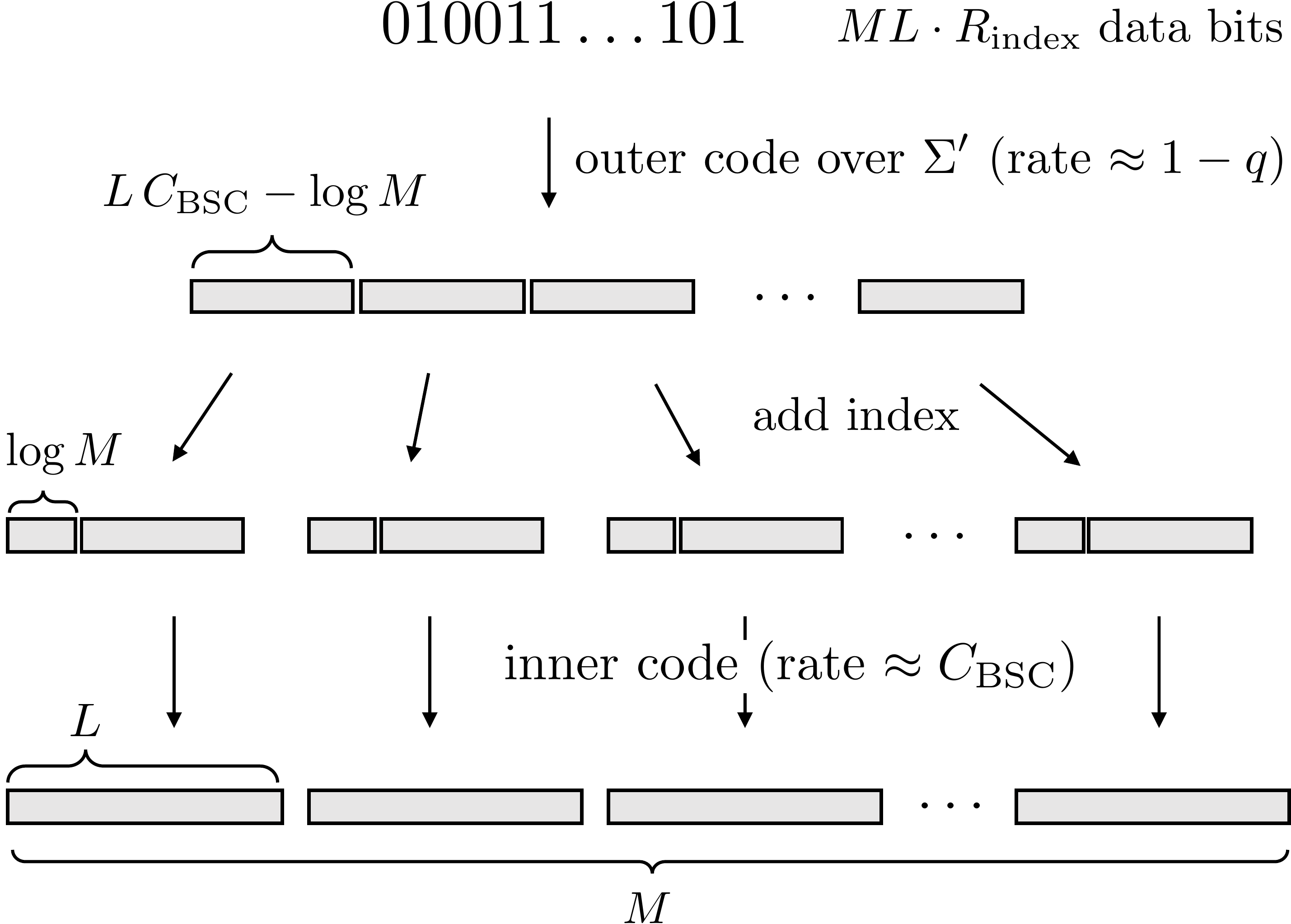} 
       \vspace{2mm}
      \caption{Index-based encoding for the BSC shuffling-sampling channel. First, $ML \cdot R_{\rm index}$ data bits are encoded using an erasure outer with rate $1-q$, where each codeword comprises $M$ symbols from an alphabet $\Sigma'$ with 
      $\log |\Sigma'| = L C_{\rm BSC} - \log M$ (or equivalently, each symbol is a binary string of length $L C_{\rm BSC} - \log M$).
      A unique $\log M$-bit index is added to each symbol, and a capacity-achieving BSC code is applied to each length-$L C_{BSC}$ binary string, producing $M$ length-$L$ binary strings.
      \label{fig:innerouter}}
\end{figure}

On the other hand, the result from Section~\ref{sec:noisefree}, with $Q \sim {\rm Ber}(1-q)$ implies that $C \leq (1-q)(1-1/\beta)$, since the error-free shuffling-sampling channel cannot be worse than the noisy shuffling-sampling channel.
Furthermore, a simple genie-aided argument where the decoder observes the shuffling map can be used to establish that $C \leq (1-q)\Cbsc$, where $\Cbsc = 1-H(p)$ is the capacity of a BSC with crossover probability $p$.
Hence, a capacity upper bound is given by
\al{
C \leq (1-q) \min\left[ 1- H(p), 1-1/\beta \right]. \label{eq:upper}
}
% since both the capacity of the BSC, $\Cbsc = 1-H(p)$, and the capacity of the shuffling channel, $1-1/\beta$, must be upper bounds to the capacity of the noisy shuffling channel, $C$.
Our main result improves on the upper bound in (\ref{eq:upper}), and establishes that for parameters $(p,\beta)$ in a certain regime, the lower bound in equation~\eqref{eq:lower} is the capacity.

%\begin{theorem}
%\label{thm:mainres}
%%If $\beta$ is large enough, 
%If $p \leq 0.1$ and $\beta \geq 6.4$, 
%%If $(p,\beta) \in \S$, the capacity of the BSC-shuffling channel is given by
%\al{
%C = 1 - H(p) - 1/\beta. \label{eq:capacity}
%}
%Moreover, if $\beta \leq 1$, $C = 0$.
%%In particular, if the capacity of the BSC-shuffling channel is greater than $0.77$, it is given by (\ref{eq:capacity}).
%%In particular, if $p < 0.01$ and $\beta \geq 6$, the capacity of the BSC-shuffling channel is given by (\ref{eq:capacty}).
%\end{theorem}
%
%
%As we will see, the converse argument presented in the next section and the capacity expression in (\ref{eq:capacity}) hold for a set of parameters $(p,\beta)$ larger than those specified in Theorem~\ref{thm:mainres}.
%In fact, for all $(p,\beta)$ in the blue region of Figure~\ref{fig:capacity}, the capacity is given by (\ref{eq:capacity}).
%Hence, for $p < 0.01$ for instance, we need $\beta \geq 2.35$ for (\ref{eq:capacity}) to hold.
%%Other points along the boundary of the blue region are $(0.025,8.4)$, $(0.05,12.9)$, and $(0.1,28.5)$.
%%Also note that, for $\beta < 1$, regardless of $p$, no positive capacity can be achieved (marked by the bottom region in Figure~\ref{fig:capacity}).

\begin{theorem}
\label{thm:bsc}
For the BSC shuffling-sampling channel,
%
%
%If $\beta > 1$ and $p < \tfrac12$ satisfy
%$
%1-H(2p)- 2/\beta > 0, %\label{eq:thm_condition}
%$
%then 
\al{
C = (1-q)(1 - H(p) - 1/\beta), \label{eq:capacity}
}
as long as $p<1/4$ and $1-H(2p)- 2/\beta > 0$.
% with $\beta > 1$ and $p < \tfrac12$.
%In particular, if $p < 0.01$ and $\beta \geq 6$, the capacity of the BSC-shuffling channel is given by (\ref{eq:capacty}).
Moreover, if $\beta \leq 1$, the capacity is $C = 0$.
\end{theorem}

The set of parameters $(p,\beta)$ such that $1-H(2p)- 2/\beta > 0$ and $p<1/4$ is the blue region in Figure~\ref{fig:capacity}.
In particular, (\ref{eq:capacity}) holds if $p \leq 0.1$ and $\beta \geq 6.4$, or if $p \leq 0.01$ and $\beta \geq 2.35$.

%\begin{figure}[h]
%\centering 
%\vspace{1cm}
%\begin{tikzpicture}[scale=0.75]
%
%    \begin{semilogxaxis}[enlargelimits=0.08,xlabel=$p$,ylabel=$\beta$]
%\addplot[name path=f,color=blue,mark=none,draw=blue!] table[x index=0,y index=1]{./curve.dat};   
%   
%   \path[name path=axis] (axis cs:0.00001,20) -- (axis cs:0.1,20);
%   \addplot[blue!20] fill between[of=f and axis];
%
%   \path[name path=axis2] (axis cs:0.00001,0.9) -- (axis cs:0.5,1) -- (axis cs:0.5,20) -- (axis cs:0.1,20);
%   \addplot[gray!20] fill between[of=f and axis2];    
%   
%   \addplot[name path=f, fill=red, fill opacity=0.5, draw=none, mark=none]
%coordinates {
%    (0.00001, 0)
%    (0.00001, 1)
%    (0.5, 1)
%    (0.5, 0)
%};

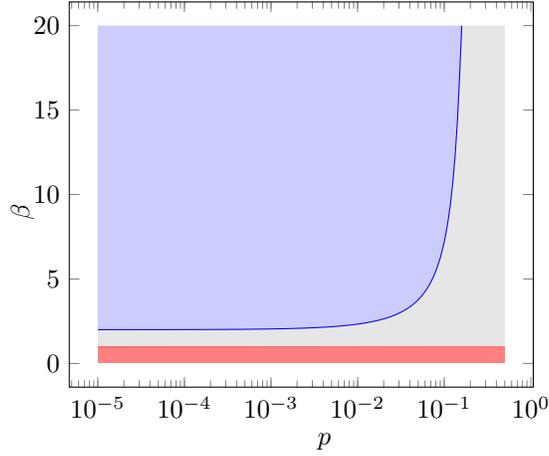
\begin{figure}[t]
\centering 
% \vspace{4mm}
\begin{tikzpicture}[scale=0.9]

    \begin{semilogxaxis}[xlabel=$p$,ylabel=$\beta$,
    enlargelimits=0.07]
\addplot[name path=f,color=blue,mark=none,draw=blue!] table[x index=0,y index=1]{./curve.dat};   
   
   \path[name path=axis] (axis cs:0.00001,20) -- (axis cs:0.1,20);
   \addplot[blue!20] fill between[of=f and axis];

   \path[name path=axis2] (axis cs:0.00001,0.9) -- (axis cs:0.5,1) -- (axis cs:0.5,20) -- (axis cs:0.1,20);
   \addplot[gray!20] fill between[of=f and axis2];    
   
   \addplot[name path=f, fill=red, fill opacity=0.5, draw=none, mark=none]
coordinates {
    (0.00001, 0)
    (0.00001, 1)
    (0.5, 1)
    (0.5, 0)
};
      
\end{semilogxaxis}
\end{tikzpicture}  
\vspace{1mm}
\caption{\label{fig:capacity}
Parameter regions for which the capacity is characterized.
The capacity in the blue region is given by $C = (1-q)(1 - H(p) - 1/\beta)$, and the capacity in the red region (i.e., for $\beta < 1$) is $0$.
In the gray region, it is still unknown.}
\vspace{-2mm}
\end{figure} 
% \iscomment{Need to fix the labels on this figure}

\subsection{Converse}
\label{sec:converse}
% !TEX root = main.tex

To derive the converse, we view the input to the channel as a binary string of length $\M \len$, denoted by
\aln{
X^{\ML} = \left[ X_1^\len, X_2^\len, \ldots , X_\M^\len \right] 
\in \{0,1\}^{\M\len}
}
or, equivalently, $\M$ strings of length $\len$ concatenated to form a single string of length $\M \len$.
Similarly, the output of the channel is
\aln{
Y^{\NL} = \left[ Y_1^\len, Y_2^\len, \ldots, Y_\N^\len \right] 
\in \{0,1\}^{\N\len},
}
where $N = \sum_i N_i$.
It is useful to define a vector $S^N \in \{1,\ldots,\M\}^\N$ indicating the input string from which each output string was sampled. 
Furthermore, we let $Z^{\NL}= \left[ Z_1^\len, \ldots, Z_\N^\len \right]$ be the random binary error pattern created by the BSC on the $N$ non-deleted strings.
We can now define the input-output relationship
\al{
Y_k^\len &= X^{\len}_{S(k)} \oplus Z^{\len}_{k}, 
\quad
\text{for $k=1,\ldots,N$},
\label{eq:inputoutput}
}
where $\oplus$ indicates elementwise modulo $2$ addition. 
Note that the $N_i$'s are fully determined by the vector $S^\N$ since $N_i = | \{ i \colon S(k) = i\} |$.
Also note that, since $Q \sim {\rm Ber}(1-q)$, $N \leq M$ with probability $1$.

Consider a sequence of codes for the BSC shuffling-sampling channel with rate $R$ and vanishing error probability.
Let 
$
X^{\ML} = \left[ X_1^\len, \ldots, X_\M^\len \right]
$
be the input to the channel when we choose one of the $2^{\ML R}$ codewords from one such code uniformly at random, and 
$
Y^{\NL} = \left[ Y_1^\len, \ldots, Y_\M^\len \right]
$
be the corresponding output.
%%We also define
%%\aln{
%%& X_k^\len = X^{\ML}[k,...,k+\len-1]
%%& Y_k^\len = Y^{\NL}[k,...,k+\len-1]
%%}
%%for $k=1,...,\M$; i.e., $X_k^\len$ is the $k$th length-$L$ in the input, and $Y_k^\len$ is the $k$th length-$L$ in the output.
%Let $Z^{\NL}= \left[ Z_1^\len, ..., Z_\M^\len \right]$ be the binary error pattern created by the BSC on $X^{\ML}$ and let $S^N \in \{1,...,\M\}^\M$ be a vector encoding the random shuffle induced by the channel, so that 
%\aln{
%Y_k^\len &= X^{\len}_{S(k)} \oplus Z^{\len}_{S(k)},
%}
%for $k=1,...,\M$.
%%\rhcomment{
%%Would it make sense to call this throughout $Y_k$? So as to view the output as a set of strings of length $L$ instead of a long string of length $ML$..
%%Also, shouldn't the equation above be:
%%\aln{
%%Y^{\NL}[k,...,k+\len-1] &= X^{\ML}[S(k),...,S(k)+\len-1]  \\
%%&+ E^{\ML}[S(k),...,S(k)+\len-1]
%%}
%%so that it is not just shuffled, but also passed through the BSC?
%%}
%RH: Nv not used anymore
%We let $\Nv = (N_1,N_2,...,N_M)$, 
% \bluechange{
From Fano's inequality we have that 
$H(X^{\ML}|Y^{\ML}) \leq 1 + P_{e,M} ML \leq ML \ep_M$,
where $P_{e,M}$ is the decoding error probability (of the code indexed by $M$) and $\{\ep_M\}$ is a sequence such that $\ep_M \to 0$ as $M \to \infty$.
% }
Thus,
\begin{align*}
MLR 
&= H\left(X^{\ML}\right) \leq I\left(X^{\ML};Y^{\NL}\right) + \ML \epsilon_{\M},
\end{align*}
where $\epsilon_{\M} \to 0$ as $\M \to \infty$.
Then,
\al{
ML(R-\epsilon_{\M}) & = H\left(Y^{\NL}\right) -  H\left(Y^{\NL}| X^{\ML} \right) \nonumber \\
& = H\left(Y^{\NL}\right) -  H\left(S^N,Z^{\NL},Y^{\NL} | X^{\ML}\right)
\nonumber \\ & \quad \quad 
+ H\left(S^N,Z^{\NL}|X^{\ML},Y^{\NL}\right) \nonumber \\ 
& = H\left(Y^{\NL}\right) -  H\left(S^N,Z^{\NL},Y^{\NL} | X^{\ML}\right)
\nonumber \\ & \quad \quad 
 + H\left(S^N|X^{\ML},Y^{\NL}\right) \label{eq:bound1}
}
The last equality follows by noticing that, given $(S^N,X^{\ML},Y^{\NL})$, 
one can compute
%$Z^{\NL}$
$
Z^{\len}_{k} = Y_k^\len \oplus X^{\len}_{S(k)} 
$
for $1 \leq k \leq N$, 
and thus we have $H\left(Z^{\NL}|X^{\ML},Y^{\NL},S^N\right) = 0$.
Since $N$ is a function of $S^N$, and $S^N$ and $Z^{NL}$ are independent of $X^{ML}$, the second term in (\ref{eq:bound1}) can be expanded as
\al{
& H\left(S^N,Z^{\NL},Y^{\NL} | X^{\ML}\right) \nonumber \\
& = H\left(S^N | X^{\ML} \right)  + H\left(Z^{\NL}| S^N, X^{\ML} \right) 
% \nonumber \\
% & \quad 
+ H\left(Y^{\NL} | X^{\ML}, S^N,Z^{\NL}\right)  \nonumber\\
& \eqnum H\left(S^N, N \right) + H\left(Z^{\NL}| S^N, N \right) 
% \nonumber \\
% & \quad 
+ H\left(Y^{\NL} | X^{\ML}, S^N,Z^{\NL}\right) \nonumber\\
&\eqnum H(N) + H\left(S^N |N  \right) + H\left(Z^{\NL}| N \right) %+ H\left(Y^{\NL} | X^{\ML}, S^N,Z^{\NL}\right) 
\nonumber\\
\quad &\eqnum H(N) + \sum_{n=1}^M \Pr(N = n) \left[ \log \frac{M!}{(M-n)!} + nL H(p) \right] \nonumber \\
\quad &\eqnum  \sum_{n=1}^M \Pr(N = n)  \left(n \log M  + n L H(p) \right) + o(ML) \nonumber \\
\quad & = \Ex[N] \M \left(  \log \M + \len H(p) \right) + o(\ML) \nonumber \\
\quad & = (1-q)\left[ \M \log \M + \ML H(p) \right] + o(\ML). \label{eq:bound2}
}
%where the last equality follows from Stirling's approximation.
% \bluechange{
In $(i)$, we used the facts that $S^N$ is independent of $X^{ML}$, $N$ is a function of $S^N$, and $Z^{NL}$ is independent of $X^{ML}$ given $S^N$.
Notice that $X^{NL}$ is only dependent on $S^N$ through $N$ (which is a random variable).
% }
For $(ii)$ we used that $H\left(Y^{\NL} | X^{\ML}, S^N,Z^{\NL}\right) = 0$ since $Y^{\NL}$ is determined by $X^{\ML}, S^N,Z^{\NL}$, 
and $(iii)$ follows from the fact that, given $N=n$, $S^N$ is chosen uniformly at random from all vectors in $\{1,\ldots,M\}^n$ with distinct elements.
For $(iv)$, we used the fact that, from Stirling's approximation,
\aln{
\log \frac{M!}{(M-n)!} & = M\log M - (M-n) \log (M-n) +  o(ML) \\
& = M\log M - (M-n) \log M \\ 
& \quad \quad  + (M-n) \log \frac{M}{M-n} +  o(ML) \\
& = n \log M + (M-n) \log \frac{M}{M-n} +  o(ML),
}
and, by Jensen's inequality,
\aln{
0 & \leq \sum_{n>0} \Pr(N = n) (M-n) \log \frac{M}{M-n}  \\ 
& \leq (M- \Ex[N]) \log \frac{M}{(M- \Ex[N])} \nonumber \\
& = (1-q) M \log 1/q = o(ML).
}
%%\iscomment{There may be a better way to argue $(i)$. It's so long...}
%
%%, using Stirling's approximation and Jensen's inequality,  
%%\aln{
%%0 \leq \sum_{n>0} \Pr(N = n) \log \frac{M^n}{n!} & \leq \sum_{n>0} \Pr(N = n) n \log \frac{eM}{n} \\
%%& \leq \Ex[N] \log \frac{eM}{\Ex[N]} = (1-q) M \log e/(1-q) = o(ML).
%%}
%%\aln{
%% \sum_{n>0} \Pr(N = n) \log n!  & = \sum_{n>0} \Pr(N = n) \left(n \log M - \log \frac{M^n}{n!} \right) \\
%% & = \Ex[N] \log M + o(ML) = (1-q) M \log M + o(ML),
%%}
%%because $0 \leq \log (M^n/n!) \leq \log M$.
%
%%The third term in (\ref{eq:bound1}) can be simplified as
%%\aln{
%%& H\left(S^N,Z^{\NL}|X^{\ML},Y^{\NL}\right) \\
%%& = H\left(S^N|X^{\ML},Y^{\NL}\right) + H\left(Z^{\NL}|X^{\ML},Y^{\NL},S^N\right).
%%}
%%Notice that, given $S^N$, one can match each length-$\len$ string in $Y^{\NL}$ to its corresponding length-$\len$ string in $X^{\ML}$, and add them (modulo 2) to obtain the error pattern.
%%This implies that $H\left(Z^{\NL}|X^{\ML},Y^{\NL},S^N\right) = 0$.
%
%
In order to finish the converse, we need to jointly bound the first and third terms in equation~(\ref{eq:bound1}).
%This is the most challenging part of the proof and 
This is summarized in a lemma.
\begin{lemma} \label{lem1}
If $\beta$ and $p < 1/4$ satisfy
%\aln{
%& \beta (1-H(\alpha)) > 2.05  \text{ and } \\
%& \beta D(\alpha/2||p) > 1
%}
\al{
1-H(2p)- 2/\beta > 0, \label{eq:condition}
}
then it holds that 
\aln{
H\left(Y^{\NL}\right) + H\left(S^N|X^{\ML},Y^{\NL} \right) \leq (1-q)\ML + o(\ML).
}
\end{lemma}
The parameter regime $(p,\beta)$ for which~\eqref{eq:condition} holds is the regime in which our capacity expression holds, illustrated in Figure~\ref{fig:capacity}. 
%The regime is depicted in blue in Figure~\ref{fig:capacity}. 
Combining (\ref{eq:bound1}), (\ref{eq:bound2}) and Lemma \ref{lem1}, we have
\aln{
& ML(R-\epsilon_{\M}) \leq (1-q)\left( \ML  - \ML H(p) - M \log \M \right) +  o(\ML).
}
Dividing by $\ML$ and letting $\M \to \infty$ yields the converse.

\subsection{Intuition for Lemma~\ref{lem1}}
\label{sec:intuition}

Rather than providing a full proof of Lemma~\ref{lem1}, here we discuss the intuition behind it, and refer to \citet{DNAStorageIT} for the complete proof.
To discuss the intuition for Lemma~\ref{lem1}, let us focus on the case $q=0$; i.e., none of the molecules are lost at the output.
In this case, $N=M$, and $S^M$ is chosen uniformly at random from all permutations of $[1,\dots,M]$.
If we naively bound each term separately, we obtain 
\aln{
H\left(Y^{ML}\right) + H\left(S^N|X^{\ML},Y^{ML}\right) \leq \ML + \M \log \M. % = \ML (1+1/\beta).
}
However, intuitively, the bound $H\left(S^M|X^{\ML},Y^{ML}\right) \leq \M \log \M$ is too loose because, as we argue below, 
if the entropy term $H\left(Y^{ML}\right)$ is large then we expect $H\left(S^M|X^{\ML},Y^{\NL}\right)$ to be small and vice versa.

To see this, first note that given $X^{\ML} = x^{\ML}$ and $Y^{ML} = y^{ML}$, one can estimate the permutation $S$ that maps each output string to the corresponding input string, $S^M$, by finding, for each $y_i^\len$, the $x_j^\len$ that is closest to it and setting $S(i) = j$.
This is a good estimate if no other $x_k^\len$ is close to $x_j^\len$. 
There are two regimes, illustrated in Figure~\ref{fig:points}, one where $S^N$ can be estimated well and one where it cannot.
\begin{figure}[t] 
%\vspace{2mm}
	\center
       \includegraphics[width=0.65\linewidth]{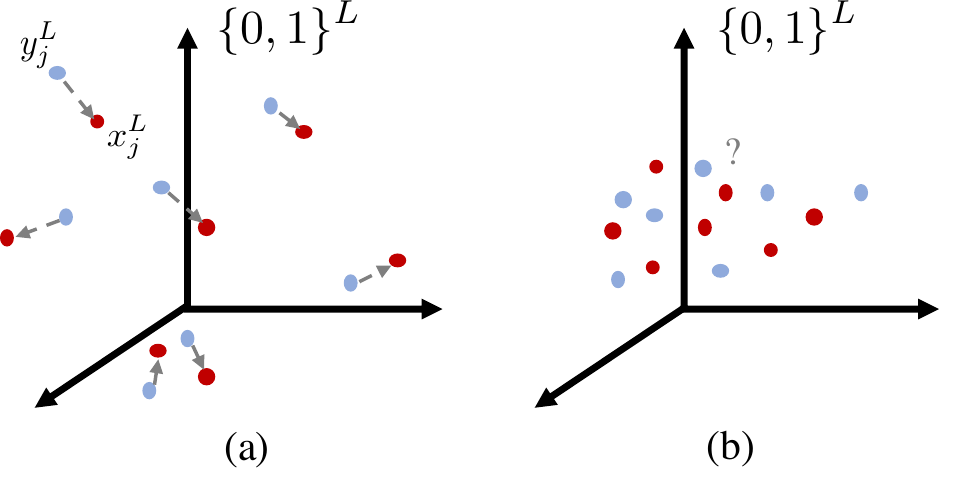} 
              \vspace{-3mm}
       \caption{Two opposite scenarios for estimating $S^N$ from $\left(X^{\ML},Y^{\NL}\right)$.\label{fig:points}}
\vspace{-2mm}
\end{figure}
In the first regime, the strings $x_1^\len,\ldots,x_\M^\len$ are all sufficiently distant from each other (in the Hamming sense).
Hence, the maximum likelihood estimate of $S^N$ given $X^{\ML} = x^{\ML}$ and $Y^{\NL} = y^{\ML}$ is ``close'' to the truth and we expect 
$H\left(S^N|X^{\ML} = x^{\ML},Y^{\NL} = y^{\ML}\right)$ to be small.
In the second regime, illustrated in Fig.~\ref{fig:points}(b), many of the sequences $x_1^\len,\ldots,x_\M^\len$ are close to each other.
So we have less information about $S^N$, and $H\left(S^N|X^{\ML} = x^{\ML},Y^{\NL} = y^{\ML}\right)$ may be large.

On the other hand, the term $H\left(Y^{\NL}\right)$ is maximized if the sequences $\{X_i^\len\}$ are independent and if their values are uniformly distributed in $\{0,1\}^\len$.
Hence, in order for $H\left(Y^{\NL}\right)$ to be large, we expect to be in the regime in Fig.~\ref{fig:points}(a) instead of the regime of Fig.~\ref{fig:points}(b).
This leads to a tradeoff of the terms $H\left(Y^{\NL}\right)$ and $H\left(S^N|X^{\ML},Y^{\NL}\right)$, which we exploit to prove Lemma~\ref{lem1}.

\section{Different noise models}

% The result in Theorem~\ref{thm:bsc} can be extended to different noisy channel models.
Given that the capacity expression for the noisy shuffling-sampling channel given in Theorem~\ref{thm:bsc} is $(1-q)(\Cbsc -1/\beta)$, where $\Cbsc = 1-H(p)$ is the capacity of a BSC,
it is natural to ask whether for a different noisy channel with capacity $C_{\text{noisy}}$, the corresponding noisy shuffling-sampling channel has capacity $(1-q)(C_{\text{noisy}}-1/\beta)$.
Notice that, when the sampling distribution $Q$ is $\Bernoulli(1-q)$, the index-based achievability scheme described in Section~\ref{sec:bsc_thm} can be extended in a straightforward way to achieve any rate below $(1-q)(C_{\text{noisy}}-1/\beta)$.
Based on this, 
it is natural to conjecture that the capacity of a noisy shuffling-sampling channel with sampling distribution  $\Bernoulli(1-q)$ and noisy channel with capacity $C_{\text{noisy}}$ is given by
\al{ \label{eq:conjec}
(1-q)(C_{\text{noisy}}-1/\beta)^+.
}
In Chapter~\ref{ch:discussion}, we present an even more general conjecture than this one.
% we have the following conjecture.
% \begin{conjecture}
% The capacity of a noisy shuffling-sampling channel with sampling distribution is $\Bernoulli(1-q)$ and noisy channel with capacity $C_{\text{noisy}}$ is given by
% \al{ \label{eq:conjec}
% (1-q)(C_{\text{noisy}}-1/\beta)^+.
% }
% \end{conjecture}
Since achieving the rate in (\ref{eq:conjec}) is straightforward, the challenging technical question is whether the converse argument in Section~\ref{sec:converse} can be generalized.
Next we discuss specific noisy channel cases where this conjecture can be proved.

% Hence, the 
% challenging technical question is whether the converse argument can be generalized.

\subsection{BEC Shuffling-Sampling Channel} \label{sec:bec}

Consider the setting studied in Section~\ref{sec:bsc_thm} except that, instead of a BSC, we have a binary erasure channel (BEC) with erasure probability $p$.
While erasures are not observed in sequencing data in practice, they can be practically motivated by the fact that sequencing technologies often provide a quality score for each base  \citep{bokulich2013quality}.
Such score can in principle be thresholded to identify reliable base calls and unreliable ones, which could then be treated as erasures.

As shown in \citet{shin2020capacity}, ideas similar to those used in the converse in the BSC case in Section~\ref{sec:converse} can be used in the BEC case as well.
While the result in \citet{shin2020capacity} is stated for the case of perfect sampling ($q_1 = 1$), it can be generalized to the case of $\Bernoulli(1-q)$ sampling using the ideas from Section~\ref{sec:converse}, yielding the following result.
\begin{theorem}
\label{thm:bec}
The capacity of the BEC shuffling-sampling channel is
\al{
C = (1-q)(1 - p - 1/\beta), \label{eq:beccapacity}
}
as long as $1-2p-  2/\beta > 0$.
If $\beta \leq 1$, then the capacity is $C = 0$.
\end{theorem}

Notice that, similar to the BSC case, the theorem only holds for a specific regime of $p$ and $\beta$.
However, the simplicity of the the erasure setting allows us to establish the converse for a larger set of parameters.
In particular, while Theorem~\ref{thm:bsc} holds for $p<1/4$ and $1-H(2p)-1/\beta > 0$, Theorem~\ref{thm:bec} holds for the strictly larger regime $p<1/2$ and $1-2p-1/\beta > 0$.

\subsection{General Noisy Channels}
\label{sec:generalchannels}

Given that (\ref{eq:conjec}) holds (for some parameter regime) for both the BSC and the BEC, it is natural to attempt to generalize the result to more general discrete memoryless channels.
To understand the challenges of this general setting, consider the simpler case of ideal sampling, i.e., 
% A simpler question is whether, for the case where 
$N_i=1$ with probability $1$ for $i=1,...,M$.
% , the general capacity expression is $C_{\text{noisy}}-1/\beta$.
Suppose we have an arbitrary, not necessarily memoryless, channel $p(y^L|x^L)$ that maps length-$L$ input strings to length-$L$ output strings (which may not be memoryless) and capacity $C_{\text{noisy}}$ (which requires the channel to be defined for $L \to \infty$).
Consider the corresponding noisy shuffling  channel.
The index-based scheme achieves any rate $R < C_{\text{noisy}}-1/\beta$.
% If we build a noisy shuffling channel based on $p(y^L|x^L)$, it is straightforward to see that the index-based scheme can always be used and achieves rate
% $R = C_{\text{noisy}}-1/\beta$.
However, extending the proof in Section~\ref{sec:converse} 
to establish $C_{\text{noisy}}-1/\beta$ as the capacity is challenging.
% for an arbitrary channel is challenging.
Following similar steps to those in (\ref{eq:bound1}), 
\aln{
& ML(R-\epsilon_{\M}) = I\left(X^{\ML};Y^{\ML}\right) \\
& =
H\left(Y^{\NL}\right) - H\left(S^M,Y^{\ML}|X^{\ML}\right)  + H\left(S^M|X^{\ML},Y^{\ML}\right) \\
& =
H\left(Y^{\NL}\right)-H\left(Y^{\ML}|X^{\ML},S^M\right) - H\left(S^M\right)  
  + H\left(S^M|X^{\ML},Y^{\ML}\right) \\
% -  H\left(Y^{\NL}| X^{\ML} \right) \nonumber \\
% & = H\left(Y^{\NL}\right) -  H\left(S^N,Z^{\NL},Y^{\NL} | X^{\ML}\right)
% % \nonumber \\ & \quad 
% + H\left(S^N,Z^{\NL}|X^{\ML},Y^{\NL}\right) \nonumber \\ 
% & = H\left(Y^{\NL}\right) -  H\left(S^N,Z^{\NL},Y^{\NL} | X^{\ML}\right)
% % \nonumber \\ & \quad
%  + H\left(S^N|X^{\ML},Y^{\NL}\right) 
& =
I\left(X^{\ML},S^M;Y^{\ML}\right)  - H\left(S^M\right)  
  + H\left(S^M|X^{\ML},Y^{\ML}\right).
}
Since $H(S^M) = M \log M + o(ML) = ML/\beta + o(ML)$, an outer bound to the noisy shuffling channel capacity in this general case is
\al{
C \leq 
\lim_{M\to \infty} \sup_{p(x^{ML})} &  \frac{I\left(X^{\ML},S^M;Y^{\ML}\right)}{ML}  + \frac{ 
   H\left(S^M |X^{\ML},Y^{\ML}  \right)}{ML}
  - 1/\beta. \label{eq:sup}
}
The main challenge in establishing a general converse is the optimization over distributions of the channel input $X^{ML}$.
Intuitively, for ``well-behaved'' channels, choosing the input distribution $p(x^{ML})$ that maximizes the first term, $I\left(X^{\ML},S^M;Y^{\ML}\right)$, causes the output strings $Y^{L}_i$, $i=1,\dots,M$, to be spread out in the output space, making the second term, $H\left(S^M |X^{\ML},Y^{\ML}  \right)$, small.
In Section~\ref{sec:converse}, we explained how the tension between these two terms can be exploited to jointly bound them.
The same idea is used in \citet{shin2020capacity} in the case of the BEC.
In both cases, this joint bounding allows us to show that (\ref{eq:sup}) is optimized by choosing $p(x^{ML})$ to be $ML$ i.i.d.~$\Bernoulli(1/2)$ random variables.

Notice that, if we know that the optimal distribution in (\ref{eq:sup}) satisfies  $p(x^{ML}) = p(x^L)\times \dots \times p(x^L)$ (i.e., independently encoding each of the input strings with the same $p(x^L)$), then the first term in the optimization becomes
\aln{
\frac{I\left(X^{\ML},S^M;Y^{\ML}\right)}{ML} & = 
\frac{H\left(Y^{\ML}\right)-\left(Y^{\ML}|X^{\ML},S^M\right)}{ML} 
\\
& = 
\frac{\sum_{i=1}^M H\left(Y_i^{L}\right)-\left(Y_i^{L}|X_{S(i)}^L\right)}{ML}
\\ 
& 
= \frac{I(X^L;Y^L)}{L},
}
and by choosing the distribution $p(x^L)$ that achieves the capacity of the noisy channel $p(y^L|x^L)$, this term becomes $C_{\text{noisy}}$.
However, establishing that this is the optimal input distribution for general channels remains an open question.
We refer to Section~\ref{sec:merhav} for additional discussion on extensions to general discrete memoryless channels, but in the multi-draw setting.

\section{Index-based coding and independent decoding}
\label{sec:index}

All achievability schemes discussed in Chapters~\ref{ch:shuffled}~and~\ref{ch:noisy} are based on the same approach, which is illustrated in Figure~\ref{fig:innerouter}.
In particular, we highlight two key elements of the achievability schemes:
\begin{enumerate}
    \item \emph{Index-based coding:} \, Distinct indices are placed in each stored sequence (prior to the encoding with an error-correcting code).
    At the decoder, the indices are used to order the input sequences and to identify missing ones.
    \item \emph{Independent decoding: } \, Each of the stored sequences is a codeword from an error-correcting code, and is independently decoded at the output (possibly followed by the decoding of an outer code).
\end{enumerate}
Notice that both index-based coding and independent decoding are highly desirable from a computational standpoint.
However, from a capacity standpoint it is surprising that optimal codes exhibit these desirable features.

% In principle, it should be possible to reduce the overall error probability of a code by considering a joint, index-free code.

\paragraph{Is index-based coding always optimal?}

A key question that naturally arises from the results in Chapters~\ref{ch:shuffled}~and~\ref{ch:noisy} is whether index-based coding is always optimal.
In the parameter regime in the gray region of Figure~\ref{fig:curves}, the capacity remains unknown and it is possible to achieve rates higher than the $(1-q)(1-H(p)-1/\beta)$ that is achieved by index-based schemes.
Notice that, this region corresponds to a high-noise short-molecule regime, where rates are expected to be low.
In a short-molecule regime, it may be reasonable that dedicating a significant fraction of each molecule for an index is wasteful from a data rate standpoint.
Furthermore, in the high-noise regime, a significant amount of redundancy needs to be added to guarantee that the indices are decoded error-free.
Hence, it may be possible to devise a scheme that uses less redundancy for indices, but the decoder only decodes the indices approximately, which may be enough to order the sequences once all indices are considered jointly.
% In principle, one could devise a coding scheme whose decoder directly maps the set of output sequences $\{\y_1,...,\y_N\}$ to the encoded message.
Furthermore, as we will show in the next chapter, once we move away from the single-draw setting, and allow multiple copies of each input sequence to be observed at the output with independent noise patterns, index-based schemes are no longer optimal, and ``global''  decoding schemes are needed to achieve the noisy shuffling-sampling channel capacity.

A similar question can be posed regarding the optimality of independent decoding.
Such a question was in fact studied in detail in the context of the \emph{bee identification problem} \citep{bee}.

\subsection{Independent decoding and the bee identification problem}
\label{sec:bee}

% Bee identification problem deals with just the unshuffling operations.
% In principle, it should be possible to reduce the overall error probability of a code by considering a joint, index-free code.

The bee-identification problem is motivated by the task where one observes a noisy picture of a massive number of bees, uniquely labeled with barcodes, and wishes to simultaneously identify all bees.
This can be modeled as a problem in which 
% In the bee-identification problem, 
a single list of $M$ strings of length $L$, $[\x_1,...,\x_M]$, is passed through a noisy shuffling channel (\citet{bee} consider the BSC shuffling channel described in Section~\ref{sec:bsc_thm}).
The goal is to find the correct matching $\sigma$ between the output sequences $[\y_1,...,\y_M]$ and the input sequences; i.e., finding the shuffling induced by the channel on the set of strings, as illustrated in Figure~\ref{fig:bee}.

\begin{figure}[t]
     \centering
     \hspace{-10mm}
      \includegraphics[width=0.5\linewidth]{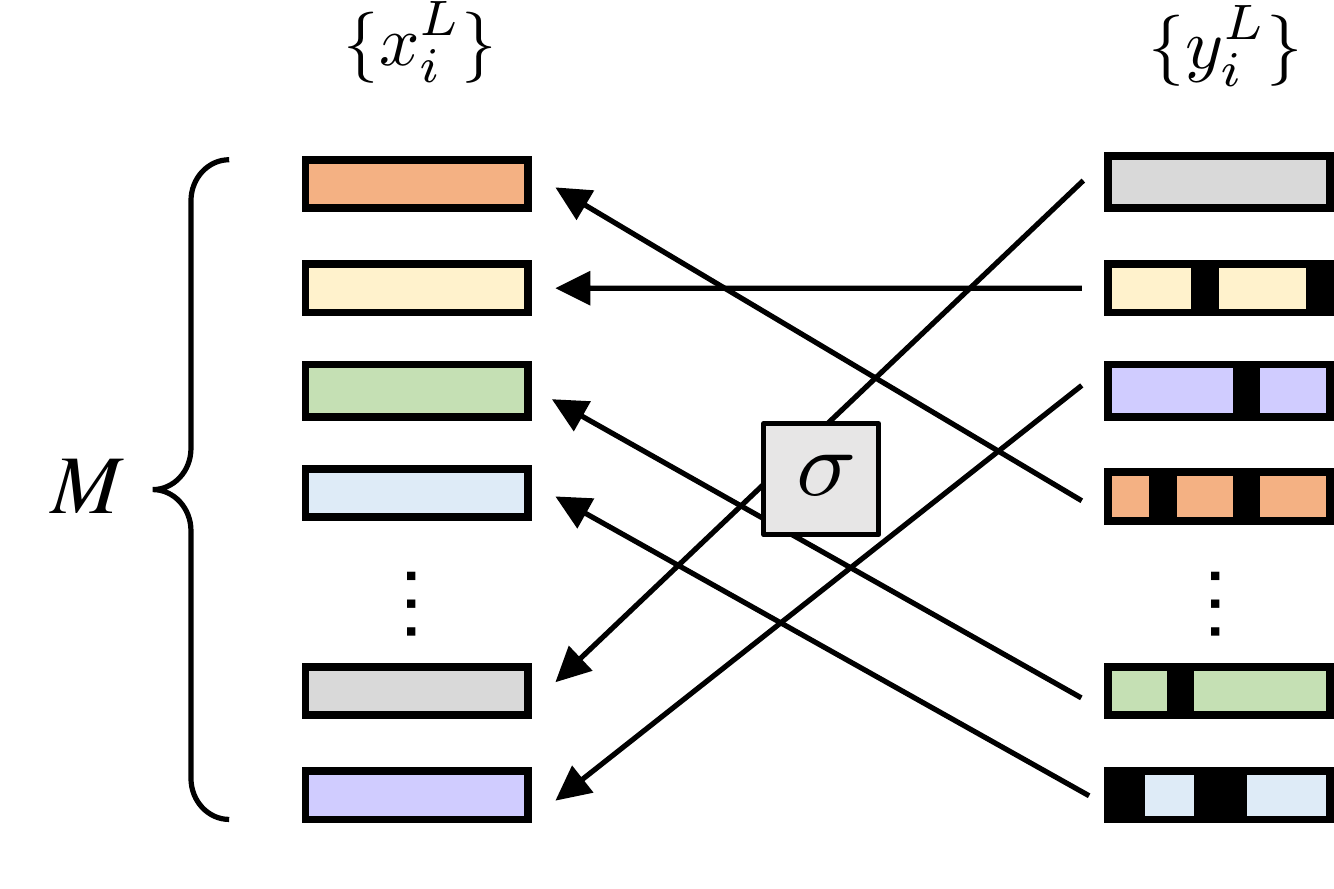}
     \caption{In the bee identification problem, we observe a set of output sequences $\{\y_1,...,\y_M\}$ and want to find a mapping $\sigma$ such that $\x_{\sigma(i)}$ was the input sequence that resulted in $\y_i$ after the noisy shuffling channel.
     \label{fig:bee}}
\end{figure}

A key contribution of \citet{bee} is to find upper and lower bounds for the error exponent of the bee identification problem under two decoding approaches: independent decoding and joint decoding.
In independent decoding, only the $i$th output sequence $\y_i$ is used to determine $\sigma(i)$, while in joint decoding, the entire set of output sequences $\{\y_1,...,\y_M\}$ is used to determine $\sigma$.
They show that the error exponent for joint decoding is strictly larger than the error exponent for independent decoding.
They also compare the error exponents of random code ensembles and typical random codes and show that typical random codes provide a significant improvement.
Since then, follow-up works have refined these error exponents, considered different scenarios in terms of what constitutes an error when recovering $\sigma$ and whether all bees are observed at the output
\citep{bee2,tamir_bee}, and have proposed efficient ways to perform the joint decoding~\citep{kiah_efficient_2021}.

In the context of DNA-based storage, the goal is not to recover the permutation $\sigma$, but rather to decode the message encoded by the shuffled set of sequences.
However, all capacity-achieving schemes discussed in Chapters~\ref{ch:shuffled}~and~\ref{ch:noisy} place unique indices on every sequence, with the goal of allowing the decoder to recover the correct permutation $\sigma$ (which will then allow the message to be recovered).
Hence, the results from \citet{bee} suggest that, also in the context of DNA storage, joint decoding can improve the error exponent with respect to independent decoding.

Nevertheless, it is shown in \citet{bee2} that, in the setting with absentee bees (i.e., when some bees are not observed at the output), the advantage of joint decoding with respect to independent decoding disappears, and the error exponents achieved are the same.
This suggests that the advantage of joint decoding may depend on the unrealistic assumption that all sequences are observed at the output.
As discussed throughout this monograph, this is not realistic based on current DNA data storage prototypes, suggesting that the potential gains of joint decoding may be minor, particularly given the significant computational costs of joint decoding.

\section{Designing practical codes for DNA data storage}
\label{sec:noisyadditional}

In this chapter, we investigated the fundamental limits of noisy shuffling channels.
While we considered standard noisy channels such as the BSC and the BEC, the design of practical codes for DNA storage systems must be tailored to the technologies being employed.
% Information stored in DNA is subject to errors that arise 
% during writing, storing as well as during reading. 
Several recent works have focused on designing error-correcting codes for the specific types of errors that may arise during writing, storing and reading data on DNA.
% have studied important additional aspects of the design of a practical DNA storage system.
Some important aspects addressed in the recent literature include DNA synthesis constraints such as sequence composition  \citep{kiah_codes_2016,yazdi_rewritable_2015,erlich_dna_2016,kovavcevic2019runlength}, the asymmetric nature of the DNA sequencing error channel \citep{gabrys_asymmetric_2015}, 
the need for codes that correct insertions and deletions \citep{sala_insertions_2016,kovavcevic2018asymptotically,chee2019linear,kiah2020coding,chrisnata2020optimal},
the need for codes that avoid homopolymers (consecutive repeated nucleotides) as they are conducive to sequencing errors \citep{erlich_dna_2016,kovavcevic2019runlength,kiah_codes_2016}, 
and the need for codes with balanced GC-content (i.e., roughly the same number of $\bG$/$\bC$ as $\bA$/$\bT$) \citep{chee2019efficient}.
An LDPC-based joint design for the inner and outer codes needed for DNA data storage was proposed in \citet{chandak2019improved}.

The shuffling aspect of DNA storage systems has also been addressed from a coding-theoretic perspective.
In particular, the problem of coding over sets (or multi-sets) has been studied \citep{kovavcevic2018codes,lenz2019codingoversets,song2020sequence,sima2019coding,wei_schwartz_2021,song_sequence_2020,Sima_Bruck_2021,Sima_Raviv_Bruck_2020}, 
with the goal of characterizing basic properties such as minimum distance and maximum size of codes in the space of multi-sets, and with the goal of providing bounds on code parameters. 
We point out that the shuffling channel is also connected with a \emph{permutation channel} that permutes the symbols in a message \citep{ahlswede1987optimal,benjamin1975coding}, which recently received renewed attention due to the emergence of DNA-based data storage~\citep{makur2018information}.

The problem of designing indices that allow efficient retrieval of data stored on DNA has also been investigated \citep{lenz2019codingoversets}. 
In particular, \citet{lenz2019anchor} explores the idea of adding 
``anchor sequences'' to the indices, with the goal of reducing the total amount of redundancy needed for the indices.
Furthermore, the careful design of indices that allow \emph{random access} via DNA \emph{hybridization} has also been explored \citep{yazdi_rewritable_2015,yazdi_portable_2017,organick_scaling_2017,chee2019efficient}.

\chapter[Multi-draw Channels]{Multi-draw Channels: Clustering Output Sequences}
\label{ch:multi}

In Chapter~\ref{ch:noisy}, we studied noisy shuffling-sampling channels with single draws.
More precisely, each sequence in the DNA library can either be drawn once or not at all.
However, as shown in Figure~\ref{fig:generalchannel}, a more complete model for DNA-based storage systems  should incorporate the fact that each DNA sequence can be sequenced multiple times.
This multi-draw setting arises because synthesis technologies generate multiple copies in the first place, and PCR used at the time of sequencing multiplies the number of copies further, which results in 
%due to the fact that PCR is used at the time of data writing in order to create 
a large number of copies of each DNA molecule at the output.
\begin{figure}
% \vspace{3mm}
     \centering{
      \includegraphics[width=0.8\linewidth]{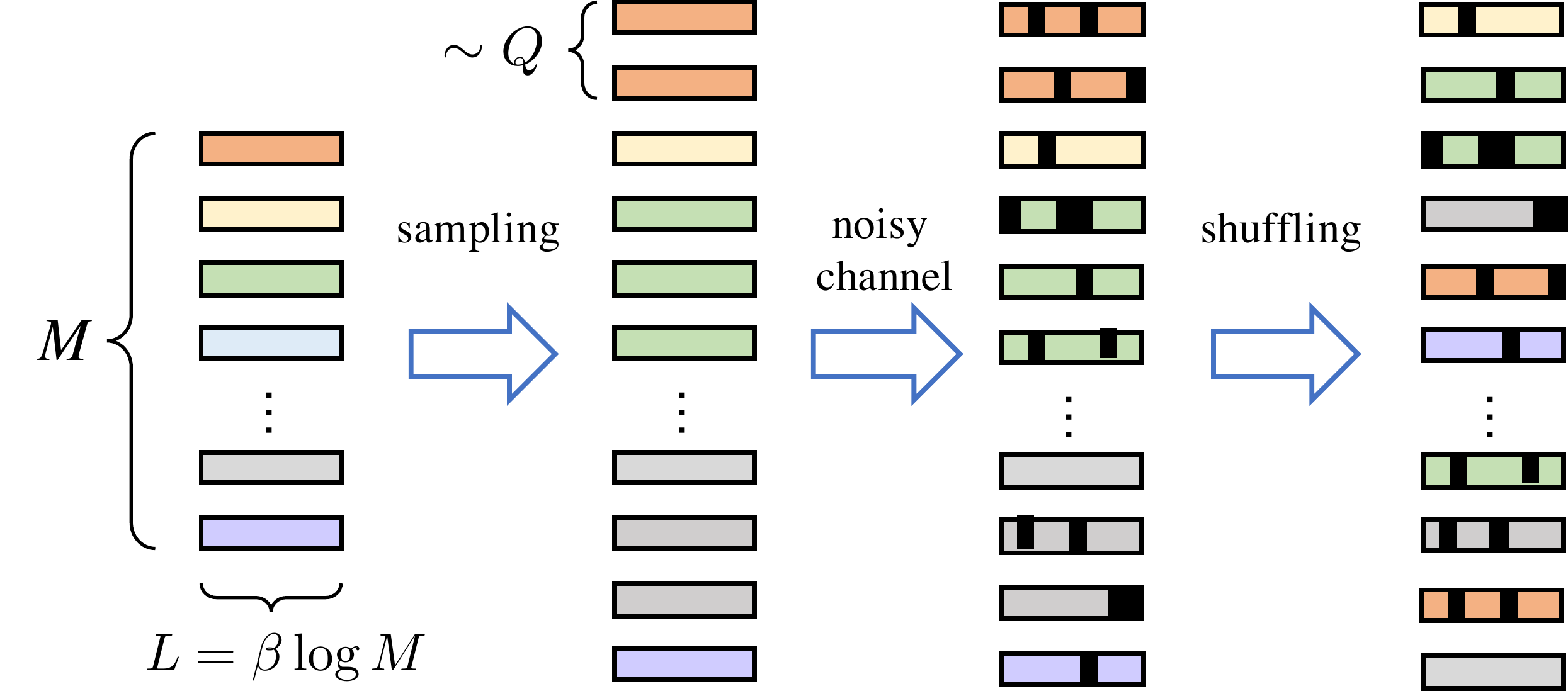}}
      \vspace{2mm}
     \caption{When the sampling distribution $Q$ (see Figure~\ref{fig:generalchannel}) allows more than one copy of each input sequence to be created (e.g., $Q$ is $\Pois(\lambda)$), the output sequences can intuitively be clustered according to the input sequence of origin.
     \label{fig:clustering}}
\end{figure}
% Each of these copies is subject to independent corruption by noise (during synthesis, storage, and sequencing).
As illustrated in Figure~\ref{fig:clustering}, at the output of the DNA storage channel, one may thus observe several different noisy copies of each of the input sequences, and a ``clustering problem'' arises from the need to identify which of these sequences correspond to the same input sequence.

In the multi-draw setting, establishing the capacity of noisy shuffling-sampling channels is significantly more challenging.
Intuitively, the presence of multiple noisy copies of the same string should increase the capacity.
More precisely, if we are able to correctly cluster the output sequences, it is possible to combine the sequences in each cluster and perform error correction, effectively reducing the error rate of the system.
The following questions arise. 
Is it still optimal to utilize an index-based scheme?
% The indices could in principle help in the clustering task, but 
% And what happens i
Is it information-theoretically optimal to first cluster the output sequences and then attempt to decode the message?

We begin our discussion on this problem by focusing on the setting where each sequence is passed through a Binary Erasure Channel (BEC), in Section~\ref{sec:multidraw}. 
It turns out that erasures are easier to treat than subsitutions or insertions and deletions, and this simplicity enables us to gain basic insights into the nature of the clustering problem that arises in the multi-draw context. 
In addition, we discuss connections with the trace reconstruction problem in Section~\ref{sec:trace}, and coding-theoretic considerations in Section~\ref{sec:multiadditional}.

\section{Noisy shuffling-sampling channel with multi-draws}
\label{sec:multidraw}

As a first step to studying the capacity of a noisy shuffling-sampling channel with multi-draws, we consider a noise model described by a binary erasure channel (BEC), studied by~\cite{levick2021achieving}.
As it turns out, the simplicity of the BEC (relative to the binary symmetric channel (BSC)) allows us to utilize a linear coding scheme to achieve the capacity of the system.
Moreover, the achievability proofs are intuitive and based on simple equation-counting arguments.

We consider the general DNA storage channel illustrated in Figure~\ref{fig:generalchannel}, where the sampling distribution $Q$ is arbitrary, and the memoryless channel $p(y|x)$ is a BEC, i.e., each symbol is erased with probability $p$ (leading to an erasure symbol $\eras$).
Recall that the capacity of the BEC is $1-p$ and, as discussed in Section~\ref{sec:bec}, the capacity of the BEC shuffling-sampling channel with single draws (with a probability $q$ of not seeing a sequence) is given by
\al{
C_{\text{BEC, single-draw}} = (1-q) (1-p -1/\beta).
}

Now consider instead an arbitrary sampling distribution $Q$, with $q_n = \Pr(N_i = n)$.
% Let $\bar q = E[N_i] = \sum_{n=0}^\infty n \, q_n$.
Notice that, in expectation, a fraction $q_0$ of the input sequences is not sampled, and thus we expect a total of $(1-q_0) M$ output clusters.

Let us first suppose we have a genie-aided channel that provides us with the correct clustering of the output strings.
In this case, it is intuitive that the optimal decoding scheme starts by combining all the sequences in each cluster into a single ``consensus'' sequence.
The consensus sequence (of length $L$) is obtained by taking, for the $i$th position, the $i$th symbol of any sequence in the cluster that is not erased.
For an output cluster with $n$ sequences, the resulting effective erasure rate is $p^n$.
Therefore, after the consensus step, in expectation we have $(1-q_0)M$ output sequences and, for $n=1,2,\ldots$, we expect that $q_n M$ of them have an erasure rate of $p^n$.

Notice that the resulting effective channel after the consensus step is similar to the BEC shuffling-sampling channel discussed in Section~\ref{sec:bec} but, instead of a single erasure probability $p$ across all sequences drawn at the output, the $(1-q_0)M$ output strings experience an average erasure probability of
\al{ \label{eq:peff}
\peff
\defeq
E[p^{N_1}| N_1 \geq 1]
=
\frac{\sum_{n=1}^\infty q_n p^n}{1-q_0}
=
\frac{\sum_{n=1}^\infty q_n p^n}{\sum_{n=1}^\infty q_n}
}
Since the capacity of this genie-aided channel (which has $(1-q_0)M$ output sequences) is larger than the capacity of the BEC shuffling-sampling channel, we can use the converse approach from Section~\ref{sec:converse} (see also \citet{shin2020capacity}) to show that the capacity of the BEC shuffling-sampling channel satisfies
\al{\label{eq:multiconverse}
C_{\text{BEC, multi-draw}} \leq (1-q_0) \left( 1 - \peff - 1/\beta \right),
}
as long as $p$ and $\beta$ satisfy $1-2p - 2/\beta > 0$. 
We point out that, when the converse approach from Section~\ref{sec:converse} is  applied in this multi-draw context, the condition $1-2p - 2/\beta > 0$ depends on $p$, and not $\peff$, but a more careful analysis may be able to relax this condition.

The main result in this section shows that, for $\beta$ large enough, this upper bound can be achieved.

\begin{theorem} \label{thm:multicapacity}
The capacity of the BEC shuffling-sampling channel is
\al{\label{eq:multicapacity}
C_{\text{BEC, multi-draw}} = (1-q_0) \left( 1 - \peff - 1/\beta \right),
}
as long as $1-2p - 2/\beta > 0$. 
\end{theorem}

\subsection{Achievability via linear schemes} \label{sec:linear}

Unlike in Chapters~\ref{ch:shuffled}~and~\ref{ch:noisy}, to achieve the capacity in this multi-draw setting, we do not utilize an index-based scheme.
Instead, we take advantage of the fact that random linear codes are known to achieve the capacity of the BEC \citep{elias1955coding}, and show that linear schemes also achieve the capacity of the BEC shuffling-sampling channel.

We construct our linear coding scheme with rate $R$ as follows.
For a parameter $B$ to be determined, we first generate a random binary $ML \times B$ matrix $\matG$ with i.i.d.~Bernoulli$(1/2)$ entries.
% \iscomment{Maybe we should choose $G$ to be an non-random invertible matrix, rather than pick the entries i.i.d.}
% \rh{Don't we need this to be random so that each subset of rows/columns also has full rank?} IS: added lemma
% \rh{for notation: should we have matrices and vectors boldface to distinguish them from integers that we often denote by capital letters?} IS: Changed this.
We then generate $2^{MLR}$ random binary vectors $\bt_i$ of size $B$ (also with i.i.d.~$\Bernoulli(1/2)$ entries).
For $i=1,\ldots,2^{MLR}$, the $i$th codeword is generated by computing $\matG\, \bt_i$ (over $\F_2$), and then breaking the resulting length-$ML$  vector into $M$ binary strings of length $L$ (which serve as channel input).
We point out that this is technically not a \emph{linear code} as the set of codewords is not the entire range of $\matG$ (only the vectors $\matG \bt_i$ for $\bt_1,\dots,\bt_{2^{MLR}}$) and hence the code is not a linear subspace of $\{0,1\}^{ML}$.

A key property of $\matG$ that we need is the following.

\begin{lemma} \label{lem:matrixG}
Let $\matG$ be an $ML \times B$ matrix with  i.i.d.~$\Bernoulli(1/2)$ entries.
Fix any $\delta \in  (0,1)$ and a submatrix $\matG'$ formed by an arbitrary set of $(1-\delta)B$ rows of $\matG$.
Then $\matG'$ is full rank (over the finite field $\F_2$) with probability tending to $1$ as $B \to \infty$.
\end{lemma}

While a version of this lemma for a random matrix over the reals is relatively straightforward, for the case of $\F_2$ it requires some work to prove.
We present the proof of Lemma~\ref{lem:matrixG} in Section~\ref{sec:lemproof}.

\subsubsection{Single-draw case}

To illustrate this linear non-index-based scheme, let us first consider the case where each sequence is sampled exactly once (i.e., $q_1 = 1$).
At the output, we observe $M$ sequences, which we would like to use for recovering the vector $\bt_i$, or simply for recovering the message index $i$.

Suppose we knew the correct ordering of the $M$ output strings. 
Then we could concatenate them into a single string $\by$ with length $ML$, and then try to solve the system $\matG \, \bt = \by$.
Since a fraction $p$ of the entries of $\by$ would be erased, we would be able to solve this system (with high probability) as long as the number of remaining equations, which is roughly $(1-p)ML$, is greater than or equal to the number of variables $B$.
We would then be able to set $B = (1-p -\ep) ML$ for any $\ep > 0$. 
We would also be able to choose all $2^B$ binary strings in $\{0,1\}^B$ to be $\bt_1,\ldots,\bt_{2^B}$, instead of choosing them randomly, which results in a code with rate of $R = 1- p -\ep$, which is the capacity of a standard BEC.

However, since in actuality we do not know the ordering of the $M$ output strings, we instead consider all $M!$ possible orderings of the output strings.
Each such ordering gives us a distinct concatenated vector $\by$ with a $p$ fraction of erasures and a corresponding set of useless rows in the matrix $\matG$, as illustrated in Figure~\ref{fig:system}.
\begin{figure}
\vspace{3mm}
     \centering{
      \includegraphics[width=0.63\linewidth]{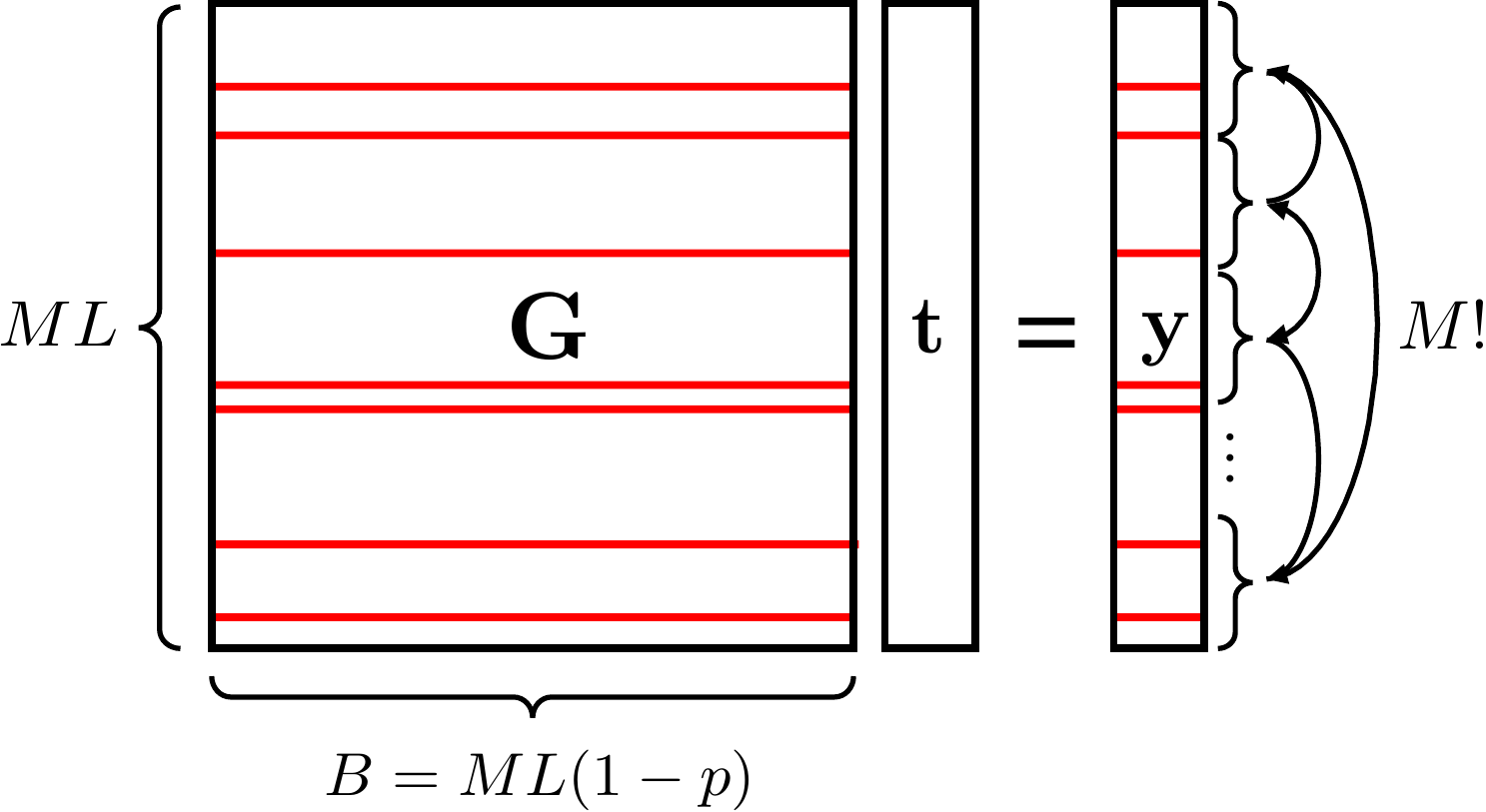}}
     \caption{In the setting where each sequence is observed exactly once at the output, there are $M!$ potential systems of equations, each with $ML(1-p)$ non-erased equations.
     Choosing our rate to be approximately $1-p-1/\beta$ guarantees that, with high probability, only one of these systems of equations will have a solution that corresponds to a codeword.
     \label{fig:system}}
\end{figure}
To guarantee that the matrix $\matG$ with erased rows is still invertible, we set $B = ML(1-p-\ep)$ for some small $\ep > 0$, and Lemma~\ref{lem:matrixG} guarantees that the true system has a unique solution $\bt$.
Moreover, to guarantee correct decoding, we must ensure that only one of the $M!$ systems (the true one) admits a solution $\bt$ that is one of the original vectors $\bt_i$ for $i=1,\ldots,2^{MLR}$.

Assume without loss of generality that message $1$ is sent (i.e., the codeword defined by $\matG\, \bt_1$).
Since all other vectors $\bt_i$, $i \ne 1$, are chosen uniformly at random from $\{0,1\}^B$, by the union bound, the probability that one of the $M! - 1$ incorrect systems has a solution $\bt$ that coincides with some $\bt_i$, $i=2,\ldots,2^{MLR}$, is at most
\aln{
(M! - 1) 2^{MLR} 2^{-B} < 2^{M \log M + MLR - B} = 2^{ML[ R - (1-p-\ep  -1/\beta)]}.
}
By choosing $\ep$ arbitrarily small, we conclude that any rate $R < 1-p - 1/\beta$ can be achieved with a vanishingly small error probability.

While this result simply recovers what the index-based approach in Chapter~\ref{ch:noisy} already achieved, it can be extended to the case of a general sampling distribution $Q$ as we see next.

\subsubsection{Multi-draw case}

Let us now consider the case of a general sampling distribution $Q$.
At the output of the channel, we expect to see roughly $E[N_1]M$ sequences, which we would like to partition into roughly $(1-q_0)M$ clusters.
A major advantage of the BEC setting is that it easier to reason about the clustering task than for the BSC setting, for example. 
Notice that all output strings originated from the same input string are \emph{consistent} with each other; i.e., they agree on any non-erased positions.
Hence, given $N$ output strings $\x_1,\ldots,\x_N$, one can in principle build a graph with $\x_1,\ldots,\x_N$ as vertices, where strings $\x_i$ and $\x_j$ are connected by an edge if they are consistent.
We refer to this graph as the \emph{consistency graph}.
A valid clustering of the output strings corresponds to any partition of $\{\x_1,\ldots,\x_N\}$ such that each group corresponds to a \emph{clique} in the consistency graph.

Since each input string has a probability $q_0$ of not producing an output cluster, in expectation we have $(1-q_0)M$ clusters.
Moreover, an application of Hoeffding's inequality implies that the probability that there are more than $(1-q_0+\ep)M$ or less than $(1-q_0-\ep)M$ true clusters at the output is bounded as
% \aln{
% \Pr\left(\text{\# true output clusters}>(1-q_0+\ep)M\right) < e^{-2M\ep^2},
% }
\aln{
\Pr\left( \left| \text{\# true output clusters} - (1-q_0)M \right| > \ep M \right) < 2e^{-2M\ep^2},
}
which tends to $0$ as $M \to \infty$ for any $\ep > 0$.
Therefore, our decoder looks for ways cluster the output strings into at least $(1-q_0-\ep)M$ clusters and at most $(1-q_0+\ep)M$ clusters, for some fixed small $\ep$.
More precisely, our decoder considers \emph{all} valid clusterings of the $N$ output strings into at most $(1-q_0+\ep)M$ and at least $(1-q_0-\ep)M$ clusters.
For each such clustering, a consensus step is performed, effectively converting the clusters into roughly $(1-q_0)M$ strings with a smaller overall erasure rate.
After this point, the decoder proceeds similarly to the case of single draws.
First, each cluster is assigned a distinct index from $\{1,\ldots,M\}$, which can be used to order the consensus strings into a 
% First, all possible orderings of the resulting strings will be considered.
% Each one will produce a 
vector $\by$ of length roughly $(1-q_0)ML$.
All possible index assignments are considered.
Each index assignment produces a vector $\by$ and a corresponding system of equations, which may yield a solution $\bt$ (provided that the system has a solution).
This clustering-based decoding is illustrated in Figure~\ref{fig:multi-clustering}.

Notice that we need to choose $B$ large enough to guarantee that the true system, obtained from the correct clustering and index assignment, has a unique solution, and $R$ small enough so that only one of the systems (the true one) yields a solution that  coincides with one of the vectors $\bt_i$ used to construct the codebook.
This allows the correct decoding of the message index.

\begin{figure}
% \vspace{3mm}
     \centering{
      \includegraphics[width=0.8\linewidth]{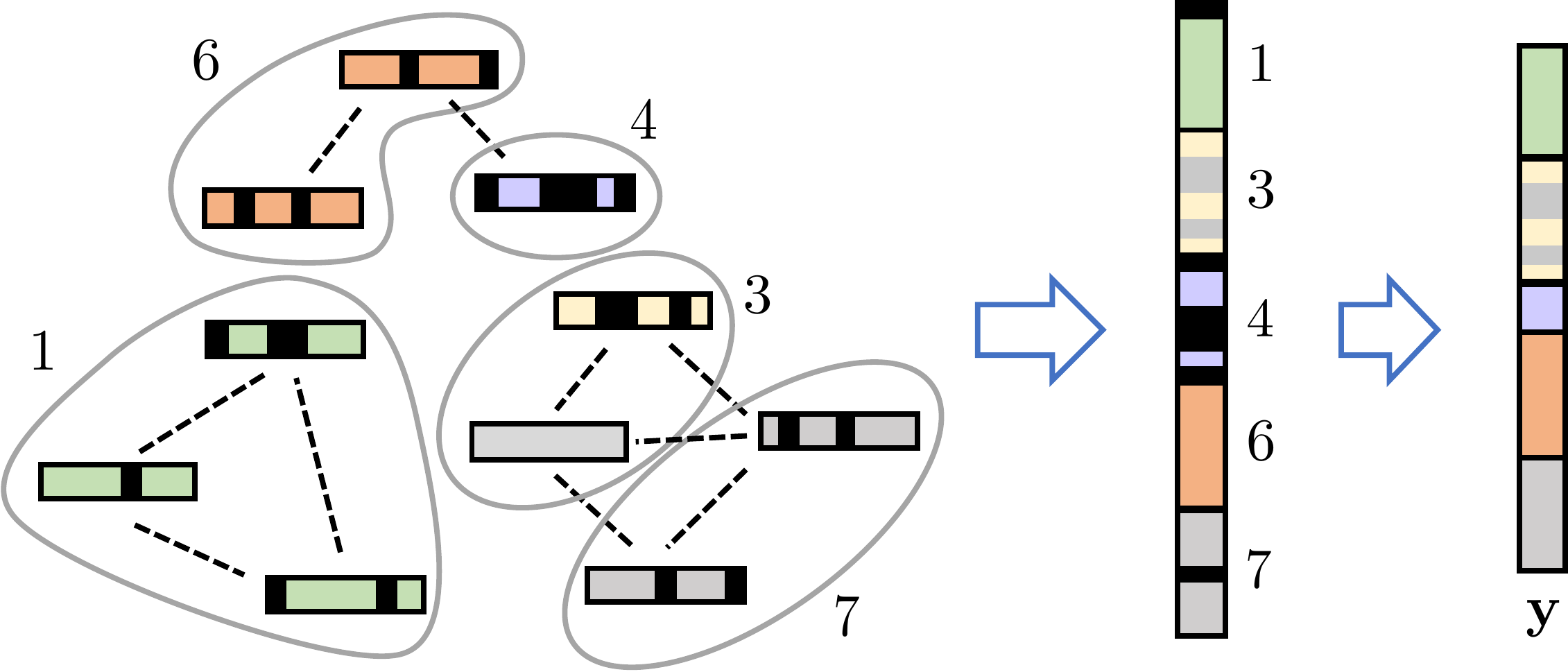}}
     \caption{The decoder first builds a consistency graph between all received sequences, and considers all possible partitions of the graph into cliques, where the total number of  cliques is between $(1-q_0 -\ep)M$ and $(1-q_0 + \ep)M$.
     For each such valid clustering, 
    %  a consensus sequence is build 
     the decoder considers all possible assignments of the indices $\{1,...,M\}$ to the clusters and uses those indices to order the consensus sequences of each cluster to form the vector 
    $\by$.
    After removing erased entries and the rows of $\matG$ corresponding to erased/missing rows in $\by$, the system $\matG \bt = \by$ can be solved.
     \label{fig:multi-clustering}}
\end{figure}

In order to guarantee that the correct $\bt_i$ can be recovered, we need to make sure that the true system has enough equations after the erasures have been discarded.
Once the consensus step is performed on the true clustering, we end up with at least $M(1-q_0-\ep)$ sequences.
% \iscomment{need to mention lower bound above}.
Moreover, the expected erasure rate is given by $\peff$ in (\ref{eq:peff}), and standard concentration inequalities can be used to show that it cannot deviate significantly from that.
Hence, the true system of equations has at least $ML(1-q_0-\ep)(1-\peff -\ep)$ equations with high probability, and if we set $B = ML(1-q_0-\ep)(1-\peff -\ep)(1-\ep)$, by Lemma~\ref{lem:matrixG}, the true system of equations has a unique solution with probability tending to $1$ as $M\to\infty$, which yields the correct solution $\bt_1$.

In order to find the maximum rate $R$ for which this scheme  succeeds with vanishing error probability, we need to bound the number of valid clusterings of the output strings.
To that end, we first analyze the total number of edges in the consistency graph.
Notice that, each true cluster of size $n$ produces ${n \choose 2} \leq n^2/2$ ``correct'' edges in the graph.
Hence, the expected number of correct edges is at most
\al{ \label{eq:correct}
\sum_{n=0}^\infty M q_n \frac{n^2}{2} = \frac{M}{2} E[N_1^2],
}
where $E[N_1^2]$ is the second moment of the sampling distribution $Q$.
% Standard concentration results can be used to 
Let $Z$ be the total number of incorrect edges in the consistency graph, and 
\aln{
\gamma \defeq - \beta \log \left( 1-\tfrac12(1-p)^2 \right),
}
which is positive for any $p \in (0,1)$.
% Then we have the following observation.
\begin{lemma} \label{lem:edges}
The number of incorrect edges $Z$ satisfies
\al{
% E[Z] = M^{2-\gamma} \quad \text{and} \quad
\Pr\left(Z > M^{2-\gamma + \ep}\right) \to 0}
as $M \to \infty$, for any $\ep > 0$. 
% \iscomment{lower bound too?}
\end{lemma} 

Lemma~\ref{lem:edges} (whose proof is presented in Section~\ref{sec:lemproof}) establishes that, as long as $\gamma > 1$, the number of incorrect edges grows slower than $M$ and, hence, will be vanishingly small compared to the number of correct edges in (\ref{eq:correct}).
Next, we bound the number of valid ways to cluster the output strings given the number of edges in a consistency graph.
As it turns out, a coarse bound suffices.
\begin{lemma} \label{lem:clusterings}
Suppose the consistency graph has a total of $U$ edges.
Then there are at most $2^U$ valid ways to cluster the output sequences.
\end{lemma}

\begin{proof}
Each valid clustering corresponds to a partition of the output sequences such that each group corresponds to a clique in the consistency graph.
Notice that a partition of the consistency graph into cliques is uniquely described by the set of edges that are part of each of the cliques.
Hence, each partition corresponds to a distinct subset of the $U$ edges in the graph.
Since there are at most $2^U$ such subsets, it follows that there are at most $2^U$ valid ways to cluster the output sequences.
\end{proof}
Now from equation~\eqref{eq:correct} and Lemma~\ref{lem:edges}, we know that for $\gamma > 1$, the number of edges $U$ in the consistency graph satisfies $U < \alpha M$ for some $\alpha > 0$, with high probability.
Lemma~\ref{lem:clusterings} thus implies that the number of valid ways to cluster the output strings is at most $2^{\alpha M}$.

The achievability proof is concluded by following the steps from the single-draw case.
\begin{figure}
\vspace{3mm}
     \centering{
      \includegraphics[width=0.8\linewidth]{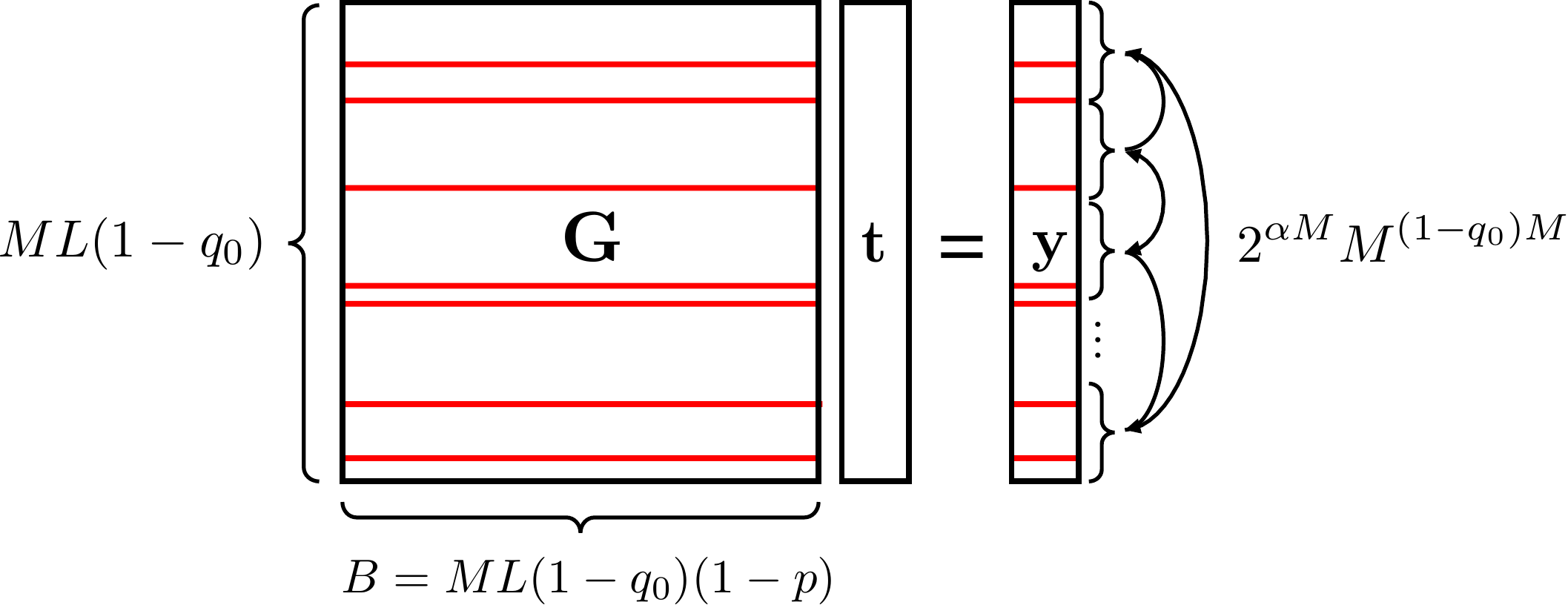}}
     \caption{In the multi-draw setting, there are $2^{\alpha M} M^{(1-q_0)M}$ potential systems of equations, obtained by considering all valid ways to cluster the output strings into $(1-q_0)M$ clusters, and then assigning each cluster to an index.
     The true system (corresponding to the correct clustering and ordering) has $ML(1-q_0)(1-\peff)$ non-erased equations in expectation.
    %  Choosing our rate to be $1-p-1/\beta$ guarantees that, with high probability, only one of these systems will have a solution that corresponds to a codeword.
     \label{fig:system_multi}}
\end{figure}
We need to guarantee that no system of equations other than the true one yields $\bt_i$, for some $i=2,\ldots,2^{MLR}$, as a solution.
The total number of systems of equations that the decoder will attempt to solve is upper-bounded by
\al{ \label{eq:numberofsystems}
2^{\alpha M} M^{(1-q_0+\ep)M},
}
because we first need to cluster all the output strings into at most $(1-q_0+\ep)M$ clusters in a valid way and then assign each of them an index in $\{1,\ldots,M\}$ for the ordering.
Therefore, if we assume that message $1$ is selected, the probability that any of the incorrect systems yields a solution $\bt$ that coincides with some $\bt_i$, for $i=2,\ldots,2^{MLR}$, is upper-bounded by
\al{ \label{eq:errormultidraw}
2^{\alpha M} M^{(1-q_0+\ep)M} 2^{MLR} 2^{-B} & = 2^{\alpha M + (1-q_0+\ep) M \log M + MLR - B} \nonumber \\
& = 2^{ML( \alpha/L + (1-q_0+\ep)/\beta + R - B/ML)}.
}
This upper bound goes to zero as $M \to \infty$ as long as
\aln{
R + \frac{\alpha}{\beta \log M} - (1-q_0-\ep)(1-\peff -\ep)(1-\ep) + \frac{1-q_0+\ep}{\beta} < 0.
}
Since $\alpha/\log M \to 0$ and $\ep$ can be chosen arbitrarily small, any rate
\aln{
R < (1-q_0)(1-\peff - 1/\beta)
}
is achievable.
This is formally stated next.
\begin{theorem} \label{thm:multiachievability}
The capacity of the BEC shuffling-sampling channel with multi-draws satisfies
\al{
C_{\text{BEC, multi-draw}} \geq (1-q_0) \left( 1 - \peff - 1/\beta \right),
}
as long as $\gamma = - \beta \log \left( 1-\tfrac12(1-p)^2 \right) > 1$. Here, $\peff$ is the effective erasure probability defined in equation~\eqref{eq:peff}.
\end{theorem}

In Figure~\ref{fig:regimes}, we compare the parameter regime in which the lower bound in Theorem~\ref{thm:multiachievability} holds with the regime in which the upper bound (\ref{eq:multiconverse}) holds.
We see that the lower bound constraint is strictly weaker than the upper bound constraint, which implies Theorem~\ref{thm:multicapacity}.

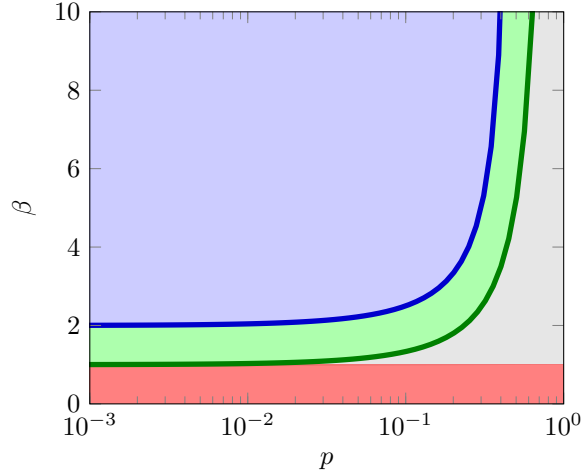
\begin{figure}
\centering 
\vspace{4mm}
\begin{tikzpicture}
% scale=0.9

    \begin{semilogxaxis}[xlabel=$p$,
        ylabel=$\beta$,
        ylabel near ticks,
        xmin=0.001,
        xmax=1,
        ymin=0,
        ymax=10,
        width=0.65\linewidth,
        axis on top]

   \addplot[name path=redrect, fill=red, fill opacity=0.5, draw=none, mark=none]
coordinates {
    (0.00001, 0)
    (0.00001, 1)
    (1, 1)
    (1, 0)
};

\addplot[name path=fout,color=green!50!black,line width=2pt,mark=none] table[x index=2,y index=3]{./curves.dat};   

\path[name path=axis] (axis cs:0.00001,20) -- (axis cs:0.1,20);
\addplot[green!80!white!100, fill opacity = 0.4] fill between[of=fout and axis];

\addplot[name path=f,color=blue!80!black,line width=2pt,mark=none] table[x index=0,y index=1]{./curves.dat};   
   
\path[name path=axis] (axis cs:0.00001,20) -- (axis cs:0.1,20);
\addplot[blue!20] fill between[of=f and axis];

\path[name path=axis2] (axis cs:0.00001,0.9) -- (axis cs:1,1) -- (axis cs:1,20) -- (axis cs:0.1,20);
\addplot[gray!20] fill between[of=fout and axis2];

\end{semilogxaxis}
\end{tikzpicture}  
\vspace{2mm}
\caption{\label{fig:regimes}
The outer bound from (\ref{eq:multiconverse}) holds in the blue region, which corresponds to $\beta > 1/(1/2-p)$.
The inner bound from Theorem~\ref{thm:multiachievability} holds above the green line, which corresponds to $\beta > -1/\log(1-\tfrac12(1-p)^2)$. 
This characterizes the capacity of the BEC shuffling channel in the blue region as $(1-q_0)(1 - \peff - 1/\beta)$.
The capacity in the red region (i.e., for $\beta < 1$) is $0$ and it is unknown in the gray region.}
% \vspace{-2mm}
\end{figure}

We conclude by noticing that, for natural choices of the sampling distribution $Q$, the effective erasure probability $\peff$ can be computed in closed-form.
For example, if $Q$ is ${\Pois}(\lambda)$, then 
\aln{
\peff = \frac{\sum_{n=1}^\infty q_n p^n}{1-q_0} = \frac{E[p^n] - q_0}{1-q_0} = 
\frac{E[e^{n \ln p}] - q_0}{1-q_0} = 
\frac{e^{- \lambda (1-p) } - e^{-\lambda}}{1-e^{-\lambda}},
}
leading to the capacity expression
\al{
C_{\text{BEC, Poisson$(\lambda)$}} & = (1-e^{-\lambda}) \left( 1 - \frac{e^{- \lambda (1-p) } - e^{-\lambda}}{1-e^{-\lambda}} - \frac{1}{\beta} \right)^+ \nonumber \\
& = \left(1-e^{-\lambda (1-p)} - \frac{1-e^{-\lambda}}{\beta}\right)^+,
\label{eq:becpoisson}
}
which holds in the blue region of Figure~\ref{fig:regimes}, and which provides an achievable region (and possibly the capacity) in the green region.

\subsection{Is performing clustering and decoding separately optimal?}
\label{sec:separation}

From a computational efficiency standpoint, it is natural to separate the clustering and decoding tasks.
Under this approach, one could use standard clustering algorithms to determine which output sequences originate from the same input sequence.
This could then be followed by a consensus step (i.e., a cluster-based error correction), and finally a decoding of the sequences back to the original message.
However, it may be suboptimal to perform the clustering and consensus tasks in a ``code-oblivious'' way; i.e., without taking advantage of the underlying codebook.

Notice that, while the coding scheme presented in Section~\ref{sec:linear} performs the clustering, consensus, and decoding tasks in sequence, the clustering task is not oblivious to the code.
This is because the scheme considers \emph{all} valid clustering partitions and, for each one, attempts to decode the message.
Therefore, in an indirect way, the code is helping with clustering.
In particular, we point out that in the green regime in Figure~\ref{fig:regimes}, where our achievability scheme holds, the number of valid clustering partitions was upper bounded by $2^{\alpha M}$ for some $\alpha > 0$.
Hence, with the aid of the codebook, the decoder is able to identify the correct clustering out of an exponentially large number of valid possibilities.

While this suggests that there may be scenarios where performing clustering and decoding separately is suboptimal, it does not provide any formal evidence.
In particular, our upper bound on the number of ways to cluster is clearly loose (but sufficient to establish capacity).

Moreover, it may be possible to cluster a large fraction of the strings correctly in a code-oblivious way, and this may be sufficient for correct decoding of the message.
In fact, as we discuss next, \citet{lenz2020achieving} proved that the capacity of the BSC shuffling-sampling channel with multi-draws can be achieved in a regime with a scheme that performs clustering and decoding separately.

% While our decoding procedure uses clustering, it does not rely on being able to correctly cluster everything.
% This is because we consider all valid clustering partitions, and check whether one returns a valid codeword.
% Hence, this provides evidence that a separation approach where we first cluster and then decode is not in general capacity optimal.

% \iscomment{this needs to be discussed more carefully}
% \iscomment{mention ``soft'' clustering}

\subsection{The BSC case} \label{sec:multidrawbsc}

Notice that, using the definition of $\peff$ in (\ref{eq:peff}), the capacity expression in Theorem~\ref{thm:multicapacity} for the BEC multi-draw channel can be rewritten as 
\al{ \label{eq:becalt}
(1-q_0)\left( E[1-p^N| N \geq 1] - 1/\beta \right),
}
and the term $E[1-p^N| N \geq 1]$ can be understood as the expected capacity of a ``multi-draw binary erasure channel,'' where the decoder observes multiple passes of the input string through a BEC, and the number of passes is given by the conditional distribution of the random variable $N$ given $N \geq 1$.

Given this, it is natural to conjecture that, for different choices of the noisy channel $p(y|x)$, the capacity of the multi-draw noisy shuffling-sampling channel is given by
\al{ \label{eq:multiconj}
(1-q_0)\left( E[C_N| N \geq 1] - 1/\beta \right),
}
or equivalently by
\al{ \label{eq:multiconj2}
\sum_{n=1}^\infty q_n \, C_n - \frac{1-q_0}{\beta}.
}
While both expressions are equivalent, (\ref{eq:multiconj}) is a more natural generalization of the capacity expressions from Chapters~\ref{ch:shuffled}~and~\ref{ch:noisy}, since in the case of $\Bernoulli(q)$ sampling, $E[C_N | N \geq 1] = C_1$ is just the capacity of the noisy channel $p(y|x)$.

The capacity expression in (\ref{eq:multiconj}) and (\ref{eq:multiconj2})  in fact holds, in addition to the BEC setting, for the BSC$(p)$ case with Poisson sampling.
More precisely, in \citet{lenz_upper_2019,lenz2020achieving}, 
assuming $q_n = e^{-\lambda} \lambda^n/n!$ for $n \geq 0$, the following result is proven.
\begin{theorem} \label{thm:multibsc}
Provided that $p < \tfrac18$ and $1-H(4p) - 2/\beta > 0$,
the capacity of the BSC shuffling-sampling channel is
\al{
C_{\text{BSC, Poisson$(\lambda)$}} 
& = (1-q_0)\left( E[C_{p,N} | N \geq 1] - 1/\beta \right) \nonumber \\ 
& = \sum_{n=1}^\infty q_n \, C_{p,n} - \frac{1-q_0}{\beta},
}
where $C_{p,n}$ is the capacity of a BSC with $n$ draws, given by
\al{  \label{eq:mitzenmacher_multi}
C_n = 1+ \sum_{k=0}^n {n \choose k} p^k (1-p)^{n-k} \log \frac{1}{1+p^{n-2k}(1-p)^{2k-n}}.
}
\end{theorem}%[\citep{lenz_upper_2019,lenz2020achieving}]
The formula (\ref{eq:mitzenmacher_multi}) for the capacity of a BSC with $n$ draws had been previously characterized by \citet{mitzenmacher2006multi}.
In addition to providing evidence that (\ref{eq:multiconj}) and  (\ref{eq:multiconj2}) may hold for more general channels, the result in \citet{lenz2020achieving} requires the introduction of several new techniques, particularly for the achievability.

Notice that the achievability based on linear codes introduced in Section~\ref{sec:linear} does not generalize in a straightforward manner beyond the erasure setting, since it requires the notion of consistency between output strings in order to build the consistency graph and identify valid clustering partitions.
The approach in \citet{lenz2020achieving} instead focuses on a $(p,\beta)$ parameter regime where a greedy clustering algorithm can be shown to  recover most of the clusters correctly with high probability.
Notice that this is different from the achievability presented in Section~\ref{sec:linear} for the BEC case, which does not formally require a code-oblivious clustering algorithm to recover most of the clusters correctly.

% still holds in regimes where correct clustering is not guaranteed at the output.

In addition, rather than dealing with solving a linear system of equations as done in Section~\ref{sec:linear} for the BEC case, in order to achieve the capacity in Theorem~\ref{thm:multibsc}, \citet{lenz2020achieving} introduce a new notion of typicality between a set of $n$ output strings and an input string. 

Notice that, as in the BEC case discussed in Section~\ref{sec:bec}, the achievability for the BSC case requires a more sophisticated scheme that goes beyond simple index-based coding.
Intuitively, if one were to use an index-based scheme, each cluster would need to be decoded independently so that the indices corresponding to each cluster could be obtained. 
However, the number of sequences per cluster varies from cluster to cluster.
Hence,
it is unclear what level of error correction needs to be used for each string, since large clusters are easier to decode than small clusters. 
We derive the rate that is achievable by index-based coding in multi-draw settings in (\ref{eq:indexachievable}) in Section~\ref{sec:general}.
% This suggests that it is not possible to achieve capacity in the multi-draw setting with the index-based coding scheme described in Section~\ref{sec:bsc_thm}.

\subsection{General Discrete Memoryless Channels}
\label{sec:merhav}

The capacity of multi-draw DNA storage channels where the noisy channel $p(y|x)$ is an arbitrary discrete memoryless channel recently received a careful treatment by~\citet{weinberger2021dna}.
The paper provides a new lower bound for general multi-draw DNA storage channels that is based on a random coding argument and holds for all values of $\beta>1$, unlike the achievability arguments for Theorems~\ref{thm:multiachievability} and \ref{thm:multibsc}, which only cover specific discrete memoryless channels. 
The paper also provides a new capacity upper bound in terms of a single-letter expression that holds for all values of $\beta$. 
The upper bound follows the general approach outlined in Section~\ref{sec:intuition} and elaborated by~\citet{lenz_upper_2019}, but is based on a more careful notion of distance between the sequences.

 \citet{weinberger2021dna} also show that there exists a $\beta_{\rm cr}$ such that, for $\beta > \beta_{\rm cr}$, the upper and lower bounds coincide, establishing the capacity of the multi-draw DNA storage channel for this regime.
 This is similar to the parameter regimes in Theorems~\ref{thm:multicapacity} and ~\ref{thm:multibsc} (which can be equivalently expressed as $\beta > \beta_{\rm cr}(p)$).
 In particular, the paper considers the case of \emph{modulo-additive} channels with input alphabet $\X = \{1,\dots,|\X|\}$ and output alphabet $\Y = \X$, and input-output relationship given by
 \aln{
 Y  = X+Z \bmod{|\X|},
 }
 where $Z$ follows some distribution over $\{1,\dots,|\X|\}$.
 The capacity of the modulo-additive multi-draw DNA storage channel is shown to be
 \al{ \label{eq:modulo_cap}
 C_{\text{mod,$Q$}} 
% & = (1-q_0)\left( E[C_{p,N} | N \geq 1] - 1/\beta \right) \nonumber \\ 
& = \sum_{n=1}^\infty q_n \, C_{n} - \frac{1-q_0}{\beta},
 }
 as long as $\beta > \beta_{\rm cr}$, where $\beta_{\rm cr}$ can be explicitly computed \citep{weinberger2021dna}, and $C_{n}$ is the capacity of the modulo-additive channel with $n$ draws.
For the special case $\X = \{0,1\}$, the modulo-additive channel reduces to the BSC, and the result in Theorem~\ref{thm:multibsc} is recovered, but with a weaker requirement on $p$ and $\beta$.
 We point out that the capacity formula in \eqref{eq:modulo_cap} can be rewritten as
\aln{
(1-q_0)\left( E[C_N | N \geq 1] - 1/\beta \right),
}
which is the general formula for all capacity expressions discussed in this monograph, providing additional evidence for the conjecture in Section~\ref{sec:general}.

\section{Connections with the trace reconstruction problem}
\label{sec:trace}

If one considers the multi-draw channel studied in this chapter but with a single input sequence ($M=1$), we obtain the problem of reconstructing an original ``seed'' string from several noisy versions of it.
This problem was originally
proposed by \citet{levenshtein1997reconstruction}.
% , levenshtein2001efficientIT}.
The special case where we observe the output of passing the seed string multiple times through a deletion channel was studied in \citet{batu2004reconstructing} under the name of trace reconstruction problem. 
This problem received considerable attention in the last few years \citep{abroshan2019coding,magner2016fundamental,mao2018models,holenstein2008trace,de2017optimal,hartung2018trace,nazarov2017trace,srinivasavaradhan2018maximum,srinivasavaradhan2019symbolwise,sabary2020reconstruction}, partially due to 
the fact that DNA sequencing technologies and  especially nanopore sequencing technologies  suffer from deletion errors \citep{mao2018models}.
In addition, the trace reconstruction problem has applications in immunogenomics, where one seeks to reconstruct the set of germline genes of an individual by analyzing antibody repertoires.
In this setting, each antibody can be viewed as a trace generated by recombining elements from the germline genes \citep{bhardwaj2020trace}.

Most of the work on the trace reconstruction problem has focused on characterizing the minimum number of traces (i.e., output sequences) needed for reconstructing the seed string correctly.
For example, for a random binary seed string of length $n$, it is known that $\exp( O(\log^{1/3} n))$ traces are sufficient \citep{hartung2018trace}, while for a worst-case binary string, 
% $\exp\left( O(n^{1/3})\right)$ traces are sufficient \citep{nazarov2017trace}.
$\exp(\tilde{O}(n^{1/5}))$ traces are sufficient~\citep{chase2021new}.
% $\exp\left( O(n^{1/3})\right)$ traces are sufficient \citep{nazarov2017trace}.

The trace reconstruction problem has a very direct connection with the problem of retrieving data stored on DNA.
As discussed in this chapter, under most proposed DNA storage prototypes, we observe a random number of copies  of each stored DNA molecule.
Hence, provided that one can correctly cluster the observed strings based on string of origin, the problem of reconstructing the sequences can be seen as several instances of the trace reconstruction problem.
Notice, however, that treating each of these instances as independent trace reconstruction problems may be suboptimal from an information-theoretic standpoint, as discussed in Section~\ref{sec:separation}.
Hence, the DNA storage setting can be more appropriately cast as a ``trace reconstruction with multiple seeds,'' as proposed and studied in \citet{bhardwaj2020trace}.

A second and arguably more important distinction between the general trace reconstruction problem (with multiple seeds) and the DNA storage setting is the fact that the trace reconstruction problem is typically concerned with the reconstruction of an arbitrary seed string (or one drawn according to a probability distribution).
On the other hand, DNA molecules in a storage system are chosen as codewords from a codebook.
As such, they can be designed to facilitate the solution of the trace reconstruction problem that the decoder faces.
This key observation led to the proposal of the 
\emph{coded trace reconstruction} problem by
\citet{cheraghchi2020coded}.
For the coded trace reconstruction with a single seed of length $n$, the authors propose marker-based code constructions, which rely on the placing of long runs of zeros throughout the seed string.
These markers can be identified by the decoder and effectively allow the problem to be converted into many small coded trace reconstruction problems.
Based on this idea, it is shown that 
one can build codes with asymptotic rate $1$ for which
the number of traces needed for correct decoding is $\exp\left(O((\log \log n)^{2/3})\right)$.
This is in stark contrast to the 
$\exp(\tilde{O}(n^{1/5}))$ 
% $\exp\left( O(n^{1/3})\right)$ 
traces needed in the case of an arbitrary seed string~\citep{chase2021new}.
% \citep{nazarov2017trace}.

We point out that while establishing the exact scaling of the number of traces needed for perfect reconstruction as a function of the seed length $n$ is an important problem, it is less relevant in the short-molecule setting that is the focus of this monograph.
Moreover, from a capacity standpoint, 
the coded trace reconstruction problem can be seen as a generalization of the standard deletion channel (where one observes the output multiple independent times instead of one). 
See the paper by~\citet{haeupler_mitzenmacher_2014} for results on the capacity of a channel where multiple copies of a sequence are passed through a deletion channel. 
As the capacity of the deletion channel is a long-standing open problem \citep{cheraghchi2019capacity}, characterizing the capacity of coded trace reconstruction with a constant number of traces should be just as difficult.
In Section~\ref{sec:deletionconj} we discuss additional open problems in the context of the deletion channel and trace reconstruction that are motivated by DNA storage.

\section{Code constructions for multi-draw channels}
\label{sec:multiadditional}

In addition to the work on the coded trace reconstruction problem discussed in Section~\ref{sec:trace}, several  works have proposed code constructions for the multi-draw settings that arise in the context of DNA storage systems.
Most of these focus on coding-theoretic questions, including the derivation of bounds on the size of codes with specific distance properties, and the development of explicit codes with efficient encoding and decoding.
% Here we list some of these.

The notion of \emph{clustering-correcting codes} was proposed by
\citet{shinkar2019clustering}.
These codes are designed so that each of the input sequences has a data field and an index field, designed to facilitate the clustering of the output sequences in the presence of noise.
Clustering is initially performed based on the index fields, but it is shown that the data fields can also be used to improve the clustering quality (and reduce the amount of redundancy needed).
Additionally, 
\citet{lenz2019codingoversets} studied the multi-draw DNA storage channel from a coding-theoretic perspective.
In particular, versions of the Gilbert-Varshamov bound and the sphere-packing bound were derived, providing insight into fundamental limits for explicit code constructions.

% \citet{rashtchian2017clustering} describe an efficient  algorithm for clustering.

\section{Proof of Lemmas} 
\label{sec:lemproof}

\begin{lemmarep}{\ref{lem:matrixG}}
Let $\matG$ be an $ML \times B$ matrix with  i.i.d.~$\Bernoulli(1/2)$ entries.
Fix any $\delta \in  (0,1)$ and a submatrix $\matG'$ formed by an arbitrary set of $(1-\delta)B$ rows of $\matG$.
Then $\matG'$ is full rank (over the finite field $\F_2$) with probability tending to $1$ as $B \to \infty$.
\end{lemmarep}

\begin{proof}
We follow the approach in the lecture notes by \citet{sayir_notes}.
In order for $\matG'$ to be full rank, the $(n+1)$th row must be chosen as a vector that is not in the span of rows $1,\dots,n$.
Note that the space spanned by $n$ linearly independent vectors in $\F_2$ has  exactly $2^n$ distinct elements.
If we assume that the first $n$ rows are linearly independent, then the probability that the $(n+1)$th row (which is a $B$-dimensional vector) is not in the span of the first $n$ rows is $1- 2^{n-B}$.
% \aln{
% 1- \frac{2^n}{2^B} = 1- 2^{n-B}
% }
By induction we see that the probability that all $(1-\delta)B$ rows are linearly independent is 
\al{
\prod_{j=1}^{(1-\delta)B} \left(1-2^{-(B-j+1)}\right) 
& = \prod_{i=\delta B + 1}^{B} \left(1-2^{-i}\right) \nonumber \\
& = \frac{\prod_{i=1}^{B} \left(1-2^{-i}\right)}{\prod_{i=1}^{\delta B} \left(1-2^{-i}\right) }. \label{eq:lemlimit}
% & = \frac{1}{\prod_{i=1}^{\delta B} \left(1-2^{-i}\right) } \prod_{i=1}^{B} \frac{2^i - 1}{2^i} \nonumber \\
}
As $B \to \infty$, both the product in the numerator and in the denominator can be verified (e.g., using numerical software) to converge to
\aln{
\prod_{i=1}^{\infty} \left(1-2^{-i}\right) = 0.28879...
}
which implies that (\ref{eq:lemlimit}) tends to $1$ as $B \to \infty$, proving the lemma.
Notice that, if $\delta = 0$, the probability does not tend to $1$ and instead tends to $0.28879$, which is in contrast to the case of a real-valued matrix with random entries from a continuous distribution, where all square submatrices will be full-rank with high probability.
\end{proof}

\begin{lemmarep}{\ref{lem:edges}}
The number of incorrect edges $Z$ satisfies
\al{
% E[Z] = M^{2-\gamma} \quad \text{and} \quad
\Pr\left(Z > M^{2-\gamma + \ep}\right) \to 0}
as $M \to \infty$, for any $\ep > 0$. 
% \iscomment{lower bound too?}
\end{lemmarep} 

\begin{proof}
Consider two output strings $\y_i$ and $\y_j$ that are generated from distinct input strings $\x_i$ and $\x_j$.
We show that $\y_i$ and $\y_j$ are consistent 
with probability $M^{-\gamma}$.

Let $\x_i[\ell]$ and $\x_j[\ell]$, for $\ell=1,\ldots,L$ be the individual symbols in the sequences.
Notice that $\x_i$ and $\x_j$ are generated by choosing $\bt_i$ and $\bt_j$ uniformly at random from $\{0,1\}^B$, and computing $\matG'\, \bt_i$ and  $\matG''\, \bt_j$, where $\matG'$ and $\matG''$ are each obtained by taking the $L$ rows from $\matG$ corresponding to the $i$th and $j$th input sequences (note that $\matG \bt_i$ has length $ML$).

First we claim that the $2L$ random variables $\x_i[\ell],\x_j[\ell]$, $\ell=1,\ldots,L$, are mutually independent Bernoulli$(1/2)$.
Treating all vectors as column vectors, we can write
\al{ \label{eq:xgt}
\begin{bmatrix}
\x_i \\ \x_j
\end{bmatrix} = 
\underbrace{
\begin{bmatrix}
\matG' & 0 \\ 0 & \matG''
\end{bmatrix}}_{H}
\underbrace{
\begin{bmatrix}
\bt_i \\ \bt_j
\end{bmatrix}}_{\tilde{\bf t}}.
}
The block diagonal matrix $H$ above has dimension $2L \times 2B$, where $B = ML (1-q_0 -\ep)(1-\peff-\ep)(1-\ep)$.
For $M$ large enough, we have $B > L$, and $H$ is full row-rank, with a null space of dimension $2B - 2L$.
Hence, for any 
$\bc \in \F_2^L$, the number of solutions $\tilde{\bt}$ to $\bc = H {\tilde{\bt}}$ is $2^{2B-2L}$ and, if ${\tilde{\bt}}$ is drawn uniformly at random from $\F_2^{2B}$,
\aln{
\Pr(H {\tilde{\bt}} = {\bc}) = \frac{2^{2B-2L}}{2^{2B}} = 2^{-2L}.
}
This implies that the column vector $\begin{bmatrix}
\x_i \\ \x_j
\end{bmatrix}$ is chosen uniformly at random from $\F_2^{2L}$.
This in turn implies that the entries of $\x_i$ and $\x_j$ are all mutually independent  i.i.d.~$\Bernoulli(1/2)$ random variables.

Given this fact, the event that $\y_i$ and $\y_j$ are consistent is the intersection of $L$ independent events
\aln{
\{ \x_i[\ell] = \x_j[\ell] \text{ or } \x_i[\ell] = \eras \text{ or } \x_j[\ell] = \eras \},
}
for $\ell=1,\ldots,L$.
Each of these events happens with probability
$1-\tfrac12(1-p)^2$,
implying that $\x_i$ and $\x_j$ are consistent with probability 
\aln{
(1-\tfrac12(1-p)^2)^L = 2^{-\gamma \log M} = M^{-\gamma}.
}
Finally, we notice that the expected number of output sequences is $M E[N_1]$, and the expected number of pairs 
of output strings 
% that are not generated from the same input string
is at most $M^2 E[N_1]^2$.
Hence, the expected number of incorrect edges satisfies
\aln{
E[Z] \leq M^2 E[N_1]^2 M^{-\gamma} = E[N_1]^2 M^{2-\gamma}.
}
Finally, using Markov's inequality, we have that 
\aln{
\Pr\left(Z > M^{2-\gamma + \ep}\right)
\leq \frac{E[Z]}{M^{2-\gamma + \ep}} \leq 
\frac{E[N_1]^2 M^{2-\gamma}}{M^{2-\gamma + \ep}}
= E[N_1]^2 M^{-\ep},
}
which tends to $0$ as $M \to \infty$ for any $\ep > 0$.
\end{proof}

\chapter{Coding and Computational Aspects}
\label{ch:coding}

\newcommand\mC{\textbf{C}}
\newcommand\mM{\textbf{M}}
\newcommand\vm{\textbf{m}}
\newcommand\vc{\textbf{c}}
%\begin{enumerate}
%    \item Section 1: inner code, outer code, Reed-Solomon, erasures, deletions
%    \item Section 2: multi-draw, clustering, alignment
%\end{enumerate}

%So far, this monograph has focused on the information theoretic aspects of DNA data storage. 

The results on the noisy shuffling-sampling channel in Chapters~\ref{ch:shuffled},~\ref{ch:noisy} and \ref{ch:multi} provide insights into optimal ways to encode data for DNA based storage (for example, the fact that an index-based scheme is optimal in the single-draw setting). 
However, when designing coding schemes for real DNA storage systems, several additional factors need to be taken into consideration.
These include the specific noise profiles of the technologies being used and the computational complexity of the algorithms employed in encoding and decoding.
In this chapter, we discuss practical coding schemes for DNA data storage, as well as their limitations and computational aspects.

The first two demonstrations of DNA storage systems~\citep{church_next-generation_2012,goldman_towards_2013} are based on indexing individual sequences. Both systems did not use modern error correcting codes to protect the information: the scheme by \citet{church_next-generation_2012} partitions the data to be stored into parts, adds an index to each part, and maps the parts to DNA sequences. 
The scheme by \citet{goldman_towards_2013} also first partitions the information into parts, but ensures that neighboring parts overlap, which introduces redundancy. Next, the scheme by \citet{goldman_towards_2013} adds an index to each part, and finally maps the parts into DNA sequences. Since no principled error-correction mechanisms are used, both systems cannot efficiently deal with errors and thus the two demonstrations of DNA storage systems relied on accurate synthesis and on sequencing with large coverage. Both systems could recover the majority of the encoded information through an accurate experimental setup, but could not recover every single bit correctly.

Starting with~\citet{grass_robust_2015}, all subsequent designs of DNA data storage systems used a combination of outer and inner error-correcting codes to protect the information against noise (for example,~\citet{erlich_dna_2016,organick_scaling_2017,blawat_forward_2016,bornholt_dna-based_2016,chandak2019improved}, and others). 
This approach, discussed in Section~\ref{sec:innerouter}, works well if the errors within the sequences are relatively few and mostly substitution errors.

If the error probabilities within sequences are large and dominated by deletions and insertions, then it is necessary to combine the information of many noisy copies of each DNA sequence to recover the information. A natural approach to do this is to first cluster the sequences, second to reconstruct a sequence close to an original sequence from each cluster,  and third to perform subsequent error correcting steps; we discuss this approach along with the associated computational costs in Section~\ref{sec:noisycoding} of this chapter.

%%%
\section{Inner-outer coding scheme}
\label{sec:innerouter}

In Section~\ref{sec:bsc_thm} we presented an index-based coding scheme based on combining an inner code and an outer code (see Figure~\ref{fig:innerouter}).
While that scheme enabled us to establish an achievable rate (and the capacity) for the BSC shuffling-sampling channel, it relied on unspecified capacity-achieving codes for the BSC and capacity-achieving codes for an erasure channel over an alphabet of arbitrary size. 
Here, we describe an inner-outer coding scheme that uses existing practical codes, and that works well if individual sequences have few errors. The encoding and decoding steps, illustrated in Figure~\ref{fig:cdscheme}, are explained next.

\paragraph{Encoding:}
The data to be stored is given as a sequence of bits.
We first map the data to $B$ information blocks, denoted by $\mM_b$, $b=1,\ldots,B$, each consisting of $m\times k_O$ symbols in a finite field. The size of the finite field depends on the code used as explained below, and the parameter $m$ depends on the desired final length of the sequences. 
The original data is a file stored in bits, and mapping the data to information blocks with elements in a finite-field simply corresponds to converting a number from one base to another. 

Each row of each information block is encoded with a code we call \emph{outer code}. This yields the outer-encoded block of codewords $\mC_b$, each consisting of $m\times n_O$ symbols, where $n_O$ is the length of the outer code. 
The outer code needs to be able to correct substitution errors and erasures; an example of such a code is a Reed-Solomon code. We discuss choices for the outer code below.

Each of the $n_O$ columns of the information block will correspond to a DNA sequence. Since the DNA sequences are unordered, we add a unique index to each of the columns of each of the information blocks. For example, the $i$th column in the $b$th information block receives the index $n_O\cdot b + i$. 
Note that the outer code can correct damaged or lost DNA sequences, as explained below. 
%\iscomment{Is there a reason why the columns become DNA sequences? Is it so that a missing DNA sequence creates only one erasure in an outer codeword? It would be good to mention this.}
%\rhcomment{}

Next, each column of an information block along with the index, denoted by $\vm_{b,1}, \ldots,\vm_{b,n_O}$, is viewed as a vector of length $k_I$ with symbols in an appropriate finite field, but not necessarily the same field as that used for the outer code. 
Each of the indexed vectors is encoded with an \emph{inner code} yielding the codewords $\vc_{b,1}, \ldots, \vc_{b,n_O}$, each consisting of $n_I > k_I$ symbols. 

Finally, each inner codeword $\vc_{b,j}$ is mapped to a string with symbols in $\{\bA,\bC,\bG,\bT\}$, corresponding to the four nucleotides. Those strings are then synthesized in DNA. 

\paragraph{Decoding:}
In the decoding stage, we start with $\N$ DNA sequences, obtained from sequencing. Those sequences are mapped to the received inner codewords $\tilde \vc_1, \tilde \vc_2, \ldots, \tilde \vc_\N$, where each codeword $\tilde \vc_j$ consists of $n_I$ symbols. 

Those inner codewords are then decoded independently one by one which yields the information vectors $\tilde \vm_1, \tilde \vm_2, \ldots, \tilde \vm_\N$. This step corrects errors within each of the sequences. 
We then sort the vectors $\tilde \vm_j$ by their index, and construct the received outer codewords $\tilde \mC_b$, which consists of $m \times n_O$ symbols. 
Note that for each column of the outer codeword $\tilde \mC_b$ there might be several candidates $\tilde \vm_j$. Out of those candidates we select the one which has the smallest number of errors, if this information is available to us from the inner decoder. 
%\iscomment{I don't understand what the number of decoding errors means here.}
If there are several candidates with the same number of decoding errors, we choose the one which occurs most frequently. If there are several that occur equally frequently, we simply choose one of the candidates at random. 

Columns of $\tilde \mC_b$ for which we do not have a candidate $\tilde \vm_i$ are marked as erasures. Columns of $\tilde \mC_b$ can be erasures, if the respective sequence has not been sequenced, and substitutions, if the candidate obtained in the previous inner-decoding step contains errors. 
Finally, the rows of $\tilde \mC_b$ are decoded individually with the outer decoder to obtain the received information blocks $\tilde \mM_b$ consisting of $m \times n_O$ many symbols. The received information blocks are then mapped back to a sequence of bits, corresponding to the recovered information.

\begin{figure}
\begin{tikzpicture}[scale=0.72]
% Information block
\draw[drop shadow,fill=white] (0,0) rectangle (2,1);
\node[] at (1, 0.5){$\mM_b$};
%\node[above] at (1,1) {$k$};
%\node[right] at (2,0.5) {$m$};

\draw[->, out=40,in =140] (2,1.4) to node[above,swap] {\small encode outer code} (3.5,1.4);

% RS codeword
\draw[drop shadow,fill=white] (3.5,0) rectangle (6,1);
%\node[above] at (4.75,1) {$n$};
%\node[right] at (6,0.5) {$m$};
\node[] at (4.75, 0.5){$\mC_b$};

\draw[->, out=40,in =140] (6,1.4) to node[above,swap] {\parbox{3cm}{\centering \small add unique index to each column}} (7.5,1.6);

% RS codeword + index
\draw[drop shadow,fill=white] (7.4,1) rectangle (10.1,1.2);
\draw[drop shadow,fill=white] (7.4,0) rectangle (10.1,1);
\node[] at (8.75, 0.5){$\vm_{b,1}..\vm_{b,n_O}$};
%\node[above] at (0.75,-5) {$n$};
%\node[right] at (2,-5.7) {$K = m+l$};

% code each column 
\draw[->, out=40,in =140] [->] (10,1.6) to node[above,swap] {\small encode columns } (11.5,1.8);
\draw[drop shadow,fill=white] (11.5,0) rectangle (11.7,1.4);
\draw[drop shadow,fill=white] (11.9,0) rectangle (12.1,1.4);
\node at (12.9,0.6) {$\cdots$};
\draw[drop shadow,fill=white] (13.6,0) rectangle (13.8,1.4);
\node[] at (12.7, -0.4){$\vc_{b,1} \vc_{b,2} \ldots \vc_{b,n_O}$};

\draw[->, out=0,in =90] [->] (14,0.7) to node[above right,swap] {\small synth.} (15.8,-0.6);
\draw[->, out=-90,in =0] [->] (15.8,-3.3) to node[below right,swap] {\small seq.} (14.2,-4.5);

\begin{scope}[xshift=2.8cm,yshift=1.1cm]
\draw[thick, drop shadow, fill=white] (13,-3) circle (1cm);
\draw[very thick, blue] (13.1,-3) to (13.2,-3.3);
\draw[very thick, blue] (13.4,-3) to (13.7,-3);
\draw[very thick, blue] (13.4,-2.7) to (13.5,-2.45);
\draw[very thick, blue] (12.4,-2.7) to (12.5,-2.45);
\draw[very thick, blue] (12.8,-2.7) to (13.1,-2.7);
\draw[very thick, blue] (12.8,-3.8) to (13.1,-3.7);
\draw[very thick, blue] (12.2,-3.4) to (12.5,-3.3);
\draw[very thick, blue] (12.9,-3.4) to (12.65,-3.2);
\draw[very thick, blue] (13,-2.45) to (12.8,-2.2);
\draw[very thick, blue] (13.5,-3.2) to (13.5,-3.5);
\draw[very thick, blue] (12.5,-3.1) to (12.5,-2.8);
\draw[very thick, blue] (12.7,-3.1) to (13.0,-3);
\draw[very thick, blue] (13.7,-3.2) to (13.9,-3);
\draw[very thick, blue] (13,-2.2) to (13.3,-2.3);
\draw[very thick, blue] (12.4,-3.5) to (12.7,-3.6);
\end{scope}

\begin{scope}[yshift=-5.2cm]

% decode inner cw
\draw[->, out=140,in =40] [->]  (11.5,1.8) to node[above,swap] { \parbox{3cm}{\centering \small decode \\ inner code}} (10,1.6);
\draw[drop shadow,fill=white] (11.5,0) rectangle (11.7,1.4);
\draw[drop shadow,fill=white] (11.9,0) rectangle (12.1,1.4);
\node at (12.9,0.6) {$\cdots$};
\draw[drop shadow,fill=white] (13.6,0) rectangle (13.8,1.4);

\node[] at (12.7, -0.4){$\tilde \vc_{1}\, \tilde \vc_{2}  \;\ldots \; \tilde \vc_{\N}$};
\draw[rounded corners, dashed] (14.3,1.6) rectangle (11.0,-.7);

\draw[drop shadow,fill=white] (7.5,0) rectangle (7.7,1.2);
\draw[drop shadow,fill=white] (7.9,0) rectangle (8.1,1.2);
\node at (8.9,0.6) {$\cdots$};
\draw[drop shadow,fill=white] (9.6,0) rectangle (9.8,1.2);
\node[] at (8.7, -0.4){$\tilde \vm_{1}\, \tilde \vm_{2}  \ldots \tilde \vm_{\N}$};
\draw[rounded corners, dashed] (7,1.4) rectangle (10.3,-.7);

% RS codeword
\draw[drop shadow,fill=white] (3.5,0) rectangle (6,1);
%\node[above] at (4.75,1) {$n$};
%\node[right] at (6,0.5) {$m$};
\draw[->, out=140,in =40]   (7.5,1.6) to node[above,swap] {\parbox{2.7cm}{\centering \small sort and remove indices}} (6,1.4);
\node[] at (4.75, 0.5){$\tilde \mC_b$};

% Information block
\draw[drop shadow,fill=white] (0,0) rectangle (2,1);
%\node[above] at (1,1) {$k$};
%\node[right] at (2,0.5) {$m$};
\draw[<-, out=40,in =140] (2,1.4) to node[above,swap] {\small decode outer code} (3.5,1.4);
\node[] at (1, 0.5){$\tilde \mM_b$};

\end{scope}
\end{tikzpicture}
\caption{\label{fig:cdscheme} 
The inner-outer coding scheme breaks information into blocks $\mM_b$, encodes each row of the blocks with an outer code that can correct erasures and substitution errors, adds a unique index to each column, and finally encodes each column with an inner code. The inner code enables correcting errors within the sequences, and the outer code recovers sequences that are not sequenced, and correct errors made by the inner decoding step.}
\end{figure}
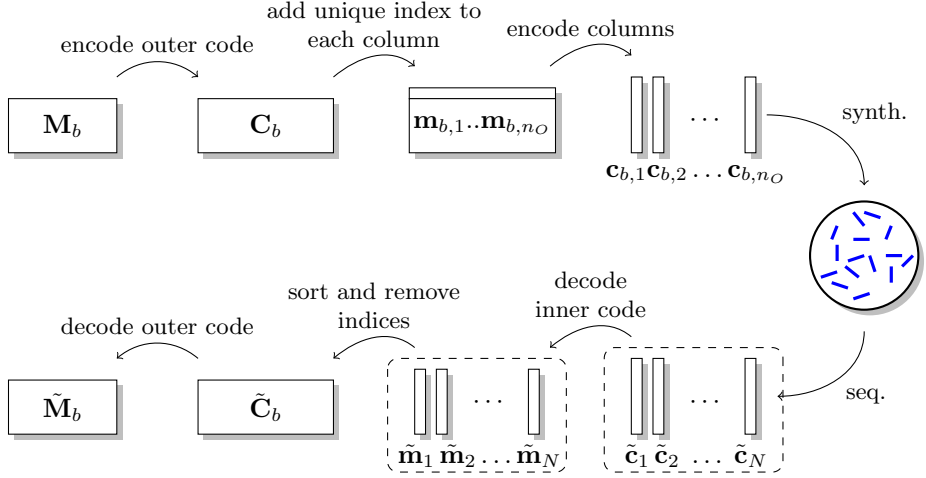

%%%
\subsubsection{Choices for the outer code}
The outer code needs to be able to correct substitution errors as well as erasures. Erasures arise from sequences that are lost, and substitution errors occur from errors within sequences that are not corrected by the inner code. 

\citet{grass_robust_2015} proposed to use Reed-Solomon codes as an outer code, and~\citet{organick_scaling_2017} also used Reed-Solomon codes as an outer code. Reed-Solomon codes are algebraic codes that can correct $e$ erasures and $s$ substitutions provided that $e + 2s \leq n-k$, where $n$ is the number of symbols of the codeword, and $k$ is the number of information symbols of the codeword. Reed-Solomon codes are optimal in that they achieve the maximal obtainable minimum distance between two codewords and thus the maximal number of erasures and substitutions that can be corrected. Therefore, Reed-Solomon codes are a natural choice for the outer code in DNA storage. However, Reed-Solomon codes of large block length (i.e., large $n$) are computationally expensive to decode, and in DNA data storage large block length are desirable.

\citet{chandak2019improved} proposed to use LDPC codes as an outer code to correct both erasures (lost sequences) as well as substitution errors in sequences. Contrary to Reed-Solomon codes, LDPC codes of long blocklength can be decoded very efficiently. Thus, designs based on LDPC can be an excellent choice for DNA storage system. 

\citet{erlich_dna_2016} proposed to use Fountain codes~(see \citet{mackay_fountain_2005} for a brief introduction to Fountain codes) as an outer code. Fountain codes are codes for channels with erasures, designed for a setup where a file is transmitted in multiple small packets, and each of those packets is either received without error or is not received. In MacKay's words, Fountain codes ``spray'' packets at the decoder without requiring knowledge of whether the packets are received. Once sufficiently many packets are received (specifically, if slightly more than $k_O$ packets are received where $k_O$ is the number of original information packets), the data can be recovered. 

The key feature of Fountain codes is its ratelessness: the transmitter does not require knowledge of how many packets have been received. 
DNA storage is not a natural setup for rateless codes, because we have to decide at storage time how much redundancy we need to add and generate the number of sequences to be stored accordingly. 
In addition, Fountain codes require $k_O + O(\sqrt{k_O} \log^2(k_O/\delta))$ symbols to recover a message with $k_O$ symbols with probability $1-\delta$. 

While Foutain codes use slightly more symbols than an optimal erasure code (such as a Reed-Solomon code which only requires $k_O$ symbols for recovering a message of length $k_O$), the extra cost is negligible once the codeword becomes long (i.e., becomes negligible in $k_O$). 
However, properly designed Fountain codes are computationally efficient to decode. For example, Luby-transform codes are an efficient choice for Fountain codes, and are decodable at cost $O(k_O \log(k_O))$ with the Luby-Transform. 
The main shortcoming of Fountain codes for DNA storage is that they cannot correct substitution errors in an obvious way, which makes them a sub-optimal choice for a DNA storage setup, as in a DNA storage channel it is difficult to guarantee that all sequences passed to the outer decoder are error free.

%%%
\subsubsection{Choices for the inner code}

Ideally, the inner code would be able to correct substitution, deletion, and insertion errors, as all those errors occur during storage. However, currently there are no short codes (say, of length 50-200) that can correct all types of errors simultaneously with little redundancy.

\citet{grass_robust_2015,organick_scaling_2017,erlich_dna_2016} used Reed-Solomon codes as an inner code. The sequences in those setups are are relatively short (60-200 base pairs long), and since short Reed-Solomon codes are computationally efficient for short codelengths they are a good choice for setups where substitution errors dominate. 

Reed-Solomon codes and other standard codes (like Reed-Muller codes) are originally not designed to correct deletions or insertions, albeit linear codes (including Reed-Solomon codes) of large length can correct some deletion and insertion errors~\citep{con_linear_2021}. Being able to correct deletions and insertions would be desirable as those errors can occur during synthesis, storage, and sequencing. 

Channels that introduce deletions and insertions as well as corresponding coding schemes are much less understood than erasure and substitution channels and codes. 
The following is known about deletions and insertions: 
\citet{mitzenmacher2006simple} established that the capacity of a deletion channel, i.e., a channel that deletes each bit with probability $p$, is lower-bounded by $(1-p)/9$ for all $p$. 
For $p\to 0$, the capacity behaves like $1-O(p\log(p))$ as established by~\cite{kanoria_optimal_2013}. 
%and  is upper-bounded by $(1-0.2p)$ \citep{cheraghchi_capacity_2018}. 
This establishes that at least in the asymptotic regime, the capacity of a deletion channel is no more than nine times lower than that of an erasure channel which erases each bit with probability $p$, and for small $p$, the capacity is similar to that of a binary symmetric channel. The exact capacity of the deletion channel is still not known to date. 
For adversarial deletions and insertions, \citet{Haeupler_Shahrasbi_2017} gave efficient insertion-deletion correcting codes for large alphabets, and~\citet{cheng_efficient_2021} built on this work by building insertion-deletions correction codes that work with smaller alphabets. See \citet{Haeupler_Shahrasbi_2021} for a recent overview on codes for deletions and insertions. 
To give a few examples, the papers by~\citet{brakensiek_efficient_2018,gabrys_codes_2019,Sima_Bruck_2021,guruswami_explicit_2021} provide codes for correcting a constant number of deletions.

%\rh{Perhaps add more references and discussion here?}

%%%
\section{Coding for noisy setups}
\label{sec:noisycoding}

The inner-outer encoding-decoding scheme does not exploit the fact that we typically have multiple noisy copies of each sequence. If the sequences are relatively error-free or only contain few substitution errors, exploiting multiple copies is not necessary. However, if the sequences contain many substitution, insertion, and deletion errors, then the inner-outer encoding-decoding scheme does not work simply because no codes exist that enable correcting sufficiently many deletion, insertion, and substitution errors in short codewords for correctly decoding sufficiently many of the sequences. 

In this section, we discuss a relatively straightforward decoding approach for very noisy sequences, which first clusters the sequences and second performs multiple alignment on the clusters in order to extract a single candidate sequence from each cluster. This candidate sequence has significantly fewer or no errors and can be passed through the previously discussed inner-outer encoding-decoding scheme to recover the information.

The idea of clustering reads before decoding and an efficient clustering scheme has been proposed  by~\citet{rashtchian2017clustering}. \citet{antkowiak_low_2020} built a DNA storage system with low-cost photolithographic synthesis which introduces large error rates in the sequences, with the clustering-multiple-alignment scheme described next. 
 
\paragraph{Introducing randomization:}
Clustering millions of DNA sequences and recovering candidate sequences from the clusters is  in general (i.e., for arbitrary sequences) intractable. 
The problem of recovering a sequence from multiple noisy copies of that sequences, each perturbed by random insertions and deletions, previously discussed in Section~\ref{sec:trace}, is known as trace reconstruction. %~\cite{holenstein_trace_2008}.

Clustering and trace reconstruction become computationally and theoretically feasible if the fragments are random. \citet{holden_subpolynomial_2018} shows that for random sequences, the trace reconstruction problem can be solved based on only $O(\log^{1/3} L)$ traces/fragments, where $L$ is the length of the sequences. 

In DNA storage, it is possible to design the sequences so that they are pseudorandom for efficient computational clustering and trace-reconstruction. Specifically, we can simply multiply the information in the sequences with a pseudorandom sequence, so that all fragments appear random. 
That guarantees that distinct molecules are far from each other, and that individual sequences look like a pseudorandom sequence. This randomization step converts a difficult clustering problem instance to an ``easy'' instance, since now the true cluster centers are almost orthogonal to each other, and it also converts a potentially difficult trace reconstruction problem into an easy one.

\paragraph{Clustering reads:}
The goal of the clustering step is to %cluster together the sequences that originate from the same original sequences, and do so 
efficiently cluster millions of short reads. 
Clustering algorithms are based on a notion of distance between DNA sequences or strings.
The natural measure of distance between two sequences---given that the perturbations are deletions, insertions, and substitutions---is the edit distance. The edit distance between two sequences is the minimal number of deletions, insertions, and substitutions required to transform one sequence into the other. The edit distance between two sequences is expensive to compute. 

We next describe an efficient method for clustering the noisy reads. 
The clustering method is based on locality sensitive hashing (LSH), and is inspired by an algorithm proposed for clustering web documents~\citep{haveliwala_scalable_2000}. LSH relies on a cheaper-to-compute proxy for the edit distance, computed using the so-called Min-Hashing (MH) method. 

To compute a proxy for the edit distance, we first split each sequence into overlapping subsequences of length $k$, called $k$-mers (also called shingles of length $k$ or $q$-grams). For example, for $k=2$, the sequence $\bA\bC\bG\bT$ becomes the set $\{\bA\bC,\bC\bG,\bG\bT\}$. 
Now, each sequence is represented by a set of $k$-mers and similarity between the two sets can be measured by the Jaccard similarity of two sets (the Jaccard similarity is defined as the size of the intersection divided by the size of the union of the sample sets). The clustering method we propose does not compute the Jaccard distance for all pairs of sequences, since this is computationally infeasible.
Instead, we use locality sensitive hashing to find similar sequences. Locality sensitive hashing first generates a signature for each sequence, with the Min-Hashing method. Those sequences are likely to be equal if two sequences are similar or the same.
Specifically, we perform the following steps:
\begin{enumerate}
\item Extract pairs of similar sequences: We first generate sets of pairs of sequences by finding all pairs of sequences that have the same signature. 
%generating locality sensitive hashing signatures using the Min-Hashing method. \iscomment{I don't understand what ``generating pairs'' means}
\item Filter pairs: As a next step, we go through all the pairs and after performing a local pairwise alignment on each pair, we drop a pair from the set if the number of matched characters falls below a threshold. %\iscomment{shouldn't this be based on the number of matching min-hashes?}
\item Generate clusters from similar pairs: 
Finally, we generate clusters from the pairs based on sorting the pairs. 
\end{enumerate}
%We will test this algorithms numerically, compare it with existing clustering approaches (in particular that from~\cite{rashtchian_clustering_2017}), and analyze it theoretically for a model where the clusters are perturbed sequences generated from randomly chosen cluster centers. Note that this is a good model for our data since the sequences are randomized.
\citet{organick_scaling_2017} used a similar clustering-based approach as described above to recover data from nanopore reads, which have large insertion and deletion and error rates. 
\citet{antkowiak_low_2020} used the scheme described above to recover data in a system that uses noisy synthesis.

\paragraph{Reconstructing a single sequence from a set of noisy copies:}
After we produced clusters consisting of noisy copies of an original sequence, our goal is to reconstruct the original sequence.
This is known as a trace reconstruction problem.

One approach to reconstruct a sequence is to first align the sequences within a cluster and second perform majority voting to extract a single sequence from the cluster. For aligning the clusters, a number of off-the-shelf algorithms such as the MUSCLE algorithm~\citep{edgar_muscle_2004} are available.
%\iscomment{citation for MUSCLE?}. 

Recent works have proposed efficient codes  trace reconstruction, i.e., reconstruction of a sequence from multiple noisy copies. 
Specifically, \citet{cheraghchi2020coded} proposed marker-based code constructions with logarithmic complexity in the number of traces, and \citet{srinivasavaradhan_trellis_2021} proposed an efficient reconstruction algorithm with complexity that is linear in the number of traces. 
\citet{lenz_concatenated_2021} proposed a concatenated coding scheme, and~\citet{Gabrys_Yaakobi_2018,chrisnata2020optimal} provide codes for correcting deletions from muliple noisy copies. Finally, \citet{Sabary_Yaakobi_Yucovich_2020} studies the correction of two insertions and deletions again from multiple reads.

\paragraph{Alternative approaches:}
\citet{antkowiak_low_2020} used the clustering algorithm described above and trace reconstruction by alignment followed by majority voting to store data reliably in a very noisy setup where the synthesis introduces large insertion and deletion error rates. While intuitive, it is unlikely that the clustering approach is information-theoretically close to optimal, as discussed in Section~\ref{sec:separation}.
% \iscomment{Maybe we should say ``it is unlikely''}

\chapter{Extensions and Open Problems}
\label{ch:discussion}

In this monograph we took steps towards understanding the fundamental limits of DNA-based storage systems. 
Our main focus were noisy shuffling-sampling channels, which are a rich class of models for DNA storage systems that capture the fact that molecules are stored in an unordered fashion, are usually short, and are corrupted by individual base errors. 
By considering special cases of noisy shuffling-sampling channels, we built an increasingly more realistic sequence of models and characterized their capacity. 
We presented and discussed several recently proposed techniques, both in terms of achievability and converse, that enabled us to characterize the information-theoretic limits of noisy shuffling-sampling channels and gain insights for the design of practical systems that perform close to what is fundamentally possible.

% \iscomment{need to add generalization to different alphabet sizes}

\section{Generalizing capacity results} 
\label{sec:general}

%While capacity results have been obtained in several different settings, several generalizations are possible and several related problems remain open.
Several generalizations of the capacity expressions in this monograph are possible. 
Note that the expression
\al{ \label{eq:capconj2}
(1-q_0)\left( E[C_N| N \geq 1] - 1/\beta \right)^+
} discussed in Section~\ref{sec:multidrawbsc}, recovers the capacity expressions of all shuffling-sampling channels we studied as special cases. Here, $C_N$ is the capacity of the respective noisy channel $p(y|x)$ with $N$ draws. 
For the BSC shuffling-sampling channel studied in Section~\ref{sec:bsc_thm}, for example, this follows from the fact that, for $\Bernoulli(q)$ sampling, $E[C_N|N\geq 1] = C_1$ is simply the capacity of the BSC that corrupts each individual sequence.
Therefore, Theorems~\ref{thm:noisefree},~\ref{thm:bsc},~\ref{thm:bec},~\ref{thm:multicapacity}, and~\ref{thm:multibsc} and the generalization by \citet{weinberger2021dna} discussed in Section~\ref{sec:merhav} can all be stated using the capacity expression in (\ref{eq:capconj2}).

\paragraph{Different alphabet sizes:}
While a quaternary alphabet $\Sigma = \{\bA,\bC,\bG,\bT\}$ is the relevant one for DNA-based storage systems, all results in this monograph were stated for binary alphabets for convenience. 
Here we discuss their generalization to an arbitrary alphabet $\Sigma$.

% The results presented here extend to quaternary alphabets as follows.
% All results in this monograph were stated for the case of a binary alphabet in order to simplify their exposition.
% Since the case of a quaternary alphabet $\Sigma = \{\bA,\bC,\bG,\bT\}$ is the relevant one for DNA-based storage systems, it would be of interest to extend the results to general alphabets.
% \rh{Can we cut the second part of the sentence above? Its a bit missleading, because we start with saying that it would be of interest to extend the results, but then we essentially explain that it's trivial.
% Is it accurate to say:
% ``While a a quaternary alphabet $\Sigma = \{\bA,\bC,\bG,\bT\}$ is the relevant one for DNA-based storage systems, we stated our results for binary alphabets for convenience. The results presented here extend to quaternary alphabets as follows. ''
% }

It is straightforward to see that the index-based scheme that achieves capacity in the error-free setting of Chapter~\ref{ch:shuffled}
% ~and~\ref{ch:noisy}
% Theorems~\ref{thm:noisefree},~\ref{thm:bsc}~and~\ref{thm:bec} 
can be extended to a general alphabet $\Sigma$.
% and a general noisy channel $p(y|x)$ with $\Sigma$ as input alphabet.
In that case, $\log_{|\Sigma|}M$ symbols from each length-$L$ sequence are needed for a unique index and we can store 
\aln{
\log |\Sigma| \left( L - \log_{|\Sigma|}M \right) = L (\log |\Sigma| - 1/\beta)
}
data bits per sequence.
Using those bits for an outer code that handles the $q_0$ fraction of missing sequences (as done in Section~\ref{sec:noisefree}), we achieve a rate of $(1-q_0)(\log |\Sigma| - 1/\beta)$.

In the case of noisy shuffling-sampling channels, the index-based approach can be similarly generalized.
Specifically, if $C_{\rm noisy}$ is the capacity of a channel $p(y|x)$ with $\Sigma$ as input alphabet, then the rate 
\al{ \label{eq:capgeneralalphabet}
(1-q_0)(C_{\rm noisy} - 1/\beta)^+
}
is achievable in the case of $\Bernoulli(1-q_0)$ sampling.
Notice that $|\Sigma|$ does not appear in \eref{eq:capgeneralalphabet} as it is captured by $C_{\rm noisy}$, which can now be larger than $1$.
Also notice that, in this case the capacity can be positive even for values of $\beta < 1$.

For natural extensions of the BEC and BSC to larger alphabets, the converse argument described in Section~\ref{sec:converse} can be extended in a natural way, establishing \eref{eq:capgeneralalphabet} as the capacity, at least for a certain regime of channel parameters.

\paragraph{A unified capacity expression:}
Since \eref{eq:capconj2} also holds for general alphabet sizes, it is natural to state the following conjecture.

\begin{conjecture} \label{conj:capacity}
The capacity of a noisy shuffling-sampling channel (illustrated in Figure~\ref{fig:generalchannel}) with sampling distribution $N \sim (q_0,q_1,\dots)$ and discrete memoryless noisy channel $p(y|x)$ is given by
\aln{ 
(1-q_0)\left( E[C_N| N \geq 1] - 1/\beta \right)^+,
}
where $C_n$ is the capacity of the noisy channel $p(y|x)$ with $n$ draws.
\end{conjecture}

We point out that, even in the special cases of the noisy shuffling-sampling channel where the capacity expression in Conjecture~\ref{conj:capacity} holds, the result was only shown in specific parameter regimes.
Hence, even in the case of the BEC and BSC, the conjecture has not been fully verified.

\paragraph{Optimality of index-based coding and independent decoding:}
As we discussed in Section~\ref{sec:index}, for single-draw channels where each sequence is either observed at the output once (with probability $1-q$) or not observed at all (with probability $q$), an index-based coding scheme (that adds a unique index to each sequence) achieves rate
\al{ \label{eq:indexconj}
(1-q_0) (C_{\rm noisy} - 1/\beta).
}
The decoding is straightforward: each output sequence is independently decoded, and the indices are then used to order the sequences and recover the message. The rate above is then achieved by adding an outer code on top of the index-based coding to handle missing output sequences and the vanishingly small fraction of output sequences that are decoded in error, as discussed in \citet{DNAStorageIT}.

In the multi-draw setting, however, a naive use of index-based coding is suboptimal.
The reason for this is that, since the number of sequences in each cluster is random, the indices must be encoded with enough redundancy to be decodable even in clusters of size one.
Therefore, only the data bits take advantage of the multi-draw setting.
More concretely, in each length-$L$ sequence, we need
$(\log M)/ C_1$ bits for the index, where $C_1$ is the capacity of the noisy channel (with one draw).
We are left with $M(L - (\log M)/C_1)$ bits for data. 
The data bits can take advantage of the multi-draw nature of the channel, and can be encoded with rate
\aln{
E[C_N] = (1-q_0) E[C_N| N \geq 1] + q_0 \cdot 0 =
(1-q_0)E[C_N| N \geq 1],
}
leading to an achievable rate of 
\al{ \label{eq:indexachievable}
& \frac{M(L - (\log M)/C_1) (1-q_0)E[C_N| N \geq 1]} {ML} \nonumber \\
& \quad = (1-q_0) \left( E[C_N | N \geq 1] - \frac{\gamma}{\beta} \right),
}
where $\gamma = E[C_N | N \geq 1]/C_1$.
Since $C_i > C_1$ for $i > 1$ for any reasonable channel, $\gamma > 1$, showing that the index-based scheme is suboptimal.
Whenever $\beta \gg \gamma$, simple index-based schemes perform close to optimal.

Despite its suboptimality in the multi-draw setting, index-based coding schemes are still interesting in practice for their simplicity.
Moreover, based on the results of this monograph, we conjecture that they are still optimal in the single-draw setting.
\begin{conjecture} \label{conj:index}
For noisy shuffling-channels with $\Bernoulli(1-q)$ sampling, index-based coding is capacity-optimal.
\end{conjecture}

Notice that, if Conjecture~\ref{conj:capacity} holds, the fact that index-based schemes achieve the rate in (\ref{eq:indexconj}) would imply that Conjecture~\ref{conj:index} holds as well.

\section{Insertions, deletions, and trace reconstruction}
\label{sec:deletionconj}

In all settings studied in this monograph, the noisy channel $p(y|x)$ is assumed to be a discrete memoryless channel.
Single-base substitutions are 
% While the model captures (moderate) substitution errors which are 
the prevalent error source of most current DNA storage systems~\citep{heckel_channel_2019}, which rely on \emph{low-error} synthesis and sequencing technologies that are relatively expensive and limited in speed.
A key idea towards developing the next-generation of DNA storage systems is to employ \emph{high-error}, but cheaper and faster synthesis and sequencing technologies such as light-directed maskless synthesis and nanopore sequencing. 
Such systems induce a significant amount of insertion and deletion errors. 
Thus, an important area of further investigation is to understand the capacity of channels that introduce deletions and insertions as well.

From the point of view of characterizing the capacity, this poses significant challenges, since the capacity of (non-memoryless) channels with insertions and deletions is a long-standing open problem \citep{cheraghchi2019sharp}.
Nevertheless, it may be possible to establish the impact of adding shuffling and sampling on the capacity, without characterizing the capacity of the noisy channel.
Given the results in Chapter~\ref{ch:noisy}, it is reasonable to conjecture the following.

\begin{conjecture} \label{conj:indel}
Consider a noisy shuffling channel where the input sequences are shuffled and independently passed through a insertion-deletion channel.
The capacity of this channel is
\aln{
(C_{\rm indel} - 1/\beta)^+,
}
where $C_{\rm indel}$ is the (unknown) capacity of the insertion-deletion channel.
\end{conjecture}

Notice that, as discussed in Section~\ref{sec:general}, index-based schemes achieve the rate in Conjecture~\ref{conj:indel}, and the challenge is in proving a matching outer bound.
We point out that, in deriving the outer bounds in the BEC and BSC cases, it was important to know the capacity of $C_{\BEC}$ and $C_{\BSC}$ explicitly (see Section~\ref{sec:generalchannels} for a discussion).
Therefore, it is unlikely that the tools presented in this monograph can be used to establish Conjecture~\ref{conj:indel}.

Another set of interesting open questions arises from studying the trace reconstruction problem from a capacity standpoint.
Most of the work on trace reconstruction focuses on characterizing the number of traces needed for perfect reconstruction, as discussed in Section~\ref{sec:trace}.
The recent work on coded trace reconstruction \citep{cheraghchi2020coded} also studies the question of the amount of redundancy that needs to be added to the sequences in order to allow perfect reconstruction.
However, the focus of this work is on the regime where the number of traces is growing, which makes the channel capacity be $1$.

From a capacity standpoint, it would be interesting to study the regime where the number of traces $T$ is constant.
Since the case $T=1$ reduces to a standard deletion channel, whose capacity is unknown, it is unlikely that the case $T > 1$ is any easier.
However, it may still be possible to understand the impact of $T$ on the capacity.
For example, in the case of a $\BEC(p)$ with $T$ ``traces'', it is straightforward to see that the capacity is $1-p^T$, as one can generate a consensus sequence with effective erasure rate $p^T$.
In order to understand the impact of the number of traces in the capacity of the deletion channel, one could generalize existing bounds for the case of general $T$.
Since several known bounds on the capacity of the deletion channel have the form $\alpha - \beta \, p$ for constants $\alpha,\beta >0$, it is reasonable to conjecture that bounds of the form $\alpha - \beta \,  p^T$ can be derived for the $T$-trace setting.

%%%

% \section{Encoding information in concentrations} 
\section{Storing data on very short molecules}
\label{sec:short}
% While several prototypes of DNA storage systems have been implemented in recent years, the cost of storing data on DNA with current technologies is still too high for it to be adopted in most practical settings. 
% In particular, 
One of the challenges in making DNA storage practically viable is 
% drawbacks of current DNA storage system prototypes is
that the cost of synthesizing relatively long DNA molecules with few errors is very expensive. 
For example, using the architecture proposed in \citet{erlich_dna_2016}, where the length of the synthesized molecules 
% (referred to therein as \emph{oligos}) 
is around 150 nucleotides, 
the estimated cost of synthesizing $1$GB of data 
% using their architecture 
was $\$3.27$ million \citep[Supplementary Material]{erlich_dna_2016}.

Synthesis costs are significantly lower if one focuses on short molecules and admits higher error rates \citep{kosuri2014large,caruthers2011brief}.
Hence, it is of interest to study the fundamental limits of DNA-based storage systems (and noisy shuffling-sampling channels) in the ``very short'' sequence regime.
When the synthesized molecules are very short, placing a unique index may use up a significant fraction of the synthesized DNA.
Moreover, in this regime it may make sense for the multi-set of stored DNA molecules to contain multiple copies of the same sequence and, in principle, one can 
% The input to the storage channel will then likely contain replicates of each sequence.
% One can then 
utilize the frequencies of different molecules in the pool as a way to encode information.

% \subsection{Revisiting the sampling-shuffling channel}
\paragraph{Histograms as codewords and the Poisson channel:}
Consider once again the error-free sampling-shuffling channel discussed in Section~\ref{sec:noisefree}.
Due to the shuffling nature of the channel, the order of the sequences $\x_i$ in a given codeword $[\x_1,\ldots,\x_M]$ does not matter, and each codeword can be equivalently represented as a histogram of size $2^L$, as illustrated in Figure~\ref{fig:histograms}.
\begin{figure}
\center
\includegraphics[width=0.99\linewidth]{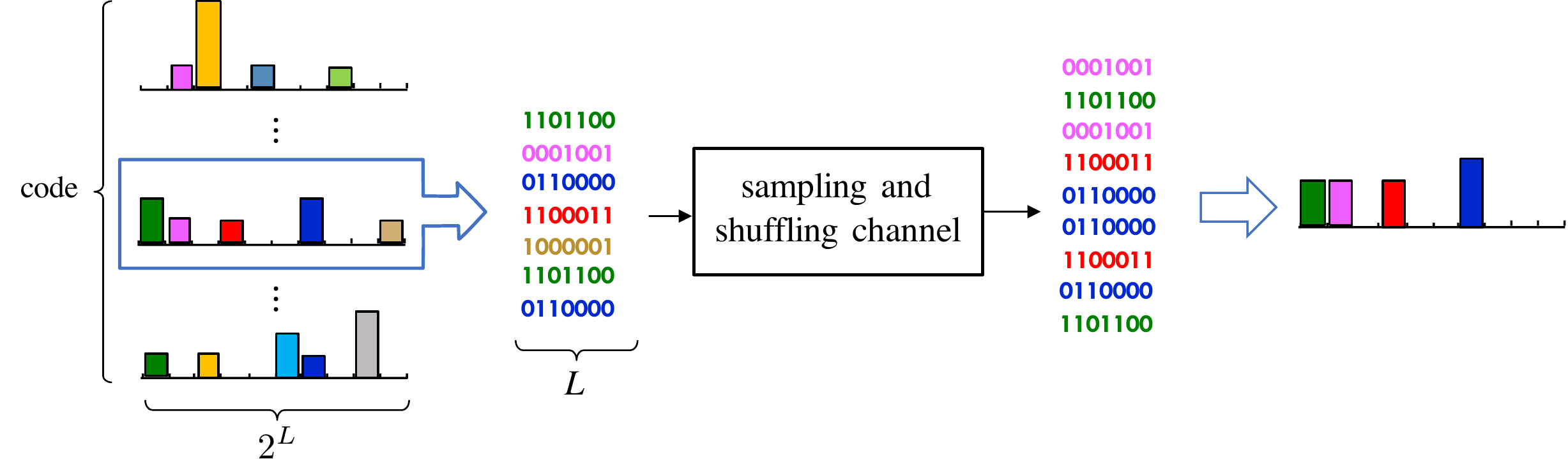}
\caption{
\label{fig:histograms}
When several copies of each molecule are present in the multi-set of input sequences, it is reasonable to view the end-to-end channel as operating on histograms of size $2^L$.
The sampling-shuffling channel modifies each entry of the histogram according to the sampling distribution $Q$.
}
\end{figure}
The effect of the sampling-shuffling channel can then be seen as simply changing the frequency of each of the molecules in the multiset; i.e., corrupting each entry in the histogram independently.

% Consider an input codeword represented as a histogram $[t[i$

To formalize this representation, we can redefine the input to the channel to be a vector $[x[t]: t \in \{0,1\}^L ]$ with entries in $\Z_+$ indexed by binary strings of length $L$ (i.e., an input histogram, as illustrated in Figure~\ref{fig:histograms}).
If we assume that the sampling distribution $Q$ is $\Pois(\lambda)$, then the channel can be described as
\al{
x[t] \in \Z_+ \quad \longrightarrow \quad  y[t] \sim \Pois(\lambda x[t]),
\label{eq:poisson}
}
for $t \in \{0,1\}^L$.
% In the error-free case, the DNA storage channel can be seen as channel over the $|\Sigma^L|$ distinct strings that can be written into the channel.
% When viewed this way, the input to the channel is 
% the number of copies of a given string in $t \in \Sigma^L$, and the output is the number of such strings that are sampled.
% Since the number of times a given molecule is sampled can be well approximated by a Poisson random variable with mean $c = N/M$, the channel can then be described (or in fact defined) as
% \al{
% x[t] \in \mathbf{Z}^+ \quad \longrightarrow \quad  y[t] \sim \Pois(cx[z]),
% }
% for $t \in \Sigma^L$.
Since we are only allowed to write $M$ molecules, 
% and if we assume the binary case for simplicity, 
we have the constraint 
\al{
\sum_{t\in \{0,1\}^L} x[t] \leq M \,  \Leftrightarrow \, \frac{1}{2^L}\sum_{t\in \{0,1\}^L} x[t] \leq M 2^{-L} = 
% M^{1-\beta}.
2^{L(1/\beta-1)}. 
\label{eq:avgpower}
}
This channel can be seen as a \emph{discrete-time Poisson channel}, previously studied in the context of optical communications \citep{shamai1993bounds,verdu1999poisson,lapidoth2011discrete,lapidoth2008capacity}.
However, unlike previously studied Poisson channels, the input is constrained to be integer-valued.
Moreover, the \emph{average power constraint} in \eref{eq:avgpower} changes with the block length $2^L$.
This channel provides us with a good model to study the shuffling channel in the ``very short'' molecule regime, as we discuss next.

% \paragraph{Storing data on very short molecules:}
\paragraph{Encoding information in concentrations:}
% \rhcomment{What about calling the subsection 7.3 ``storing data on very short molecules'', and this paragraph ``encoding information in concentrations''?}
% Throughout this monograph, we focused on the regime $\len = \beta \log \M$.
% This has the effect of forcing the molecules to be ``short'' with respect to the total amount of data being written and, as discussed in Section~\ref{sec:shuffling}, this is the regime where the capacity is nontrivial.
Notice that, for all settings considered, when $\beta < 1$, the capacity is zero.
Intuitively, this is because there are only $2^L=M^\beta < M$ distinct binary strings of length $L$.
Hence, in this regime, one is forced to use the same string many times in the input, and the capacity is zero. % 
% 
% 
% $\beta \leq 1$.
% Notice that, when $\beta \leq 1$, Theorem~\ref{thm:noiseless} asserts that the capacity is zero.
% This means that, as $M \to \infty$, the number of bits per nucleotide that can be reliably stored is asymptotically zero.
% However, for finite (and possibly large) values of $M$, we are still interested in storage rates that can be achieved by encoding information into the molecule frequencies.
% In our proposed research, we will use the abstraction of the Poisson channel to study this problem.
% 
% 
% 
% In this section we provide a brief discussion on very short molecules, i.e., molecules of length $\len = \beta \log \M$, with $\beta < 1$. In this regime, no positive rate can be achieved (as shown by Theorem~\ref{thm:noisefree}). %and there is not enough space in each molecule to append an index.
However, motivated by the fact that 
replicating DNA via PCR is a relatively common and inexpensive biotechnology technique \citep{wong2015rapid,pabinger2014survey,aird2011analyzing}, 
% it is in general much easier to synthesize very short sequences of DNA than longer ones, 
it is interesting to ask whether it is possible to encode information in the frequencies, or concentrations, of very short molecules.
% From a practical standpoint, this may be highly desirable since
% replicating DNA via Polymerase Chain Reaction is a relatively common and inexpensive biotechnology technique \citep{wong2015rapid,pabinger2014survey,aird2011analyzing}.
% However, PCR is imprecise and specifying the exact amplification factor is difficult in practice, which is captured by the Poisson channel model in (\ref{eq:poisson}) and (\ref{eq:avgpower}).

To gain insight into how to optimally encode information in frequencies one can draw from the literature on the Poisson channel.
% look at the literature
% study this Poisson channel model, 
% will study the discrete-time Poisson channel described in \eref{eq:poisson} and \eref{eq:avgpower}.
% We will do so by establishing a connection with the standard discrete-time Poisson channel studied in \cite{lapidoth2008capacity,lapidoth2011discrete}.
\citet{lapidoth2008capacity} provide the following  asymptotic characterization of the capacity of a Poisson channel.
\begin{theorem} % [\hspace{-0.1mm}\cite{lapidoth2008capacity}]
\label{thm:lapidoth}
Consider the discrete-time Poisson channel
\aln{
x[t] \geq 0 \quad \longrightarrow \quad  y[t] \sim \Pois(x[t]),
}
with average power constraint $\Ex\left[x[t]\right] \leq P$. The capacity $C(P)$ satisfies 
\aln{
\lim_{P \to \infty} \left( C(P) - \frac12 \log P \right) = 0.
}
\end{theorem}
Notice that the average power power constraint in \eref{eq:avgpower} is $P = 2^{L(1/\beta-1)}$. 
When $\beta < 1$, we have $P \to \infty$, and the approximation $C(P) \approx \frac{1}{2}\log(P)$ suggested by Theorem~\ref{thm:lapidoth} should hold.
Since for \eref{eq:poisson} we have $2^L$ channel uses,  Theorem~\ref{thm:lapidoth} suggests that the maximum achievable rate in bits per nucleotide should be roughly 
\aln{
% \frac12 \log 2^{L(1/\beta-1)} = \frac{1-\beta}{2\beta}
\frac{2^\len}{\len \M} \frac{1}{2} \log(P)
= \frac{\M^\beta L (1-\beta)}{2 M L \beta } 
= \frac{1 - \beta}{2\beta} \M^{\beta - 1}.
}
This approximation suggests that in the short-molecule regime $\beta < 1$, the number of bits per nucleotide that can be reliably stored scales as $M^{\beta - 1}$.
While this storage rate goes to zero, it suggests that the number of data bits stored can grow as $M^\beta L$, motivating the following conjecture.

% Hence, by defining the short-molecule storage rate as $\tilde R = R M^{1-\beta}$ and the short-molecule capacity $\tilde C$ as the maximum achievable $\tilde R$, we have the following conjecture.
 
\begin{conjecture} \label{conj:poissonconj}
Consider an error-free shuffling-sampling channel with $\beta \in (0,1)$.
If $\{\C_M\}$ with $M \to \infty$ is a sequence of codes with vanishing error probabilities, then
\al{ \label{eq:poissonconj}
\limsup_{M\to \infty} \frac{\log |\C_M|}{M^\beta L} \leq \frac{1-\beta}{2\beta}. % M^{\beta - 1}.
}
Moreover, there exist codes $\{\C_M\}$ that satisfy (\ref{eq:poissonconj}) with equality.
\end{conjecture}

Formally establishing this result will require understanding the non-standard Poisson channel in \eref{eq:poisson}, where the input is constrained to be integer-valued and the average power constraint changes with the block length.
% In addition to formally establishing this conjecture, 
% we will analyze its consequences for the design of coding schemes for DNA-based storage that encode information in molecule frequencies.

Once this error-free setting is established, one can seek generalizations to noisy cases, similar to what was done in Chapter~\ref{ch:noisy}, but in the short-molecule regime.
Notice that the noise, say of a BSC, causes sequences from one entry of the histogram to jump to a different entry of the histogram, and the channel is no longer memoryless, further complicating the setting.

\section{New technologies and paradigms}
\label{sec:newtech}

DNA-based storage is an emerging idea and the system-level details depend on the current state of the art in terms of DNA synthesis, storage, and sequencing.
As these technologies evolve, one may need to extend or adapt the theoretical framework proposed in this monograph to different settings.

\paragraph{Storing data on different macromolecules:}
In addition to DNA, synthetic polymers have been considered as a medium for data storage \citep{pattabiraman2019reconstruction,gabrys2020mass}.
In this setting, one can encode binary data using polyphosphodiesters in such a way that the
two bits 0 and 1 are represented by molecules of different
masses that are stitched together into a long string.
The data retrieval is performed using tandem mass spectrometry.
At a high level, many copies of the long polymer are broken down to pieces of random sizes, and the tandem
mass spectrometer provides estimates of the
\emph{masses} of the fragments.
Assuming that the masses of the 0 and 1 molecules are not multiples of each other, from the mass of each fragment one can obtain its \emph{composition}; i.e., the number of zeros and ones in it. 

Based on this setting, \citet{acharya2015string}
introduced the problem of binary string reconstruction from its \emph{substring composition multiset}.
For example, the substring composition multiset of the string 1001 is  $\{1,0,0,1,1^10^1,0^2,1^10^1,1^10^2,1^10^2\}$.
Follow-up work considered the case where the composition multiset may be subject to errors and the case where one observes only prefixes and suffixes of the string \citep{gabrys2020mass,pattabiraman2019reconstruction}.

% These works mostly focused on settings where the entire composition multiset is recovered 
From the shuffling-sampling channel standpoint taken in this monograph, an interesting research direction would be to consider a setting where the mass spectrometer randomly picks $N$ substrings from the input string and reports their sequence composition.
Unlike in the existing literature, the goal would be to study the capacity of the resulting channel (as a function of the sampling distribution and potential noisy channel).

\paragraph{New synthesis technique:}
New DNA synthesis and sequencing techniques also lead to new questions on the fundamental limits of the respective systems. 
A concrete example is DNA-based storage via combinatorial assembly commercialized by the startup CATALOG and described in~\citet{roquet_combinatorial_2021}. 

Most existing storage systems rely on array-based synthesis where DNA sequences are synthesized in parallel nucleotide-by-nucleotide. \citet{roquet_combinatorial_2021} proposed a combinatorial assembly approach, which generates DNA sequences by assembling pre-synthesized shorter sequences called components. A sequence consists of assembling $T$ components. Each component is a DNA sequence, consisting of a beginning, center, and end part. For each of the $T$ positions of the sequence, $Q$ candidate sequences exists. For example, say a sequence consists of $T=3$ components, and for each position, there are $Q=2$ candidate components, yielding the components $x_{1,a},x_{1,b},x_{2,a},x_{2,b},x_{3,a},$ and $x_{3,b}$. The beginning and ends of the components are chosen to be complementary such that when, for example, the three components $x_{1,a},x_{2,b},x_{3,a}$ are mixed together, they form the sequence $[x_{1,a},x_{2,b},x_{3,a}]$. 

This approach has the advantage that many sequences can be generated
relatively fast, but at the cost of generating a pool of long sequences where each sequence only contains little information. 
For example, \citet{roquet_combinatorial_2021} stored 26~kB with this technology on 97920 sequences, which means each sequence only contains about 2 bits. In contrast, the aforementioned systems based on array-based synthesis store more than 100 bits per sequence, which means that, to read the same amount of information, we have to read 50 times less sequences with the traditional array-based storage systems. However, characterizing the capacity of combinatorial assembly and other new technologies is an interesting problem with the potential of improving DNA storage systems based on them.
% those technologies further.

%%% 
\chapter*{Acknowledgements}

The work of Ilan Shomorony was supported in part by the National Science Foundation (NSF) under grant CCF-2007597. 
Reinhard Heckel would like to thank Robert Grass for countless helpful discussions on modeling aspects of DNA storage channels and DNA storage in general. Reinhard Heckel has received partial funding for this project from the Institute of Advanced Studies at the Technical University of Munich and from the European Research Council (ERC) under the European Union?s Horizon 2020 research and innovation programme, Grant agreement No. 964995.

% end of main matter

% \begin{acknowledgements}
% The research of IS was supported in part by NSF grant CCF-2007597. 
% IS and RH thank Kannan Ramchandran and David Tse for helpful discussions.
% RH would like to thank Robert Grass for helpful discussions on modeling aspects of DNA storage channels.
% \end{acknowledgements}

\printbibliography

\end{document}